\documentclass[11pt,letterpaper]{article}
\usepackage[utf8x]{inputenc}
\usepackage{quoting}
\quotingsetup{vskip=0pt,leftmargin=20pt,rightmargin=20pt}

\def\showauthornotes{1}
\def\showtableofcontents{1}
\def\showkeys{0}
\def\showdraftbox{0}
\def\usemicrotype{1}
\def\showfixme{0}

\usepackage{etex}

\usepackage[l2tabu, orthodox]{nag}

\usepackage{xspace,enumerate}

\usepackage[dvipsnames]{xcolor}

\usepackage[T1]{fontenc}
\usepackage[full]{textcomp}

\usepackage[american]{babel}

\usepackage{mathtools}

\usepackage{amsthm}

\newtheorem{theorem}{Theorem}[section]
\newtheorem*{theorem*}{Theorem}

\newtheorem*{proposition*}{Proposition}
\newtheorem{lemma}[theorem]{Lemma}
\newtheorem*{lemma*}{Lemma}

\newtheorem*{conjecture*}{Conjecture}
\newtheorem{fact}[theorem]{Fact}
\newtheorem*{fact*}{Fact}

\newtheorem*{hypothesis*}{Hypothesis}

\theoremstyle{definition}
\newtheorem{definition}[theorem]{Definition}

\newtheorem{example}[theorem]{Example}

\theoremstyle{remark}
\newtheorem{claim}[theorem]{Claim}
\newtheorem*{claim*}{Claim}
\newtheorem{remark}[theorem]{Remark}
\newtheorem*{remark*}{Remark}
\newtheorem{observation}[theorem]{Observation}
\newtheorem*{observation*}{Observation}

\usepackage[
letterpaper,
top=0.7in,
bottom=0.9in,
left=1in,
right=1in]{geometry}

\usepackage[varg]{pxfonts}
\usepackage{textcomp}
\usepackage[scr=rsfso]{mathalfa}
\usepackage{bm} 
\let\mathbb\varmathbb

\ifnum\showkeys=1
\usepackage[color]{showkeys}
\fi

\usepackage{hyperref}
\hypersetup{
pagebackref,
colorlinks=true,
urlcolor=blue,
linkcolor=blue,
citecolor=OliveGreen,
}
\usepackage{prettyref}

\newcommand{\savehyperref}[2]{\texorpdfstring{\hyperref[#1]{#2}}{#2}}

\newrefformat{eq}{\savehyperref{#1}{\textup{(\ref*{#1})}}}
\newrefformat{lem}{\savehyperref{#1}{Lemma~\ref*{#1}}}
\newrefformat{def}{\savehyperref{#1}{Definition~\ref*{#1}}}
\newrefformat{thm}{\savehyperref{#1}{Theorem~\ref*{#1}}}
\newrefformat{cor}{\savehyperref{#1}{Corollary~\ref*{#1}}}
\newrefformat{cha}{\savehyperref{#1}{Chapter~\ref*{#1}}}
\newrefformat{sec}{\savehyperref{#1}{Section~\ref*{#1}}}
\newrefformat{app}{\savehyperref{#1}{Appendix~\ref*{#1}}}
\newrefformat{tab}{\savehyperref{#1}{Table~\ref*{#1}}}
\newrefformat{fig}{\savehyperref{#1}{Figure~\ref*{#1}}}
\newrefformat{hyp}{\savehyperref{#1}{Hypothesis~\ref*{#1}}}
\newrefformat{alg}{\savehyperref{#1}{Algorithm~\ref*{#1}}}
\newrefformat{rem}{\savehyperref{#1}{Remark~\ref*{#1}}}
\newrefformat{item}{\savehyperref{#1}{Item~\ref*{#1}}}
\newrefformat{step}{\savehyperref{#1}{step~\ref*{#1}}}
\newrefformat{conj}{\savehyperref{#1}{Conjecture~\ref*{#1}}}
\newrefformat{fact}{\savehyperref{#1}{Fact~\ref*{#1}}}
\newrefformat{prop}{\savehyperref{#1}{Proposition~\ref*{#1}}}
\newrefformat{prob}{\savehyperref{#1}{Problem~\ref*{#1}}}
\newrefformat{claim}{\savehyperref{#1}{Claim~\ref*{#1}}}
\newrefformat{relax}{\savehyperref{#1}{Relaxation~\ref*{#1}}}
\newrefformat{red}{\savehyperref{#1}{Reduction~\ref*{#1}}}
\newrefformat{part}{\savehyperref{#1}{Part~\ref*{#1}}}

\newcommand{\Sref}[1]{\hyperref[#1]{\S\ref*{#1}}}

\usepackage{nicefrac}

\ifnum\usemicrotype=1
\usepackage{microtype}
\fi

\ifnum\showauthornotes=1
\newcommand{\Authornote}[2]{{\sffamily\small\color{red}{[#1: #2]}}}
\newcommand{\Authornotecolored}[3]{{\sffamily\small\color{#1}{[#2: #3]}}}
\newcommand{\Authorcomment}[2]{{\sffamily\small\color{gray}{[#1: #2]}}}
\newcommand{\Authorstartcomment}[1]{\sffamily\small\color{gray}[#1: }

\newcommand{\Authorfnote}[2]{\footnote{\color{red}{#1: #2}}}
\newcommand{\Authorfixme}[1]{\Authornote{#1}{\textbf{??}}}
\newcommand{\Authormarginmark}[1]{\marginpar{\textcolor{red}{\fbox{\Large #1:!}}}}
\else
\newcommand{\Authornote}[2]{}
\newcommand{\Authornotecolored}[3]{}
\newcommand{\Authorcomment}[2]{}
\newcommand{\Authorstartcomment}[1]{}

\newcommand{\Authorfnote}[2]{}
\newcommand{\Authorfixme}[1]{}
\newcommand{\Authormarginmark}[1]{}
\fi

\newcommand{\pnote}{\Authornote{P}}

\newcommand{\Bnote}{\Authornote{B}}

\ifnum\showfixme=0

\fi

\usepackage{boxedminipage}

\newcommand{\Paren}[1]{\left(#1\right)}

\newcommand{\Brac}[1]{\left[#1\right]}

\newcommand{\Abs}[1]{\left\lvert#1\right\rvert}

\newcommand{\Norm}[1]{\left\lVert#1\right\rVert}

\newcommand{\iprod}[1]{\langle#1\rangle}

\newcommand{\ceil}[1]{\lceil #1 \rceil}
\newcommand{\Esymb}{\mathbb{E}}
\newcommand{\Psymb}{\mathbb{P}}

\DeclareMathOperator*{\E}{\Esymb}

\DeclareMathOperator*{\ProbOp}{\Psymb}

\renewcommand{\Pr}{\ProbOp}

\newcommand{\defeq}{\stackrel{\mathrm{def}}=}

\newcommand{\mper}{\,.}
\newcommand{\mcom}{\,,}

\newcommand\bdot\bullet

\ifx\mathds\undefined 

\else

\fi

\DeclareMathOperator{\Tr}{Tr}

\DeclareMathOperator{\val}{val}

\DeclareMathOperator{\poly}{poly}

\newcommand{\Erdos}{Erd\H{o}s\xspace}
\newcommand{\Renyi}{R\'enyi\xspace}
\newcommand{\Lovasz}{Lov\'asz\xspace}

\newcommand{\N}{\mathbb N}
\newcommand{\R}{\mathbb R}

\newcommand{\cC}{\mathcal C}

\newcommand{\cE}{\mathcal E}

\newcommand{\cL}{\mathcal L}
\newcommand{\cM}{\mathcal M}
\newcommand{\cN}{\mathcal N}

\newcommand{\cQ}{\mathcal Q}
\newcommand{\cR}{\mathcal R}
\newcommand{\cS}{\mathcal S}
\newcommand{\cT}{\mathcal T}
\newcommand{\cU}{\mathcal U}
\newcommand{\cV}{\mathcal V}
\newcommand{\cW}{\mathcal W}

\newcommand{\cZ}{\mathcal Z}

\renewcommand{\leq}{\leqslant}

\renewcommand{\geq}{\geqslant}

\ifnum\showdraftbox=1
\newcommand{\draftbox}{\begin{center}
  \fbox{%
    \begin{minipage}{2in}%
      \begin{center}%
          \Large\textsc{Working Draft}\\%
        Please do not distribute%
      \end{center}%
    \end{minipage}%
  }%
\end{center}
\vspace{0.2cm}}
\else
\newcommand{\draftbox}{}
\fi

\let\epsilon=\varepsilon

\numberwithin{equation}{section}

\newcommand\MYcurrentlabel{xxx}

\newcommand{\MYstore}[2]{%
  \global\expandafter \def \csname MYMEMORY #1 \endcsname{#2}%
}

\newcommand{\MYload}[1]{%
  \csname MYMEMORY #1 \endcsname%
}

\newcommand{\MYnewlabel}[1]{%
  \renewcommand\MYcurrentlabel{#1}%
  \MYoldlabel{#1}%
}

\newcommand{\MYdummylabel}[1]{}

\newcommand{\torestate}[1]{%
  \let\MYoldlabel\label%
  \let\label\MYnewlabel%
  #1%
  \MYstore{\MYcurrentlabel}{#1}%
  \let\label\MYoldlabel%
}

\newcommand{\restatetheorem}[1]{%
  \let\MYoldlabel\label
  \let\label\MYdummylabel
  \begin{theorem*}[Restatement of \prettyref{#1}]
    \MYload{#1}
  \end{theorem*}
  \let\label\MYoldlabel
}

\newcommand{\restatelemma}[1]{%
  \let\MYoldlabel\label
  \let\label\MYdummylabel
  \begin{lemma*}[Restatement of \prettyref{#1}]
    \MYload{#1}
  \end{lemma*}
  \let\label\MYoldlabel
}

\newcommand{\restateprop}[1]{%
  \let\MYoldlabel\label
  \let\label\MYdummylabel
  \begin{proposition*}[Restatement of \prettyref{#1}]
    \MYload{#1}
  \end{proposition*}
  \let\label\MYoldlabel
}

\newcommand{\restatefact}[1]{%
  \let\MYoldlabel\label
  \let\label\MYdummylabel
  \begin{fact*}[Restatement of \prettyref{#1}]
    \MYload{#1}
  \end{fact*}
  \let\label\MYoldlabel
}

\newcommand{\restate}[1]{%
  \let\MYoldlabel\label
  \let\label\MYdummylabel
  \MYload{#1}
  \let\label\MYoldlabel
}

\newcommand{\addreferencesection}{
  \phantomsection
  \addcontentsline{toc}{section}{References}
}

\let\origparagraph\paragraph
\renewcommand{\paragraph}[1]{\origparagraph{#1.}}

\allowdisplaybreaks

\usepackage{paralist}

\usepackage{comment}

\let\citet\cite

\theoremstyle{definition}

\usepackage{relsize}
\usepackage[font=footnotesize]{caption}
\usepackage{yfonts}

\newcommand{\Id}{\mathop{\mathrm{Id}}\!\mathinner{}}
\newcommand{\pE}{\tilde {\mathbb E}}

\let\cL\relax
\DeclareMathOperator{\cL}{\mathcal L}

\newcommand{\Span}{\mathsf{Span}}

\DeclareUrlCommand\email{}

\newcommand{\omegan}{{\left( \frac{\omega}{n} \right)}}

\newcommand{\fU}{\mathfrak{U}}
\newcommand{\fD}{\mathfrak{D}}

\usepackage{appendix}
\usepackage{amssymb}
\usepackage{amsfonts}
\usepackage{latexsym}
\usepackage{amsmath}

\usepackage[T1]{fontenc}
\usepackage{fullpage,appendix}
\usepackage{dsfont}
\usepackage[capitalize]{cleveref}

\usepackage{color}
\usepackage{tikz}

\newcommand{\eqdef}{\stackrel{\textrm{def}}{=}}

\newcommand{\ignore}[1]{}
\definecolor{corlinks}{RGB}{64,128,128}
\definecolor{cormenu}{RGB}{0,37,94}
\definecolor{corurl}{RGB}{0,46,91}

\newcommand{\on}{\{-1,1\}}
\newcommand{\1}{\mathds{1}}

\renewcommand{\P}{\mathcal{P}}

\newcommand{\G}{\mathcal{G}}
\newcommand{\V}{\mathcal{V}}

\newcommand{\transposed}{\dagger}
\newcommand{\g}{G}

\newcommand{\M}{\mathcal{M}}
\renewcommand{\S}{\mathbb{S}}

\newcommand{\floor}[1]{\left \lfloor #1 \right \rfloor}
\newcommand{\zo}{\{0, 1\}}

\renewcommand{\cal}[1]{\mathcal{#1}}
\renewcommand{\int}{\mathsf{int}}

\newcommand{\nchoose}[1]{{{[n]} \choose {#1}}}

\renewcommand{\top}{\dagger}

\renewcommand{\1}{\mathbf{1}}

\renewcommand{\hat}{\widehat}

\renewcommand{\on}[1]{\tfrac{\binom{\omega}{#1}}{\binom{n}{#1}}}

\newcommand{\xor}{\oplus}

\renewcommand{\emptyset}{\varnothing}

\makeatletter
\let\orgdescriptionlabel\descriptionlabel
\renewcommand*{\descriptionlabel}[1]{%
  \let\orglabel\label
  \let\label\@gobble
  \phantomsection
  \edef\@currentlabel{#1}%
  \let\label\orglabel
  \orgdescriptionlabel{#1}%
}
\makeatother

\renewcommand{\on}{\{-1,1\}}
\title{A Nearly Tight Sum-of-Squares Lower Bound for the Planted Clique Problem}

\author{
Boaz Barak
\thanks{Harvard John A. Paulson School of Engineering and Applied Sciences, \url{b@boazbarak.org}. Part of the work was done while the author was at Microsoft Research New England.}
\\ Harvard University
\and Samuel B. Hopkins
\thanks{\url{samhop@cs.cornell.edu}. Partially supported by an NSF GRFP under grant no. 1144153. Part of this work was done while the author was at Microsoft Research New England.}
\\Cornell University
\and Jonathan Kelner
\thanks{\url{kelner@mit.edu}.  Partially supported by NSF Award 1111109.}
\\MIT
\and Pravesh K. Kothari
\thanks{\url{kothari@cs.utexas.edu} Part of the work was done while the author was at Microsoft Research New England.}
\\ UT Austin
\and Ankur Moitra
\thanks{\url{moitra@mit.edu} Partially supported by NSF CAREER Award CCF-1453261, a grant from the MIT NEC Corporation and a Google Faculty Research Award.}
\\MIT 
\and Aaron Potechin
\thanks{\url{aaronpotechin@gmail.com} Part of the work was done while the author was at Microsoft Research New England.}
\\Cornell University}

\begin{document}

\maketitle
\draftbox
\thispagestyle{empty}

\begin{abstract}
We prove that with high probability over the choice of a random graph $G$ from the Erdős–Rényi distribution $G(n,1/2)$, the $n^{O(d)}$-time degree $d$ Sum-of-Squares semidefinite programming relaxation for the clique problem will give a value of at least $n^{1/2-c(d/\log n)^{1/2}}$ for some constant $c>0$.  
This yields a nearly tight $n^{1/2 - o(1)}$ bound on the value of this program for any degree $d = o(\log n)$. Moreover we introduce a new framework that we call \emph{pseudo-calibration} to construct Sum of Squares lower bounds.
This framework is  inspired by taking a computational analog of Bayesian probability theory. It yields a general recipe for constructing  good pseudo-distributions (i.e., dual certificates for the Sum-of-Squares semidefinite program),  and sheds further light on the ways in which this hierarchy differs from others. 
\end{abstract}

\clearpage

\ifnum\showtableofcontents=1
{
\tableofcontents
\thispagestyle{empty}
 }
\fi

\clearpage

\setcounter{page}{1}
\color{black}
\section{Introduction}

The \emph{planted clique} (also known as \emph{hidden clique}) problem is a central question in average-case complexity.
Arising from the 1976 work of Karp~\cite{Kar76}, the problem was formally defined by Jerrum \cite{DBLP:journals/rsa/Jerrum92} and Kucera \cite{DBLP:journals/dam/Kucera95} as follows: given a random \Erdos-\Renyi graph $G$ from the distribution $G(n,1/2)$ (where every edge is chosen to be included with probability $1/2$  independently of all others) in which we \emph{plant} an additional clique (i.e., set of vertices that are all neighbors of one another) $S$ of size $\omega$, find $S$. It is not hard to see that the problem is solvable by brute force search (which in this case takes quasipolynomial time) whenever $\omega > c \log n$ for any constant $c>2$. However, despite intense effort, the best polynomial-time algorithms only work for $\omega = \epsilon \sqrt{n}$, for any constant $\epsilon>0$ \cite{DBLP:conf/soda/AlonKS98}. 

Over the years the planted clique problem and related problems have been connected to many other questions in a variety of areas including finding communities~\cite{HajekWX15}, finding signals in molecular biology \cite{PS00}, discovering network motifs in biological networks~\cite{Milo, JM15},
computing Nash equilibrium~\cite{HazanK11,AustrinBC13}, property testing~\cite{DBLP:conf/stoc/AlonAKMRX07}, sparse principal component analysis~\cite{BerthetR13}, compressed sensing~\cite{KoiranZ14}, cryptography~\cite{AriJuels2000,DBLP:conf/stoc/ApplebaumBW10} and even mathematical finance~\cite{DBLP:conf/innovations/AroraBBG10}.

Thus, the question of whether the currently known algorithms can be improved is of great interest. Unfortunately, it is unlikely that lower bounds for planted clique (because it is an average-case problem) can be derived from conjectured complexity class separations such as $P\neq NP$~\cite{FeigenbaumFortnow,BogdanovTrevisan}. Our best evidence for the difficulty of this problem comes from works showing limitations on particular \emph{classes} of algorithms. In particular, since many of the algorithmic approaches for this and related problems involve spectral techniques and convex programs, limitations for these types of algorithm are of significant interest. One such negative result was shown by Feige and Krauthgamer~\cite{DBLP:journals/siamcomp/FeigeK03} who proved that the $n^{O(d)}$-time  \emph{degree $d$ \Lovasz -Schrijver  semidefinite programming hierarchy } ($LS_+$ for short) can only recover the clique if its size is at least $\sqrt{n/2^d}$.%
\footnote{As we discuss in Remark~\ref{rem:planted-clique-variants} below, formally such results apply to the incomparable \textit{refutation} problem, which is the task of certifying that there is no $\omega$-sized clique in a random $G(n,1/2)$ graph. However, our current knowledge is consistent with these variants having the same computational complexity.}

However, recently it was shown that in several cases, the \emph{Sum-of-Squares (SoS) hierarchy}~\cite{Shor87,Parrilo00,Lasserre01} \---- a stronger family of semidefinite programs which can be solved in time $n^{O(d)}$ for degree parameter $d$ \---- can be significantly more powerful than other algorithms such as  $LS_{+}$~\cite{BBHKSZ12,BKS14,BKS15}. Thus it was conceivable that the SOS hierarchy might be able to find cliques that are much smaller than
$\sqrt{n}$ in polynomial time. 

The first SoS lower bound for planted clique was shown by Meka, Potechin and Wigderson~\cite{MPW15} who proved that the degree $d$ SOS hierarchy cannot recover a clique of size $\Tilde{O}(n^{1/d})$. 
This bound was later improved on by  Deshpande and Montanari~\cite{DM15} and then Hopkins et al~\cite{HKPRS16} to $\Tilde{O}(n^{1/2})$ for degree $d=4$ and $\Tilde{O}(n^{1/(\ceil{d/2}+1)})$ for general $d$.
However, this still left open the possibility that the  constant degree (and hence polynomial time) SoS algorithm can significantly beat the $\sqrt{n}$ bound, perhaps even being able to find cliques of size $n^{\epsilon}$ for any fixed $\epsilon > 0$.    
This paper answers this question negatively by proving the following theorem:

\begin{theorem}[Main Theorem]
  \label{thm:main}
  There is an absolute constant $c$ so that for every $d= d(n)$ and large enough $n$, the SoS relaxation of the planted clique problem has integrality gap at least $n^{1/2 - c(d/\log n)^{1/2}}$.
\end{theorem}

Beyond improving the previously known results, our proof is significantly more general and we believe provides a better intuition behind the limitations for SoS algorithms by viewing them from a ``computational Bayesian probability'' lens that is of its own interest. Moreover, there is some hope (as we elaborate below) that this view could be useful not just for more negative results but for SoS \emph{upper bounds} as well. In particular our proof elucidates to a certain extent the way in which the SoS algorithm is more powerful than the $LS_+$ algorithm. 

\medskip

\begin{quoting}
\begin{remark}[\textit{The different variants of the planted clique problem}] \label{rem:planted-clique-variants} Like other average-case problems in $\mathbf{NP}$, the planted clique problem with parameter $\omega$ has three variants of \textit{search}, \textit{refutation}, and \textit{decision}. The \textit{search} variant is the task of recovering the clique from a graph in which it was planted. The \textit{refutation} variant is the task of \textit{certifying} that a random graph in $G(n,1/2)$  (where with high probability the largest clique has size $(2+o(1))\log n$) does not have a clique of size $\omega$. The \textit{decision} problem is to distinguish between a random graph from $G(n,1/2)$ and a graph in which an $\omega$-sized clique has been planted. 
The decision variant can be reduced to either the search or the refutation variant, but we know of no reduction between the latter two variants. 
Integrality gaps for mathematical relaxations such as the Sum-of-Squares hierarchy are most naturally stated as negative results for the \textit{refutation} variant, as they show that such relaxations cannot certify that a random graph has no $\omega$-sized clique by looking at the maximum value of the objective function.  
Our result can also be viewed as showing that the natural SoS-based algorithm for the \textit{decision} problem (which attempts to distinguish on the objective value) also fails.
Moreover, our result also rules out some types of SoS-based algorithms for the \textit{search} problem as it shows that in a graph with a planted clique, there exists a solution with an objective value of $\omega$ based only on the random part, which means that it does not contain any information about which nodes participate in the clique and hence is not useful for rounding algorithms. 
\end{remark}\end{quoting}

\section{Planted Clique and Probabilistic Inference}

We now discuss the ways in which the planted clique problem differs from  problems for which strong SoS lower bounds have been shown before, and how this relates to a ``computational Bayesian'' perspective.
There have been several strong lower bounds for the SoS algorithm before, in particular for problems such as 3SAT, 3XOR and other constraint satisfaction problems as well as the knapsack problem~\cite{Gri01,Scho08,BCK15}.
However, obtaining strong lower bounds for the planted clique problem seems to have required different techniques.
A high-level way to describe the difference between the planted clique problems and the problems tackled by previous results is that lower bounds for the planted clique problem boil down to handling \textbf{weak global constraints} as opposed to  \textbf{strong local} ones.
That is, while in the random 3SAT/3XOR setting,  the effect of one variable on another is either extremely strong (if they are "nearby" in the formula) or essentially zero, in the planted clique setting every variable has a weak \textit{global} effect on all other variables. 
We now explain this in more detail.

Consider a random graph $G$ in which a clique $S$ of size $\omega$ has been planted. 
If someone tells us that vertex $17$ is not in $S$, then it makes it slightly less likely that $17$'s  neighbors are in $S$ and slightly more likely that $17$'s non-neighbors are in $S$. 
So, this information has a  \textit{weak global} effect.
In contrast, when  we have a random sparse 3SAT formula $\varphi$ in which an assignment $x$ has been planted,  if someone tells us that $x_{17}=0$ then it gives us a lot of information about the local neighborhood of the $17^{th}$ variable (the variables that are involved in constraints with $17$ or one that have a short path of constraints to it) but there is an exponential decay of these correlations and so this information basically tells us essentially nothing about the distribution of most of the variables $x_i$ (that are far away from $17$ in the sparse graph induced by $\varphi$).%
x\footnote{This exponential decay can be shown formally for the case of satisfiable random 3SAT or 3XOR formulas whose clause density is sufficiently smaller than the threshold. In our regime of overconstrainted random 3SAT/3XOR formulas there will not exist any satisfying assignments, and so to talk about ``correlations'' in the distributions of assignments we need to talk about the ``Bayesian estimates'' that arise from algorithms such as Sum-of-Squares or belief propagation. Both these algorithms exhibit this sort of exponential decay we talk about; see also Remark~\ref{rem:3XOR}}
Thus, in the random 3SAT setting information about the assignments of individual variables has a \textit{strong local effect}. 
Indeed, previous Sum-of-Squares lower bounds for random 3SAT and 3XOR~\cite{Gri01,Scho08}, could be interpreted as producing "distribution like" objects in which, conditioned on the value of a small set of variables $S$, some of the variables "close" to $S$ in the formula were completely fixed, and the rest were completely independent. 

This difference between the random SAT and the planted clique problems means that some subtleties that can be ignored in setting of  random constraint satisfaction problems need to be tackled head-on when dealing with planted cliques.
However to make this clearer, we need to take a detour and discuss Bayesian probabilities and their relation to the Sum of Square Algorithm.

\subsection{Computational Bayesian Probabilities and Pseudo-distributions}

Strictly speaking, if a graph $G$ contains a unique clique $S$ of size $\omega$, for every vertex $i$, the probability that  $i$ is in $S$ is either zero or one. 
But,  a computationally bounded observer may not know whether $i$ is in the clique or not, and we could try to quantify this ignorance using probabilities. 
These can be thought of as a computational analogs of \textit{Bayesian probabilities}, that, rather than aiming to measure  the frequency at which an event occurs in some sample space, attempt to capture the subjective beliefs of some observer.

That is, the Bayesian probability that an observer $B$ assigns to an event $E$ can be thought of as corresponding to the odds at which $B$ would make the bet that $E$ holds. 
Note that this probability could be strictly between zero and one even if the event $E$ is fully determined, depending on the evidence available to $B$. 
While typically Bayesian analysis  does not take into account computational limitations, one could imagine that  even if $B$ has access to information that fully determines whether $E$ happened or not,  she could still rationally assign a subjective probability to $E$ that is strictly between zero and one  if making the inferences from this information is computationally infeasible.
In particular, in the example above, even if a computationally bounded observer has access to the graph $G$, which information-theoretically fully determines the planted $\omega$-sized clique, she could still assign a probability strictly between zero and one to the event that vertex $17$ is in the planted $\omega$-sized clique, based on some simple to compute statistics such as how many neighbors $17$ has, etc.

The Sum-of-Squares algorithm can be thought of as giving rise to an internally consistent set of such "computational probabilities". 
These probabilities may not capture \textit{all} possible inferences that a computationally bounded observer could make, but they do capture all inferences that can be made via a certain restricted proof system.

\paragraph{Bayesian estimates for planted clique}
To get a sense for our results and techniques, it is instructive to consider the following scenario. 
Let $G(n,1/2,\omega)$ be the distribution over pairs $(G,x)$ of $n$-vertex graphs $G$ and vectors $x\in\R^n$ that is obtained by sampling a random graph in $G(n,1/2)$, planting an $\omega$-sized clique in it, and letting $G$ be the resulting graph and $x$ the $0/1$ characteristic vector of the clique.
Let $f:\{0,1\}^{\binom{n}{2}}\times \R^n \rightarrow \R$ be some function that maps a graph $G$ and a vector $x$ into some real number $f_G(x)$. 
Now imagine two parties, Alice and Bob (where Bob can also stand for "Bayesian") that play the following game: Alice samples $(G,x)$ from the distribution $G(n,1/2,\omega)$ and sends $G$ to Bob, who needs to output the expected value of $f_G(x)$. We denote this value by $\pE_G f_G$. 

If we have no computational constraints then it is clear that Bob will simply let $\pE_G f_G$ be equal to  $\E_{x|G} f_G(x)$, by which we mean the expected value of $f_G(x)$ where $x$ is chosen according to the conditional  distribution on $x$ given the graph $G$.\footnote{The astute reader might note that this expectation is somewhat degenerate since with very high probability the graph $G$ will uniquely determine the vector $x$, but please bear with us, as in the computational setting we will be able to treat $x$ as "undetermined". }  In particular, the value $\pE_G f_G$ will be \textit{calibrated} in the sense that 
\begin{equation}
\E_{G \in_R G(n,1/2,\omega)} \pE_G f_G = \E_{(G,x) \in_R G(n,1/2,\omega)} f_G(x) \label{eq:calibration}
\end{equation}

Now if Bob is computationally bounded, then he might not be able to compute the value of $E_{x|G} f_G(x)$ even for a simple function such as $f_G(x)=x_{17}$.
Indeed, as we mentioned, since with high probability the clique $x$ is uniquely determined by $G$, $\E_{x|G} x_{17}$ will simply equal $1$ if vertex $17$ is in the clique and equal $0$ otherwise. 
However, note that we don't need to compute the true conditional expectation to obtain a calibrated estimate. In particular, in the above example, simply setting $\pE x_{17} = \omega/n$ will satisfy (\ref{eq:calibration}).

Our Sum-of-Squares lower bound amounts to coming up with some reasonable ``pseudo-expectation'' that can be efficiently computed, where $\pE_G$ is meant to capture a ``best effort'' of a computationally bounded party of approximating the Bayesian conditional expectation $\E_{x|G}$.
Our pseudo-expectation will not be even close to the true conditional expectations, but will at least be internally consistent in the sense that for ``simple'' functions $f$ it satisfies (\ref{eq:calibration}). 
It will also satisfy some basic sanity checks such as that for every graph $G$ and ``simple'' $f$, $\pE_G f_G^2 \geq 0$.
In fact, since the pseudo-expectation will not distinguish between a graph $G$ drawn 
from $G(n,1/2,\omega)$ and a random $G$ from $G(n,1/2)$ it will also satisfy the following 
\textit{pseudo-calibration} condition:

\begin{equation}
\E_{G \in_R G(n,1/2)} \pE_G f_G = \E_{(G,x) \in_R G(n,1/2,\omega)} f_G(x) \label{eq:pseudo-calibration}
\end{equation}

for all ``simple'' functions $f=f(G,x)$. Note that (\ref{eq:pseudo-calibration}) does not make sense for the estimates of a truly Bayesian (i.e., computationally unbounded)  Bob, since   almost all graphs $G$ in $G(n,1/2)$ are not even in the support of $G(n,1/2,\omega)$. 
Nevertheless, our pseudo-distributions will be well defined even for a random graph and hence will yield estimates for the probabilities over this hypothetical object (i.e., the $\omega$-sized clique) that does not exist.
The ``pseudo-calibration'' condition (\ref{eq:pseudo-calibration}) might seem innocent, but it turns out to imply many useful properties. 
In particular is not hard to see that  (\ref{eq:pseudo-calibration}) implies that for every \textit{simple strong constraint} of the clique problem--- a function $f$ such that $f(G,x)=0$ for every $x$ that is a characteristic vector of an $\omega$-clique in $G$--- it must hold that $\pE_G f_G = 0$.
But even beyond these ``strong constraints'', (\ref{eq:pseudo-calibration}) implies that the  pseudo-expectation satisfies many \textit{weak constraints} as well, such as the fact that a vertex of high degree is more likely to be in the clique and  that if  $i$ is not in the clique then its neighbors are less likely and non-neighbors are more  likely to be in it. 

Indeed, the key conceptual insight of this paper is to phrase the calibration property (\ref{eq:pseudo-calibration}) as a desiderata for our pseudo-distributions. 
Namely, we define that a function $f=f(G,x)$ is ``simple'' if it is a low degree polynomial in both the entries of $G$'s adjacency matrix and the variables $x$, and then require (\ref{eq:pseudo-calibration}) to hold for all simple functions. 
It turns out that once you do so, the choice for the pseudo-distribution is essentially determined, and hence proving the main result amounts to showing that it satisfies the constraints of the SoS algorithm. 
In the next section we will outline some of the ideas involved in this proof.

\begin{quoting}
\begin{remark}[Planted Clique vs 3XOR] \label{rem:3XOR} In the light of the discussion above, it is instructive to consider the case of random 3XOR discussed before. Random 3XOR instances on $n$ variables and $\Theta(n)$ constraints are easily seen to be maximally unsatisfiable (that is, at most $ \approx 1/2$ the constraints can be satisfied by any assignment) with high probability. On the other hand, Grigorev \cite{Gri01} constructed a sum of squares pseudoexpectation that pretends that such instances instances are satisfiable with high probability, proving a sum of squares lower bound for refuting random 3XOR formulas. 

Analogous to the planted distribution $G(n,1/2, \omega)$, one can define a natural planted distribution over 3XOR instances - roughly speaking, this corresponds to first choosing a random Boolean assignment $x^{*}$ to $n$ variables and then sampling random 3XOR constraints conditioned on being consistent with $x^{*}$. 
It is not hard to show that pseudo-calibrating with respect to this planted distribution a la (\ref{eq:pseudo-calibration}) produces precisely the pseudoexpectation that Grigoriev constructed. 
However, unlike in the planted clique case, in the case of  3XOR, the pseudo-calibration condition implies that for every low-degree monomial $x_S$, either the value of $x_S$ is completely fixed (if it can be derived via low width resolution from the 3XOR equations of the instance) or it is completely unconstrained. 

The pseudoexpectations considered in previous works \cite{FeigeK03,MPW15,DM15}) are similar to Grigoriev's construction, in the sense that they essentially  respect only strong constraints (e.g., that if $A$ is not a clique in the graph, then the probability that it is contained in the planted clique is zero), but other than that assume that variables  are independent. 
However, unlike the 3XOR case, in the planted clique problem  respecting these strong constraints is not enough to achieve the pseudo-calibration condition (\ref{eq:pseudo-calibration})  and the pseudoexpectation of \cite{FeigeK03,MPW15,DM15} can be shown to violate weak probabilistic constraints imposed by (\ref{eq:pseudo-calibration}) even at degree four. See Observation~\ref{obs:fail-to-calibrate} for an example.
\end{remark}\end{quoting}

\subsection{From Calibrated Pseudo-distributions to Sum-of-Squares Lower Bounds}
\label{sec:intro:pseudo-dist}

What do these Bayesian inferences and calibrations have to do with Sum-of-Squares? In this section, we show how calibration is almost forced on any pseudodistribution feasible for the sum of squares algorithm. Specifically, to show that the degree $d$ SoS algorithm fails to certify that a random graph does not contain a clique of size $\omega$, we need to show that for a random $G$, with high probability we can come up with an operator that maps a degree at most $d$, $n$-variate polynomial $p$ to a real number $\pE_G p$ satisfying the following constraints:

\begin{enumerate}
\vspace{-0.5ex}
\item (Linearity) The map $p\mapsto \pE_G p$ is linear.
\vspace{-0.5ex}
\item (Normalization) $\pE_G 1 = 1$.
\vspace{-0.5ex}
\item (Booleanity constraint) $\pE_G x_i^2 p = \pE x_i p$ for every $p$ of degree at most $d-2$ and $i\in [n]$.
\vspace{-0.5ex}
\item (Clique constraint) $\pE_G x_ix_j p = 0$ for every $(i,j)$ that is not an edge and $p$ of degree at most $d-2$.
\vspace{-0.5ex}
\item (Size constraint) $\pE_G \sum_{i=1}^n x_i = \omega$.
\vspace{-0.5ex}
\item (Positivity) $\pE_G p^2 \geq 0$ for every $p$ of degree at most $d/2$.
\end{enumerate}

\begin{definition}\label{def:pdist}
A map $p \mapsto \pE_G p$ satisfying the above constraints 1--6 is called a \textit{degree $d$ pseudo-distribution} (w.r.t. the planted clique problem with parameter $\omega$).  
\end{definition}

We can restate our main result as follows:

\begin{theorem}[Theorem~\ref{thm:main}, restated] \label{thm:main:restated} There is some constant $c$ such that if $\omega \leq n^{1/2-c(d/\log n)^{1/2}}$ then with high probability over $G$ sampled from $G(n,1/2)$, there is a degree $d$ pseudodistribution $\pE_G$ satisfying constraints 1--6 above.
\end{theorem}

\noindent Note that all of these constraints would be satisfied if $\pE_G p$ was obtained by taking the expectation of $p$ over a distribution on $\omega$-sized cliques in $G$. 
However, with high probability there is not event a $2.1\log n$-sized clique in $G$ (and let alone a $\approx \sqrt{n}$ sized one) so we will need a completely different mechanism to obtain such a pseudo-distribution. 

Previously, the choice of the pseudo-distribution seemed to require a ``creative guess'' or an ``ansatz''. 
For problems such as  random 3SAT this guess was fairly natural and almost ``forced'', while for planted clique planted clique as well as some related problems \cite{MW15} the choice of the pseudo-distribution seemed to have more freedom, and more than one choice appeared in the literature.

For example, Feige and Krauthgamer~\cite{FeigeK03} (henceforth FK) defined a very natural pseudo-distribution $\pE^{FK}$ for a weaker hierarchy.  
For a graph $G$ on $n$ vertices, and subset $A \subseteq [n]$, $\pE^{FK}_G x_A$ is equal to zero if $A$ is not a clique in $G$ and equal to $2^{\binom{|A|}{2}}\left( \tfrac{\omega}{n} \right)^{|A|}$ if $A$ is a clique, and extended to degree $d$ polynomials using linearity.%
\footnote{The actual pseudo-distribution used by \cite{FeigeK03} (and the followup works \cite{MPW15,DM15}) was slightly different so as to satisfy $\pE_G (\sum_{i=1}^m x_i)^\ell = \omega^\ell$ for every $\ell \in \{1,\ldots,d\}$.
This property is sometimes described as satisfying the constraint $\{ \sum_i x_i = \omega \}$.}
\cite{FeigeK03} showed that that for every $d$, and $\omega < O(\sqrt{n/2^d})$, this pseudo-distribution satisfies the constraints 1--5 as in Definition~\ref{def:pdist} as well as  a weaker version of positivity (this amounts to the so called ``Lovász-Schrijver+'' SDP). Meka, Potechin and Wigderson~\cite{MPW15} proved that the same pseudo-distribution satisfies all the constraints 1--6 (and hence is a valid degree $d$ pseudo-distribution) as long as $\omega < \Tilde{O}(n^{1/d})$. This bound on $\omega$ was later improved to $\Tilde{O}(n^{1/3})$ for $d=4$ by \cite{DM15} and to $\Tilde{O}(n^{(\floor{d/2}+1)^{-1}})$ for a general $d$ by \cite{HKP15}.

Interestingly, the FK  pseudo-distribution does \textit{not} satisfy the full positivity constraint  for larger values of $\omega$. 
The issue is that while that while the FK pseudo-distribution satisfies the ``strong'' constraints that $\pE_G^{FK} x_A = 0$ if $A$ is not a clique, it does not satisfy weaker constraints that are implied by (\ref{eq:pseudo-calibration}).
For example, for every constant $\ell$, if vertex $i$ participates in $\sqrt{n}$ more $\ell$-cliques than the expected number then one can compute that the conditional probability of $i$ belonging in the clique should be a factor $1+c\omega/\sqrt{n}$  larger for some constant $c>0$. 
However, the FK pseudo-distribution does not make this correction. In particular, for every $\ell$, there's a simple polynomial that shows that the FK pseudoexpectation is not calibrated. 

\medskip

\begin{observation} \label{obs:fail-to-calibrate}
Fix $i\in [n]$ and let $\ell$ be some constant.  If $p_G = (\sum_{j} G_{i,j} x_j)^{\ell}$ then  \textbf{(i)}  $\E_{G \sim G(n,1/2)} \pE^{FK}_G [p_G^2] \leq  \omega^{\ell}$ 
and \textbf{ii}  $\E_{(G,x) \sim G(n,1/2,\omega)} [p_G(x)^2] \geq \frac{\omega^{2\ell+1} }{n}$.
In particular, when $\omega \gg n^{\frac{1}{\ell+1}}$, $\E_{G \sim G(n,1/2)} \pE^{FK}_G [p_G^2] \ll \E_{(G,x) \sim G(n,1/2,\omega)} p_G(x)$.
\end{observation}

\begin{proof}[Proof sketch]
For 2 note that with probability $(\omega/n)$ vertex $i$ is in the clique, in which case $\sum_j G_{i,j} x_j = \omega$, and hence the expectation of $p_G^2$ is at least $(\omega/n)\omega^{2\ell}$.
To compute 1, we open up the expectation and the definition to get (up to a constant depending on $\ell$)   $\sum_{j_1,\ldots,j_{2\ell}} G_{i,j_1}\ldots G_{i,j_{2\ell}} (\omega/n)^{2\ell} \E_{G \sim G(n,1/2)} 1_{\{i_1,\ldots,i_{2\ell}\} \text{ is clique}}$. Since this expectation is zero unless every variable $G_{i,j}$ is squared, in which case the number of distinct $j$'s is at most $\ell$, which means the sum is at most $n^\ell(\omega/n)^\ell = \omega^\ell$. 
\end{proof}

Observation \ref{obs:fail-to-calibrate} notes the failure of calibration for a specific polynomial $p_G(x)$ where the coefficients are (low-degree) functions of the graph $G$. 
The polynomial $p_G$ above can also be massaged to obtain a proof (due to Kelner, see \cite{HKP15})  that degree $d$ $\pE^{FK}$ does not satisfy the  positivity constraint  at degree $d$ for $\omega \gg n^{\frac{1}{\frac{d}{2}+1}}.$

\begin{fact} \label{fact:Kelner}
Let $p_G$ be as in the Observation \ref{obs:fail-to-calibrate}. Then, there exists a $C$ such that for $q = q_G = (C \omega^{\ell} x_S - p_G)$ with high probability over the graph $G \sim G(n, 1/2)$, $\pE^{FK}[ q_G^2] < 0$ for $\omega \gg n^{\frac{1}{\ell+1}}$.
\end{fact}

For the case $d=4$, Hopkins et al~\cite{HKPRS16} proposed an ``ad hoc'' fix for the FK pseudo-distribution that satisfies positivity up to $\omega = \tilde{O}(\sqrt{n})$, by explicitly adding a correction term to essentially calibrate for the low-degree polynomials $q_G$ from Fact \ref{fact:Kelner}. 

However, their method did not extend even for $d=6$, because of the sheer number of corrections that would need to be added and analyzed. Specifically, there are multiple families of polynomials such that their $\pE^{FK}$ value departs significantly from their calibrated value in expectation and gives multiple points of failure of positivity in a manner similar to Observation \ref{obs:fail-to-calibrate} and Fact \ref{fact:Kelner}. Moreover, "fixing" these families by the correction as in case of degree four leads to new families of polynomials that fail to achieve their calibrated value and exhibit negative pseudoexpectation for their squares etc. 

The \textit{coefficients} of the polynomial $p_G$ of Observation~\ref{obs:fail-to-calibrate} are themselves low degree polynomials in the adjacency matrix of $G$.
This turns out to be a common feature in all the families of polynomials one encounters in the above works. 
Thus our approach is  to fix all these polynomials \textit{by fiat}, by placing the constraint that the pseudo-distribution must satisfy (\ref{eq:pseudo-calibration}) for every such polynomial, and using that as our implicit definition of the pseudo-distribution. 
Indeed it turns our that once we do so, the pseudo-distribution is essentially determined. 
Moreover,  (\ref{eq:pseudo-calibration}) guarantees that it satisfies many of the ``weak global constraints'' that can be shown using Bayesian calculations. 

Pseudo-calibrating polynomials whose coefficient are low-degree in $G$ amounts to restricting the pseudo-distribution to satisfy that the map $G \mapsto \pE_G$ is itself a low degree polynomial in $G$. 
Why is it OK to make such a restriction?
One justification is the  heuristic that the pseudo-distribution itself must be simple since we know that it is  efficiently computable (via the SoS algorithm) from the graph $G$. 
Another justification is that by forcing the pseudo-distribution to be low-degree we are essentially making it \textit{smooth} or ``high entropy'',  which is consistent with the Jaynes \textit{maximum entropy principle}~\cite{Jaynes57,Jaynes58}.
Most importantly -- and this is the bulk of the technical work of this paper and the subject of the next subsection -- this pseudo-distribution can be shown to satisfy \textit{all} the constraints 1--6 of Definition~\ref{def:pdist} including the positivity constraint.

We believe that this principled approach to designing pseudo-distributions elucidates the power and limitations of the SoS algorithm in cases such as the planted clique, where accounting for  weak global correlations is a crucial aspect of the problem. 

\medskip
\begin{quoting}
\begin{remark}[\textit{Where does the planted distribution arise from?}] Theorem~\ref{thm:main:restated} (as well as Theorem~\ref{thm:main}) makes no mention of the planted distribution $G(n,1/2,\omega)$ and only refers to an actual random graph. 
Thus  it might seem strange that we base our pseudo-distribution on the planted distribution via (\ref{eq:pseudo-calibration}).  
One way to think about the planted distribution is that it corresponds to a \textit{Bayesian prior} distribution on the clique. 
Note that this is the \textit{maximum entropy} distribution on cliques of size $\omega$, and so it is a natural choice for a prior per Jaynes's principle of maximum entropy. Our actual pseudo-distribution can be viewed as correcting this planted distribution to a posterior that respects simple inferences from the observed graph $G$. 
\end{remark}
\end{quoting}

\subsection{Proving Positivity}
\label{sec:provingpositivity}

Now we have seen that pseudocalibration is desirable both \emph{a priori} and in light of the failure of previous lower-bound attempts.
We turn to the question: how do we formally define a pseudo-calibrated linear map $\pE_G$, and how do we show that it satisfies constraints (1) -- (6) with high probability, to prove Theorem~\ref{thm:main:restated}?

Recall that our goal is to give a map from $G$ to $\pE_G$ such that when $G$ is taken from $G(n,1/2)$ then with high probability $\pE_G$ satisfies constraints 1--6 of Definition~\ref{def:pdist}. 
Our strategy is to define $\pE_G$ in a way that it satisfies the pseudo-calibration requirement (\ref{eq:pseudo-calibration}) with respect to all functions $f=f(G,x)$ that are low degree polynomials in both the $G$ and $x$ variables. The above requirements determine  all the low-degree Fourier coefficients of the map $G \mapsto \pE_G$. 
Indeed, instantiating (\ref{eq:pseudo-calibration})  with every particular  function $f=f(G,x)$ defines a linear constraint on the pseudo-expectation operator. 
If we require (\ref{eq:pseudo-calibration}) to hold with respect to every function $f=f(G,x)$ that has degree at most $\tau$ in the entries of the adjacency matrix $G$ and degree at most $d$ in the variables $x$, and in addition we require that the map $G \mapsto \pE_G$ is itself of degree at most $\tau$ in $G$, then this completely determines $\pE_G$. For any $S \subseteq [n]$, $|S| \leq d$, using the Fourier transform it is not too hard to compute $\pE_G[x_S]$ as an explicit low degree polynomial in $\g_e$:
\begin{equation}
\label{eq:overview-pE} \pE_G[ x_S] = \sum_{ \substack{T \subseteq \nchoose{2} \\ |\V(T) \cup S| \leq \tau}}  \omegan^{|\V(T) \cup S|} \chi_{T}(G),
\end{equation}
where $\V(T)$ is the set of nodes  incident to the subset of edges (i.e., graph) $T$  and $\chi_T(G) = \prod_{e\in T} G_e$. We carry out this computation in Section \ref{sec:def-intro}. For $\omega \approx n^{0.5 - \epsilon}$, we will need to choose the truncation threshold $\tau \gtrapprox d/\epsilon$. It turns out that constraints 1--5 are easy to verify and thus we are left with proving the \emph{positivity constraint}. Indeed this is not surprising as verifying this constraint is always the hardest part of a sum of squares lower bound.

As is standard, to analyze this positivity requirement we work with the \emph{moment matrix} of $\pE_G$.
Namely, let $\M$ be the ${n \choose \leq d/2} \times {n \choose \leq d/2}$ matrix where $\M(I,J) = \pE_G \prod_{i \in I} x_i \prod_{j \in J} x_j$ for every pair of subsets $I,J \subseteq [n]$ of size at most $d/2$.
Our goal can be rephrased as  showing that $\M \succeq 0$ (i.e., $\M$ is positive semidefinite).

Given a (symmetric) matrix $N$, to show that $N \succeq 0$ our first hope might be to diagonalize $N$.
That is, we would hope to find a matrix $V$ and a diagonal matrix $D$ so that $N = VDV^\top$. Then as long as every entry of $D$ is nonnegative, we would obtain $N \succeq 0$.
Unfortunately, carrying this out directly can be far too complicated.
Even the eigenvectors of very simple random matrices--for example, a matrix with independent $\pm 1$ entries---are not explicitly understood.
Our moment matrix $\M$ is a much more complicated random matrix, with intricate dependencies among the entries.
However, as the next example demonstrates, it is sometimes possible to prove PSDness for a random matrix  using an \emph{approximate} diagonalization.

\paragraph{Example: Planted Clique Lower Bound for $d = 2$ (a.k.a. Basic SDP)}
Consider the problem of producing a pseudodistribution $\pE$ satisfying constraints 1--6 of Definition~\ref{def:pdist}, but only for $d = 2$.
In this simple case, it turns out that the subtleties of (pseudo)calibration may safely be ignored, but it is still instructive to revisit the proof of PSDness.
It will be enough to define $\pE x_i$ and $\pE x_i x_j$ for every $i \in [n]$ and $\{i,j\} \subseteq [n]$.
Let $\pE x_i = (\omega/n)$ for every $i$, and let $\pE x_i x_j$ equal  $\left(\tfrac{\omega}{n}\right)^2$  if $(i,j)$ is an an edge in $G$ and equal zero otherwise.
It's not hard to show that positivity of this pseudo-expectation reduces to showing that $\cN \succeq 0$ where  $\cN$ is the $n\times n$ matrix with $\cN_{i,j}=\pE x_ix_j$. 
Using standard results on random matrices, $\cN$ has one eigenvalue (with eigenvector  very close to the vector $\vec{u}=(1/\sqrt{n},\ldots,1/\sqrt{n})$) of value $\omega^2/n$, while all others are distributed in the interval $\tfrac{\omega}{n} \pm O\left( \tfrac{\omega^2}{n^2}\sqrt{n}\right)$ which is strictly positive as long as $\omega \ll \sqrt{n}$.
Thus, while we cannot explicitly diagonalize $\cN$, we have enough information to conclude that it is positive semidefinite. 
In other words, it was enough for us to get an \textit{approximate diagonalization} for $\cN$ of the form $\cN \approx \tfrac{\omega^2}{n}\vec{u}\vec{u}^\top + \tfrac{\omega}{n}Id + E$ for some sufficiently small (in spectral norm) ``error matrix'' $E$.

\paragraph{Approximate Factorization for $\M$}
We return now to the moment matrix $\M$ for our (pseudo)calibrated pseudodistribution.
Our goal is to give an approximate diagonalization of $\M$.
There are several obstacles to doing so:
\begin{enumerate}
  \item In the case $d = 2$ there was just one rank-$1$ approximate eigenspace to be handled.
  The number of these approximate eigenspaces will grow with $d$, so we will need a more generic way to handle them.
  \item Each approximate eigenspace corresponds to a family of polynomials $\{p\}$ whose calibrated pseudoexpectations are all roughly equal. (In the case $ d= 2$, the only interesting polynomial was the polynomial $\sum_j x_j$ whose coefficients are proportional to the vector $\vec{u}=(1/\sqrt{n},\ldots,1/\sqrt{n})$.) 
  As we saw in Observation~\ref{obs:fail-to-calibrate}, if $p_G$ is a polynomial whose coefficients depend on the graph $G$, even in simple ways, the calibrated value $\pE_G p_G$ may also depend substantially on the graph.
  Thus, when we write $\cM \approx \cL \cQ \cL^\top$ for some approximately-diagonal matrix $\cQ$, we will need the structured part $\cL = \cL(G)$ to itself be graph-dependent.
  \item The errors in our diagonalization of $\M$---corresponding in our $d = 2$ example to the matrix $E$---will not be so small that we can ignore them as we did above.
  Instead, each error matrix will itself have to be approximately diagonalized, recursively until these errors are driven down sufficiently far in magnitude.
\end{enumerate}

We now discuss at a high level our strategy to address items (1) and (2).
The resolution to item (3) is the most technical element of our proof, and we leave it for later. 
Consider the vector space of all polynomials $f : \{0,1\}^{{n \choose 2}} \times \R^n \rightarrow \R$ which take a graph and an $n$-dimensional real vector and yield a real number.
(We write $f_G(x)$, where $G$ is the graph and $x \in \R^n$.)
If we restrict attention to the subspace of those of degree at most $d$ in $x$, we obtain the polynomials in the domain of our operator $\pE_G$.
If we additionally restrict to the subspace of polynomials which are low degree in $G$, we obtain the family of polynomials so that $\E_G \pE_G f_G(x)$ is calibrated.
Call this subspace $\cV$.

Our goal would to be find an approximate diagonalization for all the non-trivial eigenvalues of $\cM$ using only elements from $\cV$. 
The advantage of doing so is that for every $f \in \cV$, we can calculate $\E_G \pE_G f_G^2$  using the pseudo-calibration condition (\ref{eq:pseudo-calibration}). 
In particular it means that if we find a function $f$ such that  $f_G$ is with high probability an approximate eigenvector of $G$, then we can compute the corresponding expected eigenvalue $\lambda(f)$. 

A crucial tool in finding such an approximate eigenbasis is the notion of \textit{symmetry}. 
For every $f$, if $f'$ is obtained from $f$ via a permutation of the variables $x_1,\ldots,x_n$, then $\E_G \pE_G f_G^2 = \E_G \pE_G f'^2_G$. 
The result of this symmetry, for us, is that our approximate diagonalization requires only of a constant (depending on $d$) number of eigenspaces.
This argument allows us to restrict our attention to a constant number of classes of polynomials, where each class is determined by some finite graph $U$ that we call its \textit{shape}.
For every polynomial $f$ with shape $U$, we compute (approximately) the value of $\E_G \pE_G f_G^2$ as a function of a simple combinatorial property of $U$, and our approximate eigenspaces correspond to polynomials with different shapes. 

We can show that that in expectation our approximate eigenspaces will have non-negative eigenvalues since the pseudo-calibration condition (\ref{eq:pseudo-calibration}) in particular implies that for every $f$ that is low degree in both $G$ and $x$, $\E_G \pE_G f_G^2 \geq 0$. 
However, the key issue is to deal with the error terms that arise from the fact that these are only approximate eigenspaces.
One could hope that, like in other ``structure vs. randomness'' partitions, this error term is small enough to ignore. 
Alas, this is not the case, and we need to handle it recursively, which is the cause of much of the technical complications of this paper.

\begin{quoting}
\begin{remark}[\emph{Structure vs. randomness}]
At a high level our approach can be viewed as falling into the general paradigm of ``structure vs. randomness'' as discussed by Tao~\cite{tao2005dichotomy}. 
The general idea of this paradigm is to separate an object $O$ into a ``structured'' part that is simple and predictable, and a ``random'' part that is unpredictable but has small magnitude or has some global statistical properties.

One example of this is the Szemerédi regularity lemma~\cite{Szemeredi78} as well  variants such as~\cite{FriezeK96} that partition a matrix into a sum of a  low rank and pseudorandom components.
Another example arises from the  random models for the \textit{primes} (e.g., see~\cite{Tao2015lnotes,granville1995cramer}). 
These can be thought of positing that, as far as certain simple statistics are concerned, (large enough) primes can be thought of as being selected randomly conditioned on not being divisible by $2,3,5$ etc.. up to some bound $w$. 

All these examples can be viewed from a computationally bounded Bayesian perspective. 
For every object $O$ we can consider the part of $O$ that can be inferred by a computationally bounded observer to be $O$'s \textit{structured} component, while the remaining uncertainty can be treated as if it is \textit{random}, even  if in actuality it is fully determined. 
Thus in our case, even though for almost every particular graph $G$ from $G(n,1/2,\omega)$, the clique $x$ is fully determined by $G$, we still think of $x$ as having a ``structured'' part which consists of all the inferences a ``simple'' observer can make from $G$ (e.g., that if $i$ and $j$ are non-neighbors then $x_ix_j = 0$), and a ``random'' part that consists of the remaining uncertainty. 
As in other cases of applying this paradigm, part of the technical work is bounding the magnitude (in our case in spectral norm) that arises from the ``random'' part, though as mentioned above in our case we need a particularly delicate control of the error terms which ends up causing much of the technical difficulty. 
\end{remark}
\end{quoting}

\section{Proving Positivity: A Technical Overview} \label{sec:overview-positivity}

We now discuss in more detail how we prove that the  \textit{moment matrix} $\cM$ corresponding to our pseudo-distribution is positive semidefinite. 
Recall that this is the $\binom{n}{\leq d/2} \times \binom{n}{\leq d/2}$ matrix $\cM$ such that $\cM(I,J) = \pE_G \prod_{i\in I}x_i \prod_{j\in J} x_j$  for every pair of subsets $I,J \subseteq [n]$ of size at most $d/2$, and that it is defined via (\ref{eq:overview-pE}) as  
\begin{equation}
\cM(I,J) = \sum_{ \substack{T \subseteq \nchoose{2} \\ |\V(T) \cup I \cup J| \leq \tau}}  \omegan^{|\V(T) \cup I \cup J|} \chi_{T}(G) \;. \label{eq:Mdef}
\end{equation}

The matrix $\cM$ is generated from the random graph $G$, but its entries are \textit{not} independent. Rather, each entry is a polynomial in $\g_e$, and there are some fairly complex dependencies between different them. 
Indeed, these dependencies will create a spectral structure for $\cM$ that is very different from the spectrum of standard random matrices with independent entries and makes proving $\cM$ positive semidefinite challenging.  Our approach to showing that $\cM$ is positive semidefinite is through a type of ``symbolic factorization'' or ``approximate diagonalization,'' which we explain next. 

\subsection{Warm Up} 

It is instructive to begin with the tight analysis presented in \cite{HKP15} of the moments constructed in \cite{MPW15}\footnote{The construction in \cite{MPW15} actually also satisfies $\sum x_i = \omega$ as a constraint which causes the precise form to differ. We ignore this distinction here.}. These moments can in fact obtained by using truncation threshold $\tau = |S|$ in  \eqref{eq:overview-pE}. This choice of $\tau$ is the smallest possible for which the resulting construction satisfies the hard clique constraints. \cite{HKP15} show that this construction satisfies positivity for $\omega \lessapprox n^{1/(\frac{d}{2} + 1)}.$

For the purpose of this overview, let us work with the principal submatrix $F$ indexed by subsets $I$ and $J$ of size exactly $d$.  The analysis in \cite{HKP15} proceeds by first splitting $F$ into $d+1$ components $F = F_0+ F_1 +\cdots+ F_d$ where $F_i(I,J) = F(I,J)$ if $|I \cap J| = i$ and $0$ otherwise. Below, we discuss two of the key ideas involved that will serve as an inspiration for us.

As discussed before, we must approximately diagonalize the matrix $F$ in the sense that the off diagonals blocks must be "small enough" to be charged to the on diagonal block. Thus the main question before us is obtain an (approximate) understanding of the spectrum of $F$ that allows us to come up with a "change of basis" in which the off diagonal blocks are small enough to be charged to the positive eigenmass in the on-diagonal blocks.

Let us consider the piece $F_0$ for our discussion here. As alluded to in Section \ref{sec:overview-positivity}, we want to break $F$ into minimal pieces so that each piece is symmetric under the permutation of vertices. We can hope that each piece will then essentially have a single dominating eigenvalue that can be determined relatively easily. Below, we will essentially implement this plan.

First, we need to decide what kind of "pieces" we will need. These are the \emph{graphical matrices} that we define next.
\begin{definition}[Graphical Matrices (see Def \ref{def:graphical-matrix} for a formal version)]
Let $U$ be a graph on $[2d]$ with specially identified subsets left and right subsets $[d]$ and $[2d] \setminus [d]$.  For any $I,J \in \nchoose{d}$, $I \cap J = \emptyset$, let $\pi_{I,J}$ be an injective map that takes $[d]$ into $I$ and $[2d] \setminus [d]$ into $J$ using a fixed convention. The graphical matrix $M_U$ with graph $U$ is then defined by $M_U(I,J) = \chi_{\pi_{I,J}(U)} (G).$
\end{definition}

The starting point of the analysis is to decompose $F_0 = \sum_{U} \omegan^{2d} M_U,$ where $M_U$ is the graphical matrix with shape $U$. Graphical matrices as above turn out to be the right building blocks for spectral analysis of our moment matrix.  This is because a key observation in \cite{HKP15} shows that a simple combinatorial parameter, the size of the maximum bipartite matching between the left and right index in $U$ (i.e. between $[d]$ and $[2d] \setminus [d]$), determines the spectral norm of $M_U$. Specifically, when $U$ has a maximum matching of size $t < d$, the spectral norm of $M_U$ is $\tilde{O}(n^{d-\frac{t}{2}})$, with high probability. Observe that when $d = 2$ and $U$ is a single edge connecting the left vertex with the right, $M_U$ is just the $\on$-adjacency matrix of  the underlying random graph and it is well known that the spectral norm in this case is $\Theta(\sqrt{n})$ matching the more general claim above. 

In particular, this implies that when $U$ has a perfect matching, $M_U$ is pseudorandom in the sense that $F_U$ essentially has the spectral norm $\approx n^{d/2}$, the same as that of an independent $\on$ random matrix of the same dimensions. This allows $M_U$ to be bounded against the positive eigenvalue $\omegan^{d}$ of the diagonal matrix $F_d$ as $\omegan^{d} \gg \omegan^{2d} n^{d/2}$ (even for $\omega$ approaching $\sqrt{n}$!). However for $M_U$ when $U$ has a maximum matching of size $t < d$, one can't bound against the diagonal matrix $F_d$ anymore. 

The next main idea is to note that for every $M_U$ there's an appropriate "diagonal" against which we must charge the negative eigenvalues of $M_U$. When $U$ has a perfect matching, this is literally the diagonal matrix $F_d$ as done above. However, when, say, $U$ is a (bipartite) matching of size $t < d$, we should instead charge against the "diagonal" matrix that can thought of as obtained by "collapsing" each matching edge into a vertex in $U$. In particular, this collapsing produces a matrix that lies in the decomposition of $F_t$.

\vspace{-0.75pc}
\begin{figure}[h]
\begin{center}
\includegraphics[scale=0.35]{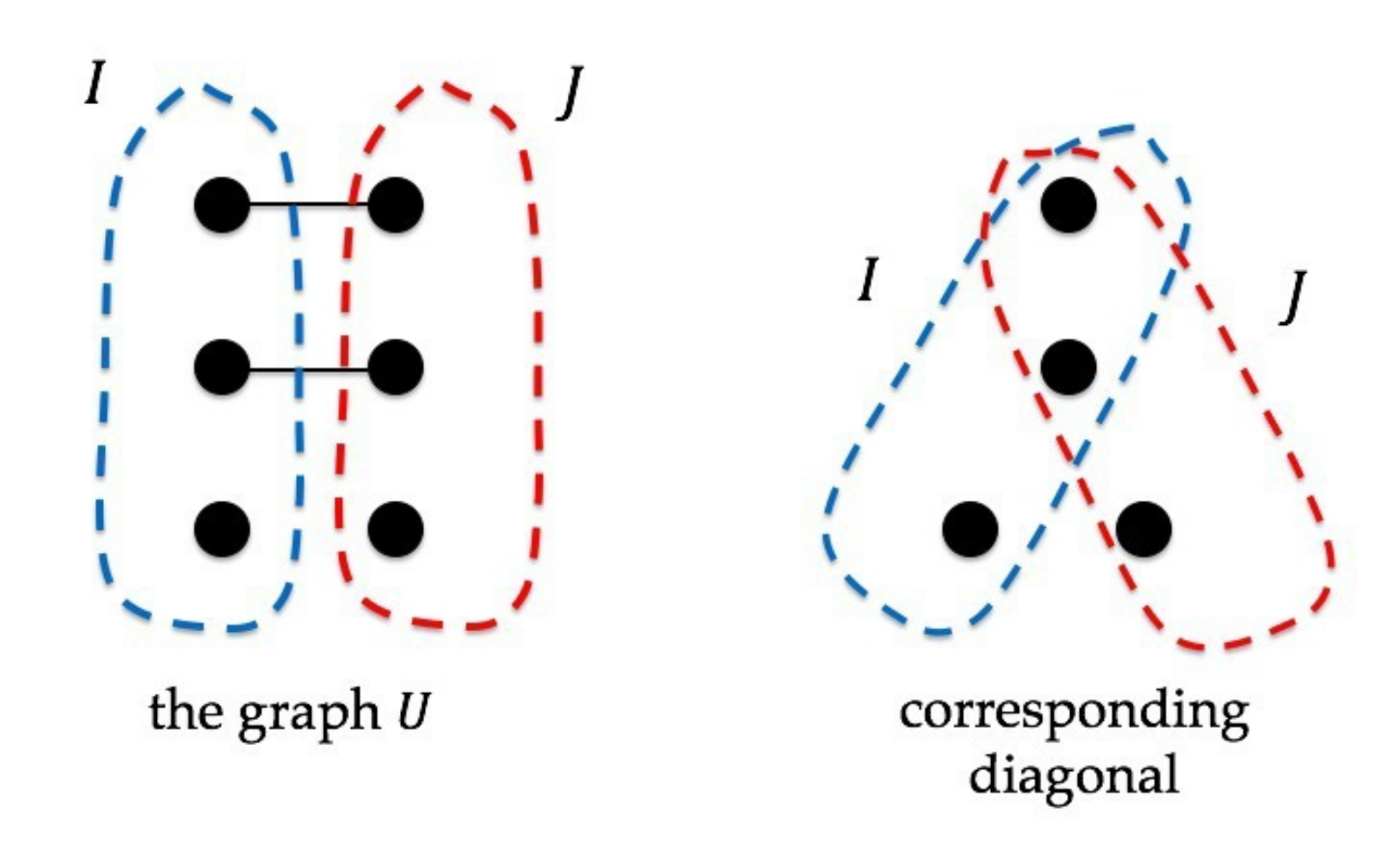}
\end{center}
\end{figure}

\vspace{-0.75pc}


There are a two main takeaways from this analysis that would serve as inspiration in the analysis of our actual construction. First is the decomposition into graphical matrices in order to have a coarse handle on the spectrum of the moment matrix. Second, the "charging" of negative eigenvalues against appropriate "diagonals" is essentially governed by the combinatorics of matchings in $U$. 

\subsection{The Main Analysis} 
We can now try to use the lessons from the warm up analysis to inspire our actual analysis. To begin with, we recall that each graphical matrix was obtained by choosing an appropriate (set of) Fourier monomials for any entry indexed by $I,J$. However, since for our actual construction we have monomials of much higher degree, we need to extend the notion of graphical matrices with \emph{shapes} corresponding to larger graphs $U$. See Def \ref{def:graphical-matrix} for a formal definition.

It turns out that the right combinatorial idea to generalize the size of the maximum matching and control the spectral norm of the graphical matrices $\cM_U$ is the maximum number of \emph{vertex disjoint paths} between specially designated left and right endpoints of $U$ (themselves the generalization of the bipartition we had in the warmup). Using Menger's theorem, this is equal to the size of a minimal collection of vertices that separates the left and right sets in the graph $U$, which we call the \emph{separator size} of $U$.

Finally, we need a "charging" argument to work with the approximate diagonalization we end up with. Generalizing the idea in the warm up here is the hardest part of our proof, but relates again to the notion of vertex separators defined above.
In the warm up, we used a naive charging scheme, breaking the moment matrix into simpler (graphical) matrices, each of which was either a ``positive diagonal'' mass or a ``negative off-diagonal mass'', and pairing up the terms.
Such a crude association doesn't work out immediately in the general setting.
Instead, large groups of graphical matrices must be treated all at once.
In each subspace of our approximate diagonalization of the moment matrix $\cM$, we collect the "positive diagonal mass" and the "negative off digonal mass" that needs to be charged to it together and build an approximately PSD matrix out of it.
As alluded to before, the error in this approximation is not negligible and thus we must further recurse on the error terms. In what follows, we discuss the factorization process that accomplishes the charging scheme implicitly and  the recursive factorization for the error terms in some more detail. 
\newcommand{\ls}{\mathsf{left-sep}}
\newcommand{\rs}{\mathsf{right-sep}}
Consider some graph $T \subseteq \binom{[n]}{2}$, that corresponds to one term in the sum  in (\ref{eq:Mdef})  above, and let $q$ be the minimum size of a set that separates $I$ from $J$ in $T$.
Such a set is not necessarily unique but we can define the \textit{leftmost} separator $\ls(T) = S_\ell$ to be the $q$-sized separator that is closest to $I$ and the \textit{rightmost} separator $\rs(T) = S_r$ to be the $q$-sized separator that is closest to $J$. 

We can rewrite the $(I,J)$ entry moment matrix $\cM$ \eqref{eq:Mdef} by collecting monomials $T$ with a fixed choice of the leftmost and rightmost separators $S_{\ell}$ and $S_r$. This step corresponds to collecting terms with similar spectral norms together accomplishing the goal of collecting together into a term, the "positive diagonal mass" and the "negative off diagonal mass" that are implicitly charged to each other in the intended approximate diagonalization. 

\begin{equation}
\cM(I,J) = \sum_{1 \leq q \leq |I|,|J|} \sum_{S_{\ell}, S_R: |S_{\ell}| = |S_r| = q} \sum_{\substack{ T \subseteq \nchoose{2}\\ |\V(T) \cup I \cup J| \leq \tau \\\ls(T) = S_{\ell}, \rs(T) = S_{r} }} \omegan^{|\V(T) \cup I \cup J|} \chi_{T}(G) \label{eq:temp-factor-1}
\end{equation}

We can then partition $T$ into three subsets $\cR_\ell$, $\cR_m$ and $\cR_r$ that represent the part of the graph $T$ between $I$ and $S_\ell$, the part between $S_\ell$ and $S_r$ and the part between $S_r$ and $J$ respectively (where edges within $S_{\ell}$ and edges within $S_r$ are all placed in $\cR_m$, see Definition~\ref{def:canfact}). 
We thus immediately obtain that $$\chi_T(G) = \chi_{\cR_\ell}(G) \chi_{\cR_m}(G) \chi_{\cR_r}(G) \;.$$

Thus:
\begin{equation} \label{eq:temp-factor-2} 
\cM(I,J) = \sum_{1 \leq q \leq |I|,|J|} \sum_{S_{\ell}, S_R: |S_{\ell}| = |S_r| = q} \!\!\!\! \sum_{\substack{ T \subseteq \nchoose{2}\\ |\V(T) \cup I \cup J| \leq \tau \\\ls(T) = S_{\ell}\\ \rs(T) = S_{r} }} 
\!\!\!\!\! \Paren{ \omegan^{|\V(\cR_{\ell})|} \chi_{\cR_{\ell}}(G)} \Paren{ \omegan^{|\V(\cR_{m})| - 2q} \chi_{\cR_m}(G)} \Paren{ \omegan^{|\V(\cR_{r})|} \chi_{\cR_r}(G)}
\end{equation}

One could hope that we could replace the RHS of (\ref{eq:temp-factor-2}) by

\begin{equation}
\sum_{\substack{1 \leq q \leq |I|,|J| \\ \tau_1+\tau_2+\tau_3\leq \tau}} \; \sum_{\substack{S_\ell \subseteq \binom{[n]}{q} \\ S_r \subseteq \binom{[n]}{q}}} \Paren{\sum_{\substack{\cR_\ell \\  \cV(\cR_\ell) \supseteq I \cup S_\ell  \\ |\V(\cR_\ell)|=\tau_1}} \!\!\!\left(\tfrac{\omega}{n}\right)^{|\cV(\cR_\ell)|} \chi_{\cR_\ell}(G)}
\Paren{ \sum_{\substack{\cR_m \\  \cV(\cR_m) \supseteq S_\ell\cup S_r \\ |\V(\cR_m)|=\tau_2}} \!\!\!\left (\tfrac \omega n\right)^{|\cV(\cR_m)| - 2q} \chi_{\cR_m}(G)} 
\Paren{\sum_{\substack{\cR_r \\ \cV(\cR_r) \supseteq S_r \cup J \\ |\V(\cR_r)| = \tau_3}} \!\!\! \left(\tfrac{\omega}{n}\right)^{|\cV(\cR_r)|}  \chi_{\cR_r}(G)}
\label{eq:Mdef-false}
\end{equation}

In fact, it turns out we can focus attention (up to sufficiently small error in the  spectral norm) to the case $\tau_1 \leq \tau/3$, $\tau_2 \leq \tau/3$, $\tau_3 \leq \tau/3$ in which case if $M(I,J)$ was equal to (\ref{eq:Mdef-false}) we could simply write 
\[
\cM = \sum_q \cL_q \cQ_q \cL_q^\dagger
\]
where for $I,S \subseteq [n]$ with $|I| \leq d$ and $|S|=q$, we let $\cL_q(I,S)$ be the sum of $(\omega/n)^{|V(\cR_\ell)|}\chi_{\cR_\ell}(G)$ over all graphs $\cR_\ell$ of at most $\tau/3$ vertices connecting $I$ to $S$, and for $S,S'$ of size $q$, we let $\cQ_{q}(S,S')$ be the sum of $(\omega/n)^{|\cR_m|-2q}\chi_{\cR_m}(G)$ over all graphs $\cR_m$ of at most $\tau/3$ vertices connecting $S$ to $S'$. 

Thus, in this case, this reduces our task of showing that $\cM$ is positive semidefinite  to showing that for every  $q$, the matrix $\cQ = \cQ_q$ is positive semidefinite.
However the main complication is that there are cross terms in the product $\cL_q \cQ_q \cL_q^\top$ that correspond to repeating the same vertex (not in $S_\ell$ and $S_r$) in more than one of $\cR_\ell$, $\cR_m$ and $\cR_r$. There is no matching term in the Fourier decomposition of $\cM(I, J)$. 
So at best, for every fixed $q$, we can write the part of $\cM$ corresponding to indices $I,J$ with minimal vertex separator equal to $q$ as 
$$\cL \cQ_0 \cL^\top - \cE_1$$
for some error matrix $\cE_1$ that exactly cancels out the extra terms contributed by cross terms with repeated vertices. 
Unfortunately,  the spectral norm of this  error matrix $\cE_1$ is not small enough that we could simply ignore it. 
Luckily however, we can recurse and factorize $\cE_1$ approximately as well. 
We can form a new graph $T'$ by taking the parity of the edge sets in $\cR_\ell$, $\cR_m$ and $\cR_r$. Now we find the leftmost and rightmost separators that separate $I$ and $J$ from each other, and from all repeated vertices. This gives us another decomposition of a graph into three pieces, from which we can write
$$\cE_1 = \cL \cQ_1 \cL^\top - \cE_2$$
for some other matrix $\cQ_1$. 
Continuing this argument gives us for every $q$ a factorization of $\cM_q$ as  
$$  \cL(\cQ_0 - \cQ_1 + \cQ_2 - \ldots - \cQ_{2d-1} + \cQ_{2d})\cL^\top - (\xi_{0} - \xi_{1} + \xi_{2} - \ldots - \xi_{2d-1} + \xi_{2d})$$
The error matrices $\xi_0, \xi_1, \ldots, \xi_{2d}$ arise from truncation issues, which we have ignored in the argument above and turn out to be negligible. 

It is not hard to show that $\cQ_0 \succeq D$ for some positive semidefinite matrix $D$ that we define later. What remains is to bound the remaining matrices $ \cQ_1, \ldots \cQ_{2d-1}$ in order to conclude that $\cM$ is positive semidefinite. Next, we elaborate on the structure of these matrices. 
It turns out that we can define the ``shape'' of a graph $\cR_m$ in an appropriate way so that
$$\cQ^U_i(S_\ell, S_r) = \sum_{ \mbox{shape}(\cR_m)= U}c_i(\cR_m) \chi_{\cR_m}$$
where $U$ is a finite (for constant $d$) sized graph with vertex set $A \cup B \cup C$, where we call $A$ the ``left'' side of $U$ and $B$ the ``right'' side of $U$.  
Moreover $\cQ_i = \sum_U \cQ^U_i$. Now $\cQ^U_i$ is a random matrix and special cases of this general family of matrices (for particular choices of $U$) arise in several earlier works on lower bounds for planted clique. 
Medarametla and Potechin~\cite{MP16} showed that the spectral norm of $\cQ^U$ can be controlled by a bound on its coefficients and a few combinatorial parameters of $U$ \---- namely $|\V(U)|$, $|A \cap B|$ and the number of vertex disjoint paths between $A / B$ and $B / A$. 

A major challenge in our work is to understand and analyze the coefficients $c_i$. 
In the course of decomposing $\cM$, we are able to characterize $c_i(\cR_m)$ as an appropriately weighted sum over $c_{i-1}(\cR_m')$ where $\cR_m'$ ranges over the middle piece of all graphs with leftmost and rightmost separators $S_\ell$ and $S_r$ that could have resulted in $\cR_m$ due to repeated vertices. 
Recall that when there are repeated vertices, we take the parity of the edge sets of the three pieces and compute a new set of left and rightmost vertex separators. 
The set of $\cR_m'$'s that could result in $\cR_m$ is complicated. Instead, our approach is to show that the various combinatorial parameters of $\cR_m'$ (which affect the spectral norm bounds) tradeoff against each other when accounting for the effect of repeated vertices. 
This allows us to bound their contribution and ultimately show that the coefficients $c_i$ decay quickly enough for all values of $\omega < n^{1/2 - \epsilon}$ that we can bound each $\cQ_i$ for $i > 1$ as $ -\frac{D}{8d} \succeq \cQ_i \succeq \frac{D}{8d}$, and this completes our proof.

\section{Preliminaries}
\subsection{General Notation}
\begin{itemize}
\item We use small Greek letters indicate constants/parameters.

\item $\P_d^n$ denotes the linear space of all \emph{multilinear} polynomials of degree at most $d$ on $\zo^n$.

\item We write $\1_Q$ for any event $Q$ to be the $0$-$1$ indicator of whether $Q$ happens.

\item For a subset $T \subseteq \nchoose{2}$ of edges of a graph on vertex set $[n]$, we write $\V(T) \subseteq [n]$ to denote the vertices that have at least one edge incident on them in $T$.

\item For a matrix $Q \in \R^{N \times N}$, $\|Q\| $ denotes its spectral norm (or the largest singular value) and $\|Q\|_{F} = \sqrt{ \sum_{x,y \in [N]} Q(x,y)^2}$ denotes its Frobenius norm.

\item For a graph $G$, let $\cC_{q} = \cC_q(G) = \{I \subseteq [n] \, : \, I \text{ is a $q$-clique in $G$} \}$, and let $\cC_{\leq q} = \bigcup_{q' \leq q} \cC_{d'}$.
Let $\cC(G) = \cC_{\leq \infty}$ be the collection of all cliques in $G$.
We count the empty set and all singletons as cliques.

\item We write $\G(n, \frac{1}{2})$ to denote the distribution on graphs on the vertex set $[n]$ where each edge is included with probability $1/2$ independently of others.

\item We say that an event $E$ with respect to the probability distribution $\G(n, \frac{1}{2})$ happens \emph{with high probability (w.h.p.)} if $\Pr[E] \geq 1- \Omega(1)/n^{10 \log n}$ for large enough $n$.

\item We write $f(n) \ll g(n)$ to mean that for every constant $c$ there is an $n_0$ such that if $n \geq n_0$, $f(n) \leq C g(n)$.

\end{itemize}
\subsection{Graphs}


We identify a graph $G$ with its $\on$ adjacency matrix and write $G_e \in \on$ for the $\on$-indicator of whether $e \in [n] \times [n]$ is an edge (indicated by $G_e = +1$) in the graph $G$ or not.
When $G \sim \G(n, \frac{1}{2})$, $\g_e$ are independent $\on$-random variables. 

A \emph{graph function} is a real-valued function of the variables $\g_e \in \on$ for $e \in \nchoose{2}$. For graphs $G^1, G^2, \ldots, G^k$ on the vertex set $[n]$, we define $\Delta(G^1, G^2, \ldots, G^k)$ to be the graph $G$ satisfying $G_e = \Pi_{i \leq k} G^i_e.$

\begin{definition}[Vertex Separator]
For a graph $G$ on $[n]$ and vertex sets $I, J \subseteq [n]$, a set of vertices $S \subseteq [n]$ is said to be a \textit{minimal vertex separator} if $S$ is a set of smallest possible size such that every path between $I$ and $J$ in $G$ passes through some vertex of $S$. 
\end{definition}

Often, $I$ and $J$ will be allowed to intersect in which case any vertex separator must contain $I \cap J$. 

\begin{fact}[Menger's Theorem] \label{fact:Mengers}
For a graph $G$ on $[n]$ and two subsets of vertices $I, J \subseteq [n]$, the maximum number of vertex disjoint paths between $I$ and $J$ in $G$ is equal to the size of any minimal vertex separator between $I$ and $J$ in $G$.
\end{fact}
\renewcommand{\G}{G}
\subsection{Fourier Analysis}
Any graph function $f:G \rightarrow \R$ can be represented as a Fourier polynomial in the variables $\g_e$: $$f(\g) = \sum_{W \subseteq \nchoose{2}} \hat{f}(W) \chi_W(\g),$$  where $\chi_W(\g)$ is the \emph{parity} function on edges in $W$: $$\chi_W(\g) = \Pi_{e \in W} \g_e.$$ The parity function $\chi_W$ are an orthonormal basis for functions on $G$ under the inner product defined by $\iprod{f,h} = \E_{G \sim \G(n, \frac{1}{2})}[f(G)h(G)]$ for any graph functions $f$ and $h$.

The following fact is easy to verify:
\begin{fact}
Let $G$ be a graph on $n$ described by the vector $\g \in \on^{{n \choose 2}}$. For any subset $S \subseteq [n]$ of the vertices, we have the identity:
\[
\sum_{W \subseteq {S \choose 2}} \chi_W(G) = \begin{cases}
2^{{{|S|} \choose 2}} & \text{ if $S$ is a clique in $G$,}\\
0 & \text{ otherwise. }
\end{cases}
\] \label{fact:Fourier-and}
\end{fact}



\subsection{The Sum-of-Squares Algorithm} 

The sum of squares algorithm has several equivalent definitions.
We follow the notation of \emph{pseudoexpectations} as in the survey of Barak and Steurer~\cite{BarakS14}.

\begin{definition}[Pseudoexpectation]
A linear operator $\pE : \P_d^n \rightarrow \R$ is said to be a \emph{degree} $d$-\emph{pseudoexpectation} if it satisfies:
\begin{enumerate}
\item Normalization: $\pE[\bm{1}] = 1$.
\item Positive Semidefiniteness: $\pE[ p^2] \geq 0$ for every polynomial $p \in \P_d^n$.
\end{enumerate}
A pseudoexpectation operator $\pE$ on $\P_d^n$ is said to satisfy a constraint $\{p = 0\}$ for any $p \in \P_d^n$ if for every polynomial $q \in \P_d^n$ such that $p \cdot q \in \P_d^n$, $\pE[ pq] = 0$. 
\end{definition}

Given a set of constraints $\{p_i = 0\}$ for $1 \leq i \leq m$ and an objective polynomial $p$, degre sum of squares algorithm of degree $d$ solves the problem $$\arg \max \pE[p]$$ over all degree $d$ pseudoexpectations $\pE$ that satisfy $\{p_i = 0\}$ for $1 \leq i \leq m$.

\section{The Pseudo-expectation}
\label{sec:def-intro}

We now define our pseudo-distribution operator $\pE_G$. 
As discussed in Section~\ref{sec:intro:pseudo-dist}, it is based on requiring (\ref{eq:pseudo-calibration}) to hold for every $f$ that has degree at most $\tau$ in $G$ and $d$ in $x$. 


\begin{center}
  \fbox{\begin{minipage}{\textwidth}
  \paragraph{Important Parameters}
The following parameters will be fixed for the rest of the paper.  \begin{itemize}
  \item $\epsilon \in (0,1/2)$, which determines the size $\omega = n^{1/2 - \epsilon}$ of the planted clique.
  \item $d = d(n) \in \N$, the degree of the SoS relaxation against which we prove a lower bound.
  \item $\tau = \tau(n) \in \N$, the degree of our pseudoexpectation $\pE$ as a function of $G \sim \G(n,1/2)$.
  \end{itemize}
  We always assume that $Cd/\epsilon \leq \tau \leq (\epsilon/C) \log n$ and $\epsilon \geq C \log \log n / \log n$ for a sufficiently-large constant $C$.
  Eventually we will set $d = (\epsilon/C)^2 \log n$, (this yields the parameters stated in Theorem \ref{thm:main}, since then $n^{1/2 - \epsilon} = n^{1/2 - \Omega(d/\log n)^{1/2}}$), which implies that $\epsilon \gg \log \log n / \log n$.
  \end{minipage}}
\end{center}

\subsection{Definition of $\pE$}
As discussed previously, $\pE$ is completely specified by its \emph{multilinear moments}: $\pE[ x_I]$ for $I \subseteq [n]$ and $|I| \leq d$. $\pE[x_I]$ is a function of $\g_e$ for $e \in \nchoose{2}$ and can be written as a polynomial in $\g_e$ with coefficients $\hat{\pE[x_S]}(T)$ for each $T \subseteq \nchoose{2}$ (the "Fourier coefficients"). These Fourier coefficients will be fixed by our insistence on the pseudoexpectation being pseudocalibrated with respect to the planted distribution $G(n, 1/2,\omega)$.

\begin{definition}[$\pE$ of degree $d$, clique-size $\omega$, truncation $\tau$]
  Let $S \subseteq [n]$ be a set of vertices of size $|S| \leq d$.
  Let $T \subseteq \nchoose{2}$ be a set of edges.
  Let $\chi_T = \prod_{e \in T} \g_e$.
  Let
  \[
    \widehat{\pE[x_S]}(T) = \begin{cases}
        \E_{(G,x) \sim \G(n,1/2,\omega)}[\chi_T(G) x_S] & \quad \text{if $|\V(T) \cup S| \leq \tau$}\\
        0 & \quad \text{otherwise}\mper
        \end{cases}
  \]
  As usual, $\pE[x_S] = \sum_{T \subseteq \nchoose{2}} \widehat{\pE[x_S]}(T) \cdot \chi_T(G)$.
\end{definition}


The Fourier coefficients can in fact be explicitly computed easily:
\begin{lemma} \label{lem:Fourier-pseudo-expectation}
Let $T \subseteq \nchoose{2}$, $S \subseteq [n]$ and $\V(T) \subseteq [n]$ be the vertices incident to edges in $T$. Then $$ \E_{(H,x) \sim \G(n,1/2,\omega)} [\chi_T \cdot x_S] = \left(\tfrac{\omega}{n}\right)^{|\V(T) \cup S|}.$$
\end{lemma}
\begin{proof}
Throughout this proof, we suppress explicit notation for the underlying random variable which is $(H,x) \sim \G(n, \frac{1}{2},\omega)$.
We claim that $\E[\chi_T \cdot x_S] = \Pr[ x_{\V(T)\cup S} = 1 ]$.
To see this, note that
\begin{multline}
\E[\chi_T \cdot x_S] = \Pr[ x_{\V(T) \cup S}=1 ] \cdot \E[ \chi_T \cdot x_S \, | \, x_{\V(T) \cup S} =0 ]\\
+ (1 - \Pr[ x_{\V(T) \cup S} =1]) \cdot \E[ \chi_T \cdot x_S \, | \, x_{\V(T) \cup S}=0 ].
\end{multline} We note that the second term above is $0$. It's easy to see if $x_S =0 $. Otherwise, $x_{\V(T)}=0$, and there is an edge $e  \in T$ but not contained in the clique $x$. Thus, 
  \[
    \E[ \chi_e \chi_{T \setminus e} \cdot x_S \, | \, x_{\V(T) \cup S} =0 ] = 0\mper
  \]

If $x_{\V(T) \cup S}=1$ then $\chi_T = 1$, so $\E[ \chi_T \cdot x_S \, | \, x_{\V(T) \cup S}=1 ] = 1$. By a simple computation,
  \[
    \Pr [x_{V(T) \cup S}=1 ] = \left(\tfrac{\omega}{n}\right)^{|\V(T) \cup S|}\mper\qedhere
  \]
 \end{proof}

As discussed in Section \ref{sec:provingpositivity}, our construction of $\pE$ is pseudocalibrated. The following lemma captures this formally.
We include the (straightforward) proof in Appendix~\ref{sec:calibration}.
\begin{lemma}\torestate{
    \label{lem:pE-fools-simple-tests}
  Let $f_G(x)= \sum_{|S| \leq 2d} c_S(G) \cdot x_S$ be a real-valued polynomial on $\zo^n$ whose coefficients have degree at most $\tau$ when expressed in the $\pm 1$ indicators $\g_e$ for edges in $G$.
  Then, $\E_{G \sim \G(n, \frac{1}{2})}[ \pE[f_G(x)]] = \E_{(H,x) \sim \G(n,1/2,\omega)}[f_H(x)]$. 
  }
\end{lemma}

\subsection{$\pE$ Satisfies Constraints}
We now show that the $\pE$ defined in the previous section satisfies all linear constraints among (1) -- (6) in Section~\ref{sec:intro:pseudo-dist} and has an objective value of $\omega$. 
That is, 1) $\pE[1] \approx 1$, 2) $\pE[\sum_{i \in [n]} x_i] \approx \omega$, and 3) $\pE[x_S] = 0$ for every $S \subseteq [n]$ which is not a clique in $G$.


We analyze $\pE[1]$ and $\pE[ \sum_{i \in [n]} x_i]$ in the next lemma and include a proof based on moment-method in Appendix~\ref{sec:normalization}.
\begin{lemma}\torestate{\label{lem:normalization}
  With high probability, $\pE[1] = 1 \pm n^{-\Omega(\epsilon)}$ and $\pE[\sum_{i \in [n]} x_i] = \omega \cdot (1 \pm n^{-\Omega(\epsilon)})$.}
\end{lemma}

The next lemma shows that $\pE[x_S] = 0$.
\begin{lemma}\torestate{\label{lem:clique-constraints}
With probability $1$, if $S \subseteq [n]$ of size at most $d$ is not a clique in $G$, then $\pE[x_S] = 0$.}
\end{lemma}
\begin{proof}
  Let $S \subseteq [n]$ have size at most $d$.
  Recall that $\1_{S \text{ is a clique in } G} = 2^{-{|S| \choose 2}} \sum_{T \subseteq {S \choose 2}} \chi_T$.
  Becasue the Fourier expansion of $\pE[x_S]$ is truncated using the threshold $|\cV(T) \cup S| \leq \tau$, two Fourier characters $\chi_T, \chi_{T'}$ have the same coefficient in $\pE[x_S]$ if $T \xor T' \subseteq {S \choose 2}$.
  So we can factor $\pE[x_S] = \1_{S \text{ is a clique in } G} \cdot f_S(G)$ for some function $f_S$.
\end{proof}

\subsection{Proof of Main Theorem}
Our main technical claim is that $\pE=\pE_G$ is (approximately) PSD. That is:
\begin{lemma}\torestate{\label{lem:PSD-main}
  With high probability over $G$ from $G(n,1/2)$, every $p \in \P_{d}$ satisfies,
\[
  \pE_G[p(x)^2] \geq 0
\]}
\end{lemma}
It is easy to complete the proof of Theorem \ref{thm:main} now:


\begin{proof}[Proof of Theorem~\ref{thm:main}]
  By Lemma~\ref{lem:normalization}, Lemma~\ref{lem:clique-constraints}, and Lemma~\ref{lem:PSD-main}, there is a universal $C$ so that if $Cd/\epsilon \leq \tau \leq (1/C) \epsilon \log n$, (by a union bound) with high probability the following all hold:
  \begin{enumerate}
    \item $\pE[1] = 1 \pm n^{-\Omega(\epsilon)}$.\label{itm:1}
    \item $\pE[x_S] = 0$ for every $S$ of size at most $d$ not a clique in $G$.\label{itm:cliques}
    \item $\pE[\sum_i x_i] \geq (1 - n^{-\Omega(\epsilon)}) \omega$.\label{itm:objective}
    \item $\pE[p(x)^2] \geq 0$ for every $p \in \P_{d}$.\label{itm:psd}
  \end{enumerate}
  Thus, choose $\epsilon = (C^2 d /\log n)^{1/2}$ and $\tau = (1/C) \epsilon \log n$.
  The operator given by $\pE^*[p(x)] = \pE[p(x)]/\pE[1]$ is a valid degree-$d$ pseudo-distribution with $\pE[\sum_i x_i] \geq \Omega(n^{1/2 - \Theta(d/\log n)^{1/2}})$ as desired.
  
\subsection{Proof Plan} 
As is standard, we can reduce Lemma \ref{lem:PSD-main} to showing that the associated \emph{moment matrix}, is positive semidefinite. 
\begin{definition}[Moment Matrix] \label{def:moment-matrix}
Let $\cM \in \R^{\nchoose{\leq d} \times \nchoose{\leq d}}$ be given by $\cM(I,J) = \pE[x_I x_J]$.
\end{definition}

Thus, Lemma \ref{lem:PSD-main} is equivalent to showing:
\begin{lemma}\label{lem:PSD-main2}
  With high probability, $\cM \succeq 0.$
\end{lemma}

At a high level our plan involves first getting an approximate factorization of the moment matrix $\cM = \cL \cQ_{0} \cL^{\transposed} +"error"$ for appropriately defined matrices $\cL$ and $\cQ_{0}$. This step is the key technical part of the proof - given such a factorization, our task reduces to showing that $\cQ_{0}$ and $\cL \cL^{\transposed}$ has large enough positive eigenvalues to compensate for the error. The first approximate factorization step will occupy us in Section \ref{sec:factorization}. The technical work in second step involves showing upper bounds on the spectral norms of appropriately defined pieces of $\cQ_0$ and is the content of Section \ref{sec:spectral-norms}. 



\newcommand{\Shape}{\mathsf{shape}}



\end{proof}

%
%

\section{Approximate Factorization of the Moment Matrix}
\label{sec:factorization}

\subsection{Ribbons and Vertex Separators}
In this section we get set up for the first step in the proof of Lemma~\ref{lem:PSD-main2} by setting up some definitions.
\emph{Ribbons} will play a crucial role in our analysis:

\begin{definition}[Ribbon]
An $(I,J)$-ribbon $\cR$ is a graph with edge set $W_{\cR} \subseteq \nchoose{2}$ and vertex set $V_{\cR} \supseteq \V(W_{\cR}) \cup I \cup J$, for two specially identified subsets $I,J \subseteq [n]$, each of size at most $d$, called the \emph{left} and the \emph{right ends}, respectively. We sometimes write $\V(\cR) \eqdef V_{\cR}$ and call $|\V(\cR)|$ the \emph{size} of $\cR$. Also, we write $\chi_{\cR}$ for the monomial $\chi_{W_{\cR}}$ where $W_{\cR}$ is the edge set of the ribbon $\cR$. 
\end{definition}

 In our analysis, $(I,J)$-ribbons arise as the terms in the Fourier decomposition of the entry $\cM(I,J)$ in the moment matrix. It is important to emphasize that the subsets $I$ and $J$ in an $(I,J)$-ribbon are allowed to intersect.
Also $\V(\cR)$ can contain vertices that are not in $\V(W_\cR)$ if there are isolated vertices in the ribbon.

Ultimately, we will want to partition a ribbon into three subribbons in such a way that we can express the moment matrix as the sum of positive semidefinite matrices, and some error terms. Our partitioning will be based on minimum vertex separators. 

\begin{definition}[Vertex Separator]
For an $(I,J)$-ribbon $\cR$ with edge set $W_{\cR}$, a subset $Q \subseteq \V(\cR)$ of vertices is a \emph{vertex separator} if $Q$ separates $I$ and $J$ in $W_{\cR}$. A vertex separator is \emph{minimum} if there are no other vertex separators with strictly fewer vertices. The \emph{separator size} of $\cR$ is the cardinality of any minimum vertex separator of $\cR$.
\end{definition}

The following elementary lemma establishes that a ribbon has a unique \emph{leftmost} and \emph{rightmost} vertex separator of minimum size. We defer its proof to Appendix~\ref{sec:ribbons-proofs}.

\begin{lemma}[Leftmost/Rightmost Vertex Separator]\torestate{\label{lem:left-right-sep}
Let $\cR$ be an $(I,J)$-ribbon. There is a unique minimum vertex separator $S$ of $\cR$ such that $S$ separates $I$ and $Q$ for any vertex separator $Q$ of $\cR$. We call $S$ the \emph{leftmost} separator in $\cR$. We define the \emph{rightmost} separator analogously and we denote them by $S_L(\cR)$ and $S_R(\cR)$ respectively.}
\end{lemma}

We illustrate the notion of a leftmost and rightmost vertex separator in the example below. 
 \begin{figure}[h]
\begin{center}
\includegraphics[scale = 0.5]{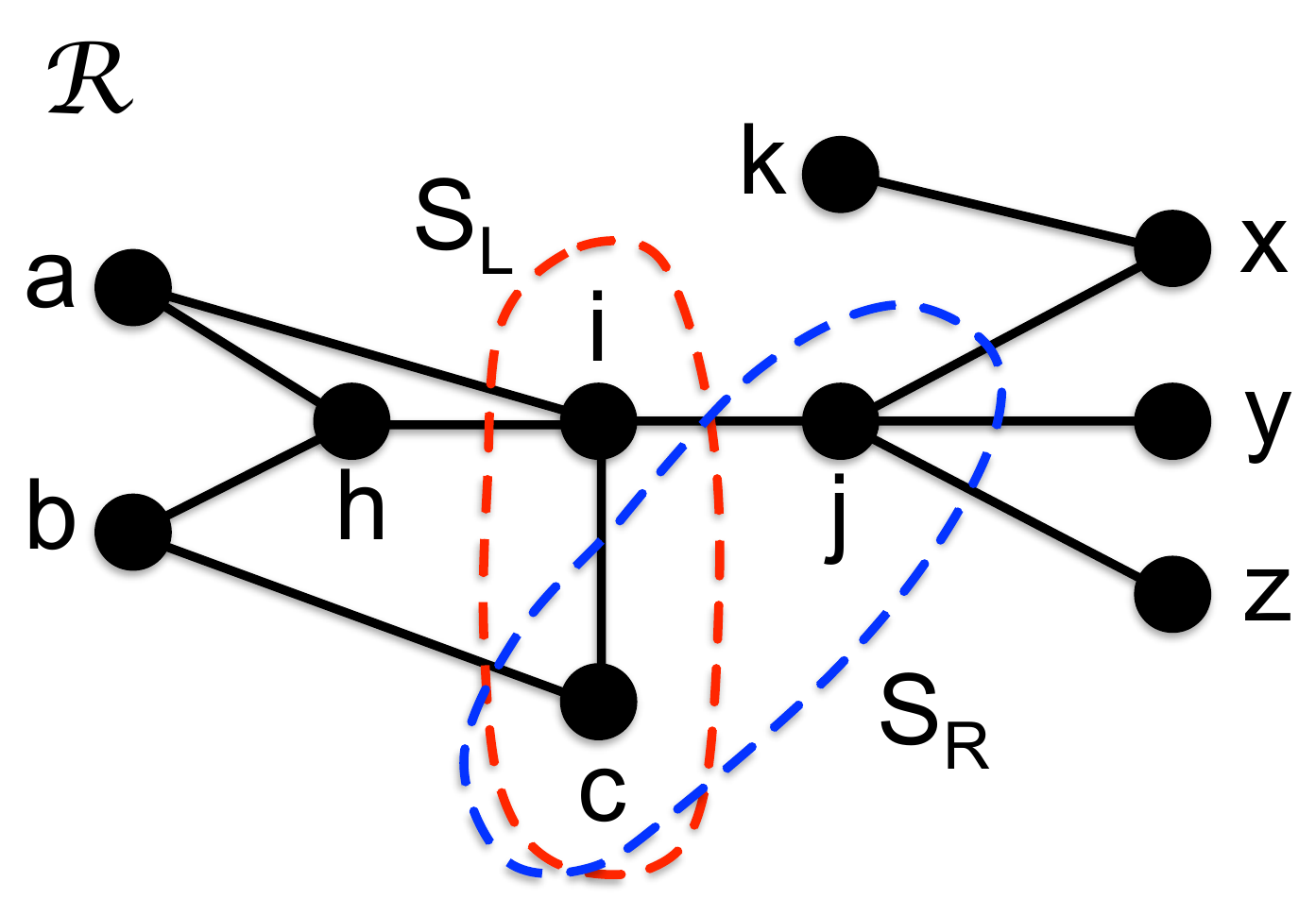}
\end{center}
\end{figure}

\noindent Let $I = \{a, b, c\}$ and let $J = \{c, x, y, z\}$. The maximum number of vertex disjoint paths from $I$ to $J$ is $2$ \---- for example, we could take the path $\{c\}$ and the path $\{b, h, i, j, z\}$. 
The leftmost and rightmost separators are $S_L = \{c, i\}$ and $S_R = \{c, j\}$ respectively. This example illustrates an important point that when $I$ and $J$ intersect, $S_L$ and $S_R$ must both contain $I \cap J$.

\subsection{Factorization of Monomials}
Our factorization of $\cM$ will rely on an iterative argument for grouping and factoring the Fourier characters in the decomposition of $\cM(I,J)$.

\begin{definition}[Canonical Factorization]\label{def:canfact}
  Let $\cR$ be an $(I,J)$-ribbon with edge set $W_{\cR}$ and vertex set $V_\cR$.
  Let $V_\ell$ be the vertices reachable from $I$ without passing through $S_L(\cR)$, and similarly for $V_r$, and let $V_m = V_\cR \setminus (V_\ell \cup V_r)$.
  Let $W_\ell \subseteq W_{\cR}$ be given by $$W_\ell = \{ (u,v) \in W_{\cR} \, : \, u \in V_\ell  \mbox{ and } v \in V_\ell \cup S_L \}$$
  and similarly for $W_r$.
  Finally, let $W_m = W_{\cR} \setminus (W_\ell \cup W_r)$.

  Let $\cR_\ell$ be the $(I, S_L(\cR))$-ribbon with vertex set $V_\ell \cup S_L(\cR)$ and edge set $W_\ell$ and similarly for $\cR_r$.
  Let $\cR_m$ be the $(S_L(\cR), S_R(\cR))$-ribbon with vertex set $V_m$ and edge set $W_m$ .
  The triple $(\cR_\ell, \cR_m, \cR_r)$ is the \emph{canonical factorization} of $\cR$.
\end{definition}

Some facts about the canonical factorization are worth emphasizing. First, $W_\ell, W_m$ and $W_r$ are disjoint and are a partition of $W_\cR$ by construction. Hence $\chi_{\cR} = \chi_{W_\ell} \cdot \chi_{W_m} \cdot \chi_{W_r}$. Second, some vertices in $I$ may not be in $V_\ell$ at all. However any such vertices that are in $I$ but not $V_\ell$ are necessarily in $S_L$ and thus will be contained in $\cR_\ell$ anyways. This is why we can say that $\cR_\ell$ is an $(I, S_L(\cR))$-ribbon. The following illustrates what the canonical factorization would look like in our earlier example:

 \begin{figure}[h]
\begin{center}
\includegraphics[scale = 0.5]{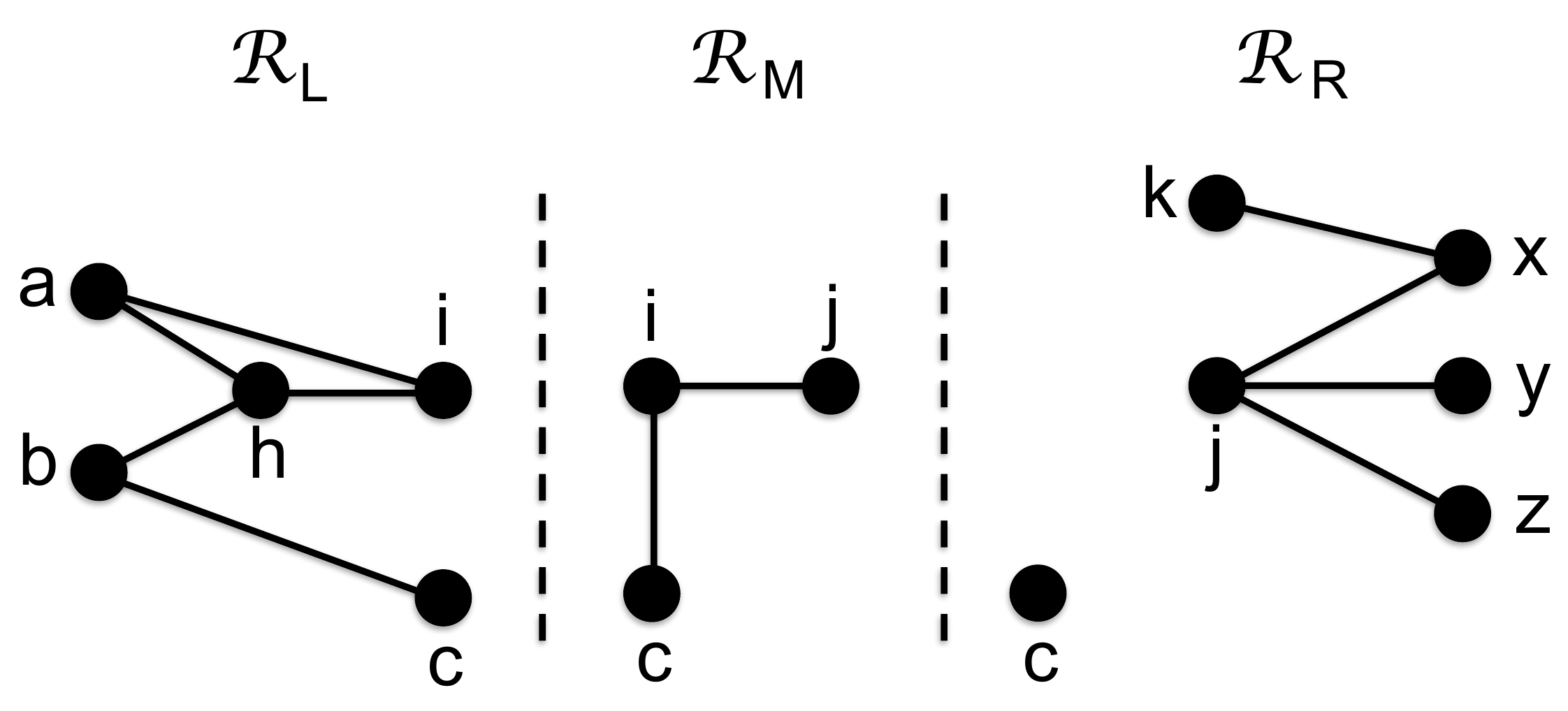}
\end{center}
\end{figure}

We chose this example to illustrate a subtle point. The edge $(i, c)$ has both its endpoints in both $\cR_\ell$ and $\cR_m$. We could in principle choose to place it in either, but we have adopted the convention that because both of its endpoints are in $S_L$ we place it in $\cR_m$. In this way, there are no edges within $S_L$ in $\cR_\ell$ or within $S_R$ in $\cR_m$. Finally, note that there can be isolated vertices in $\cR_\ell$ or $\cR_r$ but such vertices need to be in $I$ or $J$ respectively.

 With the definition of the canonical factorization in hand, we will collect some important properties about it that we will make use of later:

\begin{claim}\label{clm:vertex-factorization}
  Let $\cR$ be an $(I,J)$-ribbon with canonical factorization $(\cR_\ell, \cR_m, \cR_r)$.
  Then $$|\cV(\cR)| = |\cV(\cR_\ell)| + |\cV(\cR_m)| + |\cV(\cR_r)| - |S_L(\cR)| - |S_R(\cR)|.$$
\end{claim}
\begin{proof}
It is important to note that $S_L(\cR)$ and $S_R(\cR)$ are not necessarily disjoint (indeed, this happens in the example above). Nevertheless, we know that by construction $V_\ell$, $V_m$ and $V_r$ are disjoint and that $S_L(\cR) \cup S_R(\cR) \subseteq V_m$. Every vertex that appears just once in $S_L(\cR)$ and $S_R(\cR)$ appears twice in the canonical factorization. And every vertex that is in $S_L(\cR) \cap S_R(\cR)$ appears three times. Thus
$$|\cV(\cR)| = |\cV(\cR_\ell)| + |\cV(\cR_m)| + |\cV(\cR_r)| - |S_L(\cR) / S_R(\cR)| - |S_R(\cR)/ S_L(\cR)| - 2 |S_L(\cR) \cap S_R(\cR)|$$
which completes the proof. 
\end{proof}

In the discussion above, we established some properties that a canonical factorization must satisfy. Next we show the reverse direction, that any collection of ribbons that satisfies the below properties must be a canonical factorization. Consider a collection of ribbons $\cR_0, \cR_1, \cR_2$, and the following list of properties:

\begin{center}
  \fbox{\begin{minipage}{\textwidth}
      \paragraph{$S_\ell, S_r$ Factorization Conditions for $\cR_0, \cR_1, \cR_2$ (Here $S_\ell, S_r \subseteq [n]$.)}
      \begin{enumerate}
    \item $\cR_0$ is an $(I,S_\ell)$-ribbon with $S_L(\cR_0) = S_R(\cR_0) = S_\ell$, and all vertices in $\V(\cR_0)$ are either reachable from $I$ without passing through $S_\ell$ or are in $I$ or $S_\ell$.
    Finally, $\cR_0$ has no edges between vertices in $S_\ell$. \label{itm:left}
    \item $\cR_2$ is an $(S_r,J)$-ribbon with $S_L(\cR_2) = S_R(\cR_2) = S_r$,
    and all vertices in $\V(\cR_2)$ are either reachable from $J$ without passing through $S_r$ or are in $J$ or $S_r$.
    Finally, $\cR_2$ has no edges between vertices in $S_r$. \label{itm:right}
    \item $\cR_1$ is an $(S_\ell, S_r)$-ribbon with $S_L(\cR_1) = S_\ell$ and $S_R(\cR_1) = S_r$. Every vertex in $\cV(\cR_1) \setminus (S_\ell \cup S_r)$ has degree at least $1$. \label{itm:middle}
    \item $W_{\cR_0}, W_{\cR_1}, W_{\cR_2}$ are pairwise disjoint. Also, $V_{\cR_0} \cap V_{\cR_1} = S_\ell, V_{\cR_1} \cap V_{\cR_2} = S_r$, and $V_{\cR_0} \cap V_{\cR_2} = S_\ell \cap S_r$. \label{itm:disjoint}
      \end{enumerate}
  \end{minipage}}
\end{center}

\begin{lemma}
  Let $\cR_0, \cR_1, \cR_2$ be ribbons.
  Then $(\cR_0, \cR_1, \cR_2)$ is the canonical factorization of the $(I,J)$-ribbon $\cR$ with edge set $W_{\cR_0} \oplus W_{\cR_1} \oplus W_{\cR_2}$ and vertex set $\V(\cR_0) \cup \V(\cR_1) \cup \V(\cR_2)$ if and only if the $S_\ell, S_r$ factorization conditions hold for $\cR_0, \cR_1, \cR_2$ for some $S_\ell, S_r \subseteq [n]$.
\end{lemma}
\begin{proof}
  If $\cR$ is a ribbon with leftmost and rightmost vertex separators $S_\ell$ and $S_r$ and canonical factorization $(\cR_0, \cR_1, \cR_2)$, then 
  many of the conditions above are automatically satisfied. By construction, $W_{\cR_0}, W_{\cR_1}, W_{\cR_2}$ are pairwise disjoint. Because any edge with both endpoints in $S_\ell$ is included in $\cR_m$ we have that there are no edges between vertices in $S_\ell$ in $\cR_0$, and similarly for $\cR_2$. Finally suppose there is a vertex $u$ in $\cR_0$. If $u$ is not reachable from $I$ without passing through $S_\ell$ and is not in $I$ or $S_\ell$ then it would not be included in $\cR_0$. An identical argument holds for $\cR_2$. 
  
All that remains is to verify that $S_L(\cR_0) = S_R(\cR_0) = S_\ell$ and similarly for $\cR_1, \cR_2$.
  If $S_\ell = S_L(\cR)$ is not a minimum-size vertex separator for $\cR_0$, then it is also not a minimum-size vertex separator for $\cR$, which is impossible.
  Similarly, if it is not the leftmost separator for $\cR_0$ then it was not the leftmost separator for $\cR$.
  Since $\cR_0$ is an $(I,S_\ell)$-ribbon and $S_\ell$ is a minimum-size separator, it must also be the right-most minimum-size separator.
  
  Now in the reverse direction, suppose that $\cR_0, \cR_1, \cR_2$ are ribbons that meet the $S_\ell, S_r$ factorization conditions. We claim that $S_\ell$ is the leftmost separator for $\cR$. If not, then either their is a smaller vertex separator, or there is a vertex separator $S'_\ell$ of the same size that separates $I$ and $S_\ell$. To rule out the former case, note that since $S_\ell$ and $S_r$ are both minimum vertex separators for $\cR_1$, we must have $|S_\ell| = |S_r|$. Then it follows from the $S_\ell, S_r$ factorization conditions that there are $|S_\ell|$ vertex disjoint paths from $I$ to $J$, but this would contradict the fact that there is a vertex separator with fewer than $|S_\ell|$ vertices. In the latter case, any other vertex separator $S'_\ell$ of the same size that separates $I$ and $S_\ell$ would contradict the condition $S_L(\cR_0) = S_\ell$. An identical argument shows that $S_r$ is the rightmost separator for $\cR$. 
  
  Finally, by assumption all the vertices in  $\V(\cR_0)$ are either reachable from $I$ without passing through $S_\ell$ or are in $I$ or $S_\ell$ and hence would be included in $\cR_0$. Similarly, there are no edges in $W_{\cR_0}$ with both endpoints in $S_\ell$. Thus if we were to compute the canonical factorization for $\cR$ we would get the same set of vertices in each ribbon and the same partition of the edges. 
\end{proof}

\subsection{Factorization of Matrix Entries}
This leads to our first factorization of the entries $\cM(I,J)$ of $\cM$.
Unfortunately, the error terms in this first attempt will be too large.
Using canonical factorizations and Claim~\ref{clm:vertex-factorization}, for any $I,J\subseteq [n]$ of size at most $d$ we can write
\begin{align}
  \cM(I,J) & = \sum_{\substack{ \cR \text{ an $(I,J)$-ribbon with edge set $W$,}\\ |\V(W)| \leq \tau \\ \text{canonical factorization $(\cR_\ell,\cR_m,\cR_r)$}}}  \Paren{\frac \omega n}^{|\V(\cR)|} \cdot \chi_{\cR_\ell} \cdot \chi_{\cR_m} \cdot \chi_{\cR_r} \nonumber \\
           & = \sum_{\substack{S_\ell, S_r \subseteq [n]\\|S_\ell| = |S_r| \leq d}} \Paren{\frac \omega n}^{-\frac{|S_\ell| + |S_r|}{2}} \sum_{\substack{\cR_\ell, \cR_m, \cR_r \subseteq \nchoose{2}\\ \text{ satisfying $S_\ell, S_r$ factorization conditions}\\ \text{and } |\V(\cR_\ell) \cup \V(\cR_m) \cup \V(\cR_r)| \leq \tau}} \Paren{\frac \omega n}^{|\cV(\cR_\ell)| + |\cV(\cR_m)| + |\cV(\cR_r)| -\frac{|S_\ell| + |S_r|}{2}} \cdot \chi_{\cR_\ell} \cdot \chi_{\cR_m} \cdot \chi_{\cR_r} \nonumber
   \intertext{Notice that except for the disjointness condition, the $S_\ell, S_r$ factorization conditions can be separated into condition~\ref{itm:left} for $\cR_\ell$, condition~\ref{itm:middle} for $\cR_m$, and condition~\ref{itm:right} for $\cR_r$. We use this to rewrite as}
  & = \sum_{\substack{S_\ell, S_r \subseteq [n]\\ |S_\ell| = |S_r| \leq d}} \Paren{\frac \omega n}^{- \frac{|S_\ell| + |S_r|}{2}}
  \Paren{\sum_{\substack{\cR_\ell \text{ having \ref{itm:left}} \\ |\V(\cR_\ell)| \leq \tau}} \Paren{\frac \omega n}^{|\V(\cR_\ell)|} \chi_{\cR_\ell} }
  \Paren{\sum_{\substack{\cR_m \text{having \ref{itm:middle}} \\ |\V(\cR_m)| \leq \tau}} \Paren{\frac \omega n}^{|\V(\cR_m)| - \frac{|S_\ell| + |S_r|}{2}} \chi_{\cR_m} }
  \Paren{\sum_{\substack{\cR_r \text{having \ref{itm:right}} \\ |\V(\cR_r)| \leq \tau}} \Paren{\frac \omega n}^{|\V(\cR_r)|} \chi_{\cR_r} } \label{eq:factor-0} \\
  & -  \underbrace{\sum_{\substack{S_\ell, S_r \subseteq [n]\\ |S_\ell| = |S_r| \leq d}} \Paren{\frac \omega n}^{-\frac{|S_\ell| - |S_r|}{2}} \sum_{\substack{\cR_\ell, \cR_m, \cR_r \\ \text{ satisfying $S_\ell, S_r$ conditions}\\ |\V(\cR_\ell)|, |\V(\cR_m)|, |\V(\cR_r)| \leq \tau,  \\ |\V(\cR_\ell) \cup \V(\cR_m) \cup \V(\cR_r)| > \tau}} \Paren{\frac \omega n}^{|\V(\cR_\ell)| + |\V(\cR_m)| + |\V(\cR_r)| - \frac{|S_\ell| + |S_r|}{2}} \cdot \chi_{\cR_\ell} \cdot \chi_{\cR_m} \cdot \chi_{\cR_r}}_{\defeq \xi_0(I,J), \text{ the error from ribbon size}} \label{eq:factor-1}\\
  & - \underbrace{\sum_{\substack{S_\ell, S_r \subseteq [n]\\ |S_\ell| = |S_r| \leq d}} \Paren{\frac \omega n}^{-\frac{|S_\ell| - |S_r|}{2}} \sum_{\substack{\cR_\ell, \cR_m, \cR_r \text{ satisfying}\\\text{\ref{itm:left},\ref{itm:middle},\ref{itm:right} and not \ref{itm:disjoint}}  \\
|\V(\cR_\ell)|, |\V(\cR_m)|, |\V(\cR_r)| \leq \tau}} \Paren{\frac \omega n}^{|\V(\cR_\ell)| + |\V(\cR_m)| + |\V(\cR_r)| - \frac{|S_\ell| + |S_r|}{2}} \cdot \chi_{\cR_\ell} \cdot \chi_{\cR_m} \cdot \chi_{\cR_r}}_{\defeq E_0(I,J), \text{ the error from ribbon nondisjointness}}\label{eq:factor-2}\mper
\end{align}

\subsection{Factorization of the Matrix $\cM$}
In lines \ref{eq:factor-1} and \ref{eq:factor-2} we have defined two error matrices, $\xi_0, E_0 \in \R^{\nchoose{\leq d} \times \nchoose{\leq d}}$.
Inspired by the factorization of $\cM(I,J)$ in line \ref{eq:factor-0}, we define another pair of matrices as follows:
\begin{align*}
  \cQ_0 \in \R^{\nchoose{d} \times \nchoose{d}} \quad \text{ given by } \quad \cQ_0(S_\ell, S_r) = \sum_{\substack{\cR_m \text{ having \ref{itm:middle}} \\ |\V(\cR_m)| \leq \tau}} \Paren{\frac \omega n}^{|\V(\cR_m)| - \frac{|S_\ell| + |S_r|}{2}} \chi_{\cR_m}\\
  \cL \in \R^{\nchoose{d} \times \nchoose{d}} \quad \text{ given by } \quad \cL(I,S) = \Paren{\frac \omega n}^{-\frac{|S|}{2}} \sum_{\substack{\cR_\ell \text{ having \ref{itm:left}} \\ |\V(\cR_\ell)| \leq \tau}} \Paren{\frac \omega n}^{|\V(\cR_\ell)|} \chi_{\cR_\ell}\mper
\end{align*}
The powers of $(\omega/n)$ are split between $\cQ_0$ and $\cL$ so that the typical of eigenvalue of $\cQ_0$ will be approximately $1$ (although it will be some time before we are prepared to prove that).

The equation in lines \ref{eq:factor-0}, \ref{eq:factor-1}, and \ref{eq:factor-2} can be written succinctly as
\[
  \cM = \cL \cQ_0 \cL^\top - \xi_0 - E_0\mper
\]
As we will see later, with high probability $\cQ_0 \succeq 0$, and thus also $\cL \cQ_0 \cL^\top \succeq 0$.
So long as $\tau$ is sufficiently large, the spectral norm $\|\xi_0\|$ of the error term that accounts for ribbons whose size is too large will be negligible.
However, the error $E_0$ does not turn out to be negligible.
To overcome this we will apply a similar factorization approach to $E_0$ as we did for $\cM$; iterating this factorization will push down the error from ribbon nondisjointness.

We record an elementary fact about $\cQ_0$:
\begin{lemma}
  Let $\Pi$ be the projector to $\Span\{e_C \, : \, C \in \cC_{\leq d}\}$.
  Then $\cQ_0 = \Pi\cQ_0 = \cQ_0 \Pi$.
\end{lemma}
\begin{proof}
  Suppose $S$ is not a clique in $G$.
  We need to show that the row $\cQ_0(S,\cdot)$ is zero.
  For every entry $\cQ_0(S,S')$, notice that the Fourier coefficients $\widehat{\cQ_0(S,S')}(T) = \widehat{\cQ_0(S,S')}(T')$ if $T,T' \subseteq \nchoose{2}$ disagree only on edges inside $S$. (That is, $T \xor T' \subseteq {S \choose 2}$.)
  This means that $\cQ_0(S,S') = \1_{S \text{ is a clique in $G$}} \cdot f_{S,S'}(G)$ for some function $f_{S,S'}$.
\end{proof}

\subsection{Iterative Factorization of $E_0$}
We recall now the definition of the matrix $E_0 \in \R^{\nchoose{\leq d} \times \nchoose{\leq d}}$.
\[
  E_0(I,J) = \sum_{\substack{S_\ell, S_r \subseteq [n]\\ |S_\ell| = |S_r| \leq d}} \Paren{\frac \omega n}^{-\frac{|S_\ell| + |S_r|}{2}} \sum_{\substack{\cR_\ell, \cR_m, \cR_r \text{ satisfying}\\\text{\ref{itm:left},\ref{itm:middle},\ref{itm:right} and not \ref{itm:disjoint}}  \\
  |\V(\cR_\ell)|, |\V(\cR_m)|, |\V(\cR_r)| \leq \tau}} \Paren{\frac \omega n}^{|\V(\cR_\ell)| + |\V(\cR_r)| + |\V(\cR_m)| - \frac{|S_\ell| + |S_r|}{2}} \cdot \chi_{\cR_\ell} \cdot \chi_{\cR_m} \cdot \chi_{\cR_r}\mper
\]
In what follows, we will show how to factor a slightly more general sort of matrix; this factorization will be applicable iteratively, starting with $E_0$.

\subsubsection{The matrix $\cE_c$ and its factorization}
To express the family of matrices we will factor, we introduce a relaxation of our definition of ribbon and a corresponding relaxation \hyperref[itm:middle-sep]{3*} of condition \ref{itm:middle} of the $S_\ell, S_r$ factorization conditions.
\begin{definition}[Improper Ribbon]
  An \emph{improper $(I,J)$-ribbon} $\cR$ is an $(I,J)$-ribbon $\cR_0$ together with a set $\cZ(\cR) \subseteq [n]$ of vertices disjoint from $\cV(\cR_0)$.
  (Think of adding the vertices $\cZ(\cR)$ to the ribbon $\cR_0$ as degree-$0$ nodes.)
  We write $\cV(\cR) = \cV(\cR_0) \cup \cZ(\cR)$.
  When we need to distinguish, we sometimes call ordinary ribbons ``proper''.
\end{definition}
Every ribbon is also an improper ribbon by taking $\cZ(\cdot) = \emptyset$, and every improper ribbon has a corresponding ribbon given by deleting its degree-$0$ vertices.

\begin{center}
  \fbox{\begin{minipage}{\textwidth}
  \paragraph{Relaxed Factorization Condition for ribbon $\cR_1$ with $\cS_\ell, \cS_r \subseteq [n]$}
      \begin{enumerate}
        \item [\quad 3*.] $\cR_1$ is an improper $(S_\ell, S_r)$-ribbon.\label{itm:middle-sep}
      \end{enumerate}
  \end{minipage}}
\end{center}

Let $c$ be a $\R$-valued function $c(\cR)$ on (possibly improper) ribbons. Let $\cE_c \in \R^{\nchoose{\leq d} \times \nchoose{\leq d}}$ be given by
\begin{align}
  \cE_c(I,J) = \sum_{\substack{S_\ell, S_r \subseteq [n] \\ |S_\ell|, |S_r| \leq d}} \Paren{\frac \omega n }^{-\frac{|S_\ell| + |S_r|}{2}}\sum_{\substack{\cR_\ell, \cR_m, \cR_r \text{ satisfying}\\\text{\ref{itm:left},\hyperref[itm:middle]{3*},\ref{itm:right} and not \ref{itm:disjoint}}  \\
  |\V(\cR_\ell)|, |\V(\cR_m)|, |\V(\cR_r)| \leq \tau}} c(\cR_m) \Paren{\frac \omega n}^{|\V(\cR_\ell)| + |\V(\cR_r)| + |\V(\cR_m)|-\frac{|S_\ell| + |S_r|}{2}} \cdot \chi_{\cR_\ell} \cdot \chi_{\cR_m} \cdot \chi_{\cR_r}\mper\label{eqn:Ec}
\end{align}

Note that  \ref{itm:middle} is a strictly more restrictive condition than \hyperref[itm:middle]{3*}. Hence we can define the function $c_0$ by $c_{0}(\cR_m) = 1$ if $\cR_m$ satisfies \ref{itm:middle} and $c_{0}(\cR_m) = 0$ otherwise. Then $E_0 = \cE_{c_{0}}$. In this subsection, we will show how to factor any matrix of the form $\cE_c$ as
\[
  \cE_c = \cL \cQ_{c'} \cL^\top - \cE_{c'} - \xi_c
\]
for some function $c'$ on ribbons and matrices $\cQ_{c'}, \xi_c \in \R^{\nchoose{\leq d} \times \nchoose{\leq d}}$ where $\|\xi_c\|$ is negligible with high probability.

Just as our initial factorization of $\cM$ began with a factorization of each ribbon appearing in the Fourier expansion, our factorization of $\cE_c$ depends on a factorization for each triple $(\cR_\ell, \cR_m, \cR_r)$ appearing in \ref{eqn:Ec}.
Since they do not satisfy \ref{itm:disjoint}, there must be some vertices occurring in more than one of $\cV(\cR_\ell), \cV(\cR_m), \cV(\cR_\ell)$.
Before, the canonical factorization depended on the leftmost and rightmost vertex separators in an $(I,J)$-ribbon $\cR$ separating $I$ from $J$. But now we will be interested in leftmost and rightmost separators that separate both $I$ and $J$ from each other and from these repeated vertices.
\begin{definition}[Separating Factorization]
  Let $\cR_\ell, \cR_m, \cR_r$ be ribbons satisfying $S_\ell, S_r$ factorization conditions \ref{itm:left}, \hyperref[itm:middle]{3*}, \ref{itm:right} but not \ref{itm:disjoint}, with $|\V(\cR_\ell)|, |\V(\cR_m)|, |\V(\cR_r)| \leq \tau$.
  Let $\cR$ be the $(I,J)$-ribbon with edge set $W_{\cR_\ell} \oplus W_{\cR_m} \oplus W_{\cR_r}$  and vertex set $\V(\cR_\ell) \cup \V(\cR_m) \cup \V(\cR_r)$.
  (Thus, $\chi_{\cR_\ell} \cdot \chi_{\cR_m} \cdot \chi_{\cR_r} = \chi_{\cR}$.)

  Let $S'_\ell$ be the leftmost minimum-size vertex separator in $\cR$ which separates $I$ from $J$ and any vertices appearing in more than one of $\cV(\cR_\ell), \cV(\cR_m), \cV(\cR_r)$.
  Similarly, let $S'_r$ be the rightmost minimum-size vertex separator in $\cR$ separating $J$ from $I$ and these repeated vertices.
  (Notice that $S'_\ell$ and $S'_r$ could have different sizes.)

  Let $V_\ell'$ be the vertices reachable from $I$ without passing through $S'_\ell$ and similarly for $V_r'$. Let $V_m' = V_\cR \setminus (V_\ell' \cup V_r')$.
  Let $W_\ell' = \{(u,v) \in W_{\cR} \, : \, u \in V_\ell, v \in V_\ell \cup S_\ell' \}$ and similarly for $W_r'$, and let $W_m' = W_\cR \setminus (W_\ell' \cup W_r')$.

  Let $\cR_\ell'$ be the $(I,S'_\ell)$-ribbon with vertex set $V_\ell' \cup S_\ell'$ and edge set $W_\ell'$ and let $\cR_r'$ be the $(S'_r, J)$-ribbon with vertex set $V_r' \cup S_r'$ and edge set $W_r'$.
  Finally, let $\cR_m'$ be the improper $(S'_\ell, S'_r)$-ribbon with edge set $W'_m$ and vertex set $(\cV(\cR) \setminus (V'_\ell  \cup V'_r)) \cup S_\ell' \cup S_r')$.
\end{definition}

Note that $\chi_{\cR_\ell} \cdot \chi_{\cR_m} \cdot \chi_{\cR_r} = \chi_{\cR_\ell'} \cdot \chi_{\cR_m'} \cdot \chi_{\cR_r'}$ if $\cR_\ell', \cR_m', \cR_r'$ is the separating factorization for $\cR_\ell, \cR_m, \cR_r$.
We can use this to rewrite $\cE_c$ as
\begin{align}
  & \cE_c(I,J) = \nonumber \\
  & \sum_{\substack{S_\ell, S_r \subseteq [n] \\ |S_\ell|, |S_r| \leq d}} \Paren{\frac \omega n }^{-\frac{|S_\ell| + |S_r|}{2}}\sum_{\substack{\cR_\ell, \cR_m, \cR_r \text{ satisfying}\\\text{\ref{itm:left},\hyperref[itm:middle]{3*},\ref{itm:right} and not \ref{itm:disjoint}}  \\
  |\V(\cR_\ell)|, |\V(\cR_m)|, |\V(\cR_r)| \leq \tau \\ \text{separating factorization} \\ \cR_\ell', \cR_m', \cR_r', S_\ell', S_r'}} c(\cR_m) \Paren{\frac \omega n}^{|\V(\cR_\ell)| + |\V(\cR_r)| + |\V(\cR_m)|-\frac{|S_\ell| + |S_r|}{2}} \cdot \chi_{\cR_\ell'} \cdot \chi_{\cR_m'} \cdot \chi_{\cR_r'}\label{eq:exactEc}
  \end{align}
  Our goal is to find some coefficient function $c'$ on (improper) ribbons and a matrix $\cQ_{c'}$ so that this is approximately equal to $\cL \cQ_{c'} \cL^\top - \cE_{c'}$. For $c'$ yet to be chosen, we take
\[
 \cQ_{c'}(S'_\ell, S'_r) \defeq \sum_{\substack{\cR'_m \text{ having \hyperref[itm:middle]{3*}} \\ |\V(\cR'_m)| \leq \tau}} c'(\cR'_m)\Paren{\frac \omega n}^{|\V(\cR'_m)| - \frac{|S'_\ell| + |S'_r|}{2}} \chi_{\cR'_m}
\]
and have that 
  \begin{align}
 & \cL \cQ_{c'} \cL^\top(I,J) - \cE_{c'}(I,J) = \nonumber \\
& \sum_{\substack{S'_\ell, S'_r \subseteq [n] \\ |S'_\ell|, |S'_r| \leq d}} \Paren{\frac \omega n }^{-\frac{|S'_\ell| + |S'_r|}{2}}\sum_{\substack{\cR'_\ell, \cR'_m, \cR'_r \text{ satisfying}\\\text{\ref{itm:left},\hyperref[itm:middle]{3*},\ref{itm:right}, and \ref{itm:disjoint}}  \\
  |\V(\cR'_\ell)|, |\V(\cR'_m)|, |\V(\cR'_r)| \leq \tau}} c'(\cR'_m) \Paren{\frac \omega n}^{|\V(\cR'_\ell)| + |\V(\cR'_r)| + |\V(\cR'_m)|-\frac{|S'_\ell| + |S'_r|}{2}} \cdot \chi_{\cR'_\ell} \cdot \chi_{\cR'_m} \cdot \chi_{\cR'_r}\mper\label{eqn:approxEc}
\end{align}
  We will compare \eqref{eq:exactEc} and \eqref{eqn:approxEc} by collecting like terms, but first we handle the discrepancy in the size bounds on the ribbons with a corresponding error term $\xi_c$.
  The following matrix is similar to $\cE_c$, but places a size bound on the ribbons in the separating factorization $|\cV(\cR_\ell')|, |\cV(\cR_m')|, |\cV(\cR_r')| \leq \tau$.
  We define
\begin{align*}
  & \cE_c'(I,J) = \\
  & \sum_{\substack{S_\ell, S_r \subseteq [n] \\ |S_\ell|, |S_r| \leq d}} \Paren{\frac \omega n }^{-\frac{|S_\ell| + |S_r|}{2}}\sum_{\substack{\cR_\ell, \cR_m, \cR_r \text{ satisfying}\\\text{\ref{itm:left},\hyperref[itm:middle]{3*},\ref{itm:right} and not \ref{itm:disjoint}}  \\
  \text{separating factorization} \\ \cR_\ell', \cR_m', \cR_r', S_\ell', S_r' \\
|\V(\cR_\ell')|, |\V(\cR_m')|, |\V(\cR_r')| \leq \tau }}  c(\cR_m) \Paren{\frac \omega n}^{|\V(\cR_\ell)| + |\V(\cR_r)| + |\V(\cR_m)|-\frac{|S_\ell| + |S_r|}{2}} \cdot \chi_{\cR_\ell'} \cdot \chi_{\cR_m'} \cdot \chi_{\cR_r'}
  \end{align*}
  We take $\xi_c = \cE_c' - \cE_c$ and we will show below that with high probability the error $\|\xi_c\|$ is negligible. Before doing this, we show that $\cE_c'$ is exactly equal to $\cL^\top \cQ_{c'} \cL^\top - \cE_{c'}$ for the correct choice of $c'$.

  To collect like terms, it helps to define the following quantity $\gamma_{\cR_\ell', \cR_m', \cR_r',I,J,S_\ell',S_r'}$.
  \[
    \gamma_{\cR_\ell', \cR_m', \cR_r',I,J,S_\ell', S_r'} \eqdef \sum_{\substack{\cR_\ell, \cR_m, \cR_r \text{ satisfying}\\\text{\ref{itm:left},\hyperref[itm:middle]{3*},\ref{itm:right} and not \ref{itm:disjoint} for some } S_\ell, S_r
   \\ \text{separating factorization } \cR_\ell', \cR_m', \cR_r', S_\ell', S_r'}}
    c(\cR_m) \Paren{\frac \omega n}^{|\cV(\cR_\ell)| + |\cV(\cR_m)| + |\cV(\cR_r)| + \frac{|S_\ell'| + |S_r'|}{2} - |S_\ell| - |S_r|}\mper
  \]
  Then we can rewrite $\cE_c'(I,J)$ again as
\begin{align*}
  \cE_c'(I,J) = \sum_{\substack{S'_\ell, S_r' \subseteq [n]\\ |S_\ell'|,|S_r'| \leq d}} \Paren{\frac \omega n}^{-\frac{|S_\ell'| + |S_r'|}{2}}
  \sum_{\substack{\cR'_\ell, \cR'_m, \cR'_r\\ \text{ satisfying \ref{itm:left}, \hyperref[itm:middle]{3*}, \ref{itm:right}, \ref{itm:disjoint} for $S_\ell', S_r'$} \\ |\cV(\cR'_\ell)|, |\cV(\cR'_m)|, |\cV(\cR_r)| \leq \tau }}   \gamma_{\cR_\ell', \cR_m', \cR_r',I,J,S_\ell',S_r'}
\cdot \chi_{\cR_\ell'} \cdot \chi_{\cR_m'} \cdot \chi_{\cR_r'}
\end{align*}
We will obtain $\cE'_c = \cL^\top \cQ_{c'} \cL^\top - \cE_{c'}$ if we define $c'(\cR_m')$ so that
\[
c'(\cR'_m) \Paren{\frac \omega n}^{|\V(\cR'_\ell)| + |\V(\cR'_r)| + |\V(\cR'_m)|-\frac{|S'_\ell| + |S'_r|}{2}} = \gamma_{\cR_\ell', \cR_m', \cR_r',I,J,S_\ell', S_r'}
\]
To express this in terms of the function $c$, we expand out $\gamma_{\cR_\ell', \cR_m', \cR_r', I,J,S_\ell',S_r'}$. It is useful to define:
\begin{definition}
  Let 
  \[
    r = (|\cV(\cR_\ell)| + |\cV(\cR_m)| + |\cV(\cR_r)| - |S_\ell| - |S_r|)- (|\cV(\cR_\ell')| + |\cV(\cR_m')| + |\cV(\cR_r')| - |S_\ell'| - |S_r'|)\mper
   \]
  (The ribbons $\cR_\ell, \cR_m, \cR_r, \cR_\ell', \cR_m', \cR_r'$ will always be clear from context.)
\end{definition}
  
Note that $(\cV(\cR_\ell)| + |\cV(\cR_m)| + |\cV(\cR_r)| - |S_\ell| - |S_r|)$ is the total number of vertices we would have in the $(I,J)$-ribbon with vertex set $\cV(\cR_\ell) \cup \cV(\cR_m) \cup \cV(\cR_\ell)$ if $\cR_\ell,\cR_m,\cR_r$ satisfied condition \ref{itm:disjoint} (which they do not!).
Similarly, $(|\cV(\cR'_\ell)| + |\cV(\cR'_m)| + |\cV(\cR'_r)| - |S'_\ell| - |S'_r|)$ is the total number of vertices in the $(I,J)$-ribbon with edge set $\cW(\cR_\ell') \cup \cW(\cR_m') \cup \cW(\cR_r')$ and vertex set $\cV(\cR_\ell') \cup \cV(\cR_m') \cup \cV(\cR_r')$.
Thus, $r$ is the number of vertices occurring with multiplicity higher than they should in $\cV(\cR_\ell) \cup \cV(\cR_m) \cup \cV(\cR_r)$.

We can rewrite the $\gamma$'s as
\[
    \gamma_{\cR_\ell', \cR_m', \cR_r',I,J,S_\ell', S_r'} = \Paren{\frac \omega n}^{|\cV(\cR_\ell')| + |\cV(\cR_m')| + |\cV(\cR_r')| - \frac{|S_\ell'| + |S_r'|}{2}} \sum_{\substack{\cR_\ell, \cR_m, \cR_r \text{ satisfying}\\\text{\ref{itm:left},\hyperref[itm:middle]{3*},\ref{itm:right} and not \ref{itm:disjoint} for some } S_\ell, S_r\\
    r \text{ intersections outside } S_\ell, S_r\\
    \\ \text{separating factorization } \cR_\ell', \cR_m', \cR_r' S_\ell', S_r'}}
    c(\cR_m) \Paren{\frac \omega n}^{r}\mper
\]
Thus, we will have that $\cE_c' = \cL \cQ_{c'} \cL^\top - \cE_{c'}$ if and only if for every $(S_\ell', S_r')$-ribbon $\cR_m'$ and every $\cR_\ell', \cR_r'$ satisfying \ref{itm:left}, \ref{itm:right},
\[
  c'(\cR_m') = \sum_{\substack{\cR_\ell, \cR_m, \cR_r \text{ satisfying}\\\text{\ref{itm:left},\hyperref[itm:middle]{3*},\ref{itm:right} and not \ref{itm:disjoint} for some } S_\ell, S_r\\
    r \text{ intersections outside } S_\ell, S_r\\
    \\ \text{separating factorization } \cR_\ell', \cR_m', \cR_r' S_\ell', S_r'}}
    c(\cR_m) \Paren{\frac \omega n}^{r}\mper
\]
Note that for this to happen, the right hand side must be independent of $\cR_\ell'$ and $\cR_r'$. If this is the case, then we can define 
\[
  c'(\cR_m') \defeq \sum_{\substack{\cR_\ell, \cR_m, \cR_r \text{ satisfying}\\\text{\ref{itm:left},\hyperref[itm:middle]{3*},\ref{itm:right} and not \ref{itm:disjoint} for some } S_\ell, S_r\\
    r \text{ intersections outside } S_\ell, S_r\\
    \\ \text{separating factorization } \cR_\ell', \cR_m', \cR_r' S_\ell', S_r'}}
    c(\cR_m) \Paren{\frac \omega n}^{r} \text{for some } \cR_\ell', \cR_r' \text{ satisfying } \ref{itm:left}, \ref{itm:right}\mper
\]
The next claim shows that, indeed, the choice of $\cR_\ell', \cR_r'$ does not matter.
(This would not have been true without passing from $\cE_c$ to $\cE_c'$.)
\begin{claim}
  Let $\cR'_\ell, \cR_m', \cR_r'$ satisfy \ref{itm:left}, \hyperref[itm:middle]{3*}, \ref{itm:right}, \ref{itm:disjoint} for some $S_\ell',S_r' \subseteq [n]$.
  Let $\cR_\ell^{''}$ and $\cR_r^{''}$ also satisfy \ref{itm:left} and \ref{itm:right}, respectively, for $S_\ell', S_r'$, respectively.
  Then
  \[
\sum_{\substack{\cR_\ell, \cR_m, \cR_r \text{ satisfying}\\\text{\ref{itm:left},\hyperref[itm:middle]{3*},\ref{itm:right} and not \ref{itm:disjoint} for some } S_\ell, S_r\\
    r \text{ intersections outside } S_\ell, S_r\\
    \\ \text{separating factorization } \cR_\ell', \cR_m', \cR_r' S_\ell', S_r'}}
    c(\cR_m) \Paren{\frac \omega n}^{r}
    = \sum_{\substack{\cR_\ell, \cR_m, \cR_r \text{ satisfying}\\\text{\ref{itm:left},\hyperref[itm:middle]{3*},\ref{itm:right} and not \ref{itm:disjoint} for some } S_\ell, S_r\\
    r \text{ intersections outside } S_\ell, S_r\\
\\ \text{separating factorization } \cR_\ell^{''}, \cR_m', \cR_r^{''} S_\ell', S_r'}}
    c(\cR_m) \Paren{\frac \omega n}^{r}\mper
  \]
  (Notice that the left-hand sum refers to $\cR_\ell', \cR_r'$ and the right-hand one to $\cR_\ell^{''}, \cR_r^{''}$.)
\end{claim}
\begin{proof}
We prove this by showing that there is an exact match between terms on the left hand side and terms on the right hand side. Consider a term on the left hand side. Note that the part of $\cR_{\ell}$ between $I$ and $S'_{\ell}$ must be $\cR'_{\ell}$ while the part of $\cR_{\ell}$ between $S'_{\ell}$ and $S_{\ell}$ becomes part of $\cR'_m$. To shift from $\cR'_{\ell}$ to $\cR''_{\ell}$, we simply replace $\cR'_{\ell}$ by $\cR''_{\ell}$ within $\cR_{\ell}$. Similarly, to shift from $\cR'_{r}$ to $\cR''_{r}$, we simply replace $\cR'_{r}$ by $\cR''_{r}$ within $\cR_{r}$.

To show that this gives an exact match, we need to show that $r$ is unaffected by these shifts. To see that shifting from $\cR'_{\ell}$ to $\cR''_{\ell}$ does not affect $r$, note that all vertices in $\cV(\cR'_{\ell}) \setminus S'_{\ell}$ or  $\cV(\cR'_{\ell}) \setminus S'_{\ell}$ must appear in the corresponding $\cR_{\ell}$ and cannot appear in $\cR_{m}$ or $\cR_{r}$. Thus, these vertices always have multiplicity $1$ in $\cV(\cR_{\ell}) \cup \cV(\cR_{m}) \cup \cV(\cR_{r})$. All other vertices (including the ones in $S'_{\ell}$) may appear in $\cR_{m}$ or $\cR_{r}$ as well as $\cR_{\ell}$ but whether or not they do so is unaffected by the shift so their multiplicities in $\cV(\cR_{\ell}) \cup \cV(\cR_{m}) \cup \cV(\cR_{r})$ are unaffected by the shift and $r$ remains the same. A similar argument holds for shifting from $\cR'_{r}$ to $\cR''_{r}$
\end{proof}
\begin{remark}
For this argument, it was important to keep track of the isolated vertices in $\cR'_m$. If we did not keep track of isolated vertices and instead had them disappear, we could have a situation where there is a vertex $v$ which appears in $\cR_{\ell}$ and $\cR_{m}$ but disappears from $\cR'_{m}$ and is not in $S'_{\ell}$. Since $v$ is no longer in $\cR'_{m}$, $\cR''_{\ell}$ could contain $v$. If so, then we cannot shift from  $\cR'_{\ell}$ to $\cR''_{\ell}$ as this would create a copy of $v$ to the left of $S'_{\ell}$ but $v$ should be to the right of $S'_{\ell}$.
\end{remark}
Putting everything together, $\cE_c' = \cL \cQ_{c'} \cL^\top - \cE_{c'}$. Since we defined $\xi_c = \cE'_{c} - \cE_{c}$, we get that $\cE_c = \cL \cQ_c \cL^\top - \cE_{c'} - \xi_c$, as needed.

The remaining step will be to show that with high probability, the error term $\xi_c$ has negligible norm, which we will accomplish in Section~\ref{sec:xi}.

Finally, we record the following easy lemma about separating factorizations, which will be useful in the application of the foregoing to factor $\cE_0$.
\begin{lemma}\label{lem:sep-size-incr}
Suppose $\cR_\ell, \cR_m, \cR_r$ satisfy conditions \ref{itm:left}, \hyperref[itm:middle]{3*}, \ref{itm:right}, but not \ref{itm:disjoint}.
Let $\cR_\ell', \cR_m', \cR_r'$ be their separating factorization, with separators $S_\ell', S_r'$.
Then
\[
\frac{|S_\ell'| + |S_r'|}{2} - \frac{|S_\ell| + |S_r|}{2} \geq \frac 1 2
 \]
\end{lemma}
\begin{proof}
We claim that $|S_\ell| + |S_r| + 1 \leq |S_\ell'| + |S_r'|$
    By the violation of condition \ref{itm:disjoint}, we cannot have $S_\ell = S_\ell'$ and $S_r = S_r'$.
    But since $S_\ell'$ separates $I$ from $S_\ell$ in $\cR_\ell$ and $\cR_\ell$ is an $(I,S_\ell)$-ribbon whose rightmost vertex separator is also $S_\ell$, if $S_\ell \neq S_\ell'$ then $|S_\ell| < |S_\ell'|$, and similarly for $S_r$ and $S_r'$.
    So either $|S_\ell| < |S_\ell'|$ or $|S_r| < |S_r'|$, and since the separator sizes are integers, so the difference must be at least $1$ and we are done.
\end{proof}

\subsubsection{Application to $E_0$ and $\cM$}
We are ready to define our recursive factorization of $E_0$. Recall that $c_{0}(\cR_m) = 1$ if $\cR_m$ satisfies \ref{itm:middle} and $c_{0}(\cR_m) = 0$ otherwise and $E_0 = \cE_{c_{0}}$.
Applying the factorization above to $\cE_{c_{0}}$ we obtain matrices $\xi_1 = \xi_{c_0}, \cQ_1$, and $\cE_{c_1}$.
Then of course we can apply the factorization again to $\cE_{c_1}$.

Proceeding inductively, for all $i \in [1,2d]$ let $\xi_i = \xi_{c_{i-1}}, \cQ_i,$ and $\cE_{c_i}$ be the matrices given by applying the factorization to $\cE_{c_{i-1}}$ at step $i$.
\begin{claim}\label{clm:factorization}
  \[
    \cM = \cL(\cQ_0 - \cQ_1 + \cQ_2 - \ldots - \cQ_{2d-1} + \cQ_{2d})\cL^\top - (\xi_{0} - \xi_{1} + \xi_{2} - \ldots - \xi_{2d-1} + \xi_{2d})\mper
  \]
\end{claim}
\begin{proof}
We have that $\cM = \cL(\cQ_0)\cL^\top - \cE_0 - \xi_{0}$ and $\cE_{i-1} = \cL \cQ_{i} \cL^\top - \cE_{i} - \xi_{c_{i-1}} = \cL \cQ_{i} \cL^\top - \cE_{i} - \xi_{i}$. We prove the claim by starting with the first formula and appliying the second formula for each $i \in [1,2d]$. At the end, we are left with an extra term $\cE_{2d}$. We must show that $\cE_{2d} = 0$.

To see why $\cE_{2d} = 0$, note that every time we have a separating factorization $\cR'_{\ell},\cR'_{m},\cR'_{r}$ for $\cR_{\ell},\cR_{m},\cR_{r}$, the size of either the left separator or the right separator must increase (see Lemma~\ref{lem:sep-size-incr}). However, the size of these separators is always at most $d$, so the only way we can do this for $2d$ steps is if we started with the empty set as the separators and increased the size of either the left or right separator by $1$ each time, but not both. However, this too is impossible as if we start with the empty set as the separators, after the first step both the new left separator and the new right separator must have size at least $1$.
\end{proof}
\section{$\cM$ is PSD} \label{sec:spectral-norms}
In this section we combine the factorization of $\cM$ in terms of the matrices $\cL, \cQ_i, \xi_i$ that we obtained in Section~\ref{sec:factorization} with estimates on the eigenvalues of the $\cQ$s and $\xi$s.
The starting point is the following PSDness claim for $\cQ_0$.
\begin{lemma}\torestate{
\label{lem:Q0psd}
  Let $D \in \R^{\nchoose{\leq d} \times \nchoose{\leq d}}$ be the diagonal matrix with $D(S,S) = 2^{{|S| \choose 2}}/4$ if $S$ is a clique in $G$ and $0$ otherwise.
  With high probability, $\cQ_0 \succeq D$.
  }
\end{lemma}
We also need to bound $\|\cQ_i\|$ for $i > 0$.
\begin{lemma}\torestate{
\label{lem:Qismall}
Let $D \in \R^{\nchoose{\leq d} \times \nchoose{\leq d}}$ be the diagonal matrix with $D(S,S) = 2^{{S \choose 2}}/4$ if $S$ is a clique and is otherwise zero.
With high probability, every $\cQ_i$ for $i \in [1,2d]$ satisfies
\[
  \frac {-D}{8d} \preceq \cQ_i \preceq \frac{D}{8d}\mper
\]
}
\end{lemma}

The preceding lemmas are enough to obtain $\cQ_0 - \ldots + \cQ_{2d} \succeq D/2$, but in the end we need to work with the matrix $\cL (\cQ_0 - \ldots + \cQ_{2d}) \cL^\top - (\xi_0 - \ldots + \xi_{2d})$.
The next two lemmas allow us to make this last step.

\begin{lemma}\torestate{\label{lem:Lbound}
  With high probability, $\Pi \cL \Pi \cL^\top \Pi \succeq \Omega (\omega/n)^{d+1} \cdot \Pi$, where as usual $\Pi$ is the projector to $\Span \{e_C \, : \, C \in \cC_{\leq d}\}$.
  }
\end{lemma}

Finally, we need a bound on the $\xi$ matrices.
\begin{lemma}\torestate{\label{lem:xi-bound}
  With high probability, $\|\xi_0 - \ldots + \xi_{2d} \| \leq n^{-16d}$.
  }
\end{lemma}

We can now prove Lemma~\ref{lem:PSD-main2}.
\begin{proof}[Proof of Lemma~\ref{lem:PSD-main2}]
  By Claim~\ref{clm:factorization},
  \[
    \cM = \cL(\cQ_0 - \cQ_1 + \cQ_2 - \ldots - \cQ_{2d-1} + \cQ_{2d})\cL^\top - (\xi_{0} - \xi_{1} + \xi_{2} - \ldots - \xi_{2d-1} + \xi_{2d})\mper
  \]
  By a union bound, with high probability the conclusions of Lemmas~\ref{lem:Q0psd},~\ref{lem:Qismall},~\ref{lem:Lbound}, and~\ref{lem:xi-bound} all hold.
  By Lemma~\ref{lem:Q0psd} and Lemma~\ref{lem:Qismall},
  \[
     \cQ_0 - \cQ_1 + \cQ_2 - \ldots - \cQ_{2d-1} + \cQ_{2d} \succeq \frac D 2 \succeq \frac \Pi 2\mper
  \]
  where as usual $\Pi$ is the projector to $\Span{e_C \, : \, C \in \cC_{\leq d}}$.
  Thus by Lemma~\ref{lem:Lbound}, we obtain $\cL (\cQ_0 - \ldots + \cQ_{2d}) \cL^\top \succeq \Omega(\omega/n)^{d+1} \cdot \Pi$.
  Finally, by Lemma~\ref{lem:xi-bound} we have
  \[
    \cM = \Pi \cdot \cM \cdot \Pi \succeq \Omega \Paren{\frac \omega n}^{d + 1} \cdot \Pi + n^{-16d} \cdot \Pi \succeq 0\mper\qedhere
  \]
\end{proof}

In the next subsections, we prove the foregoing lemmas.

\subsection{Ribbons and Spectral Norms}
Our PSDness arguments require bounds on the spectral norm of certain random matrices.
 

Our random matrices arise out of decompositions of the moment matrix from Definition \ref{def:moment-matrix} and are functions of a graph $G$ on vertex set $[n]$. Our norm bounds will hold for what we call as \emph{graphical matrices}, that are are defined to capture the matrices are invariant under permutation of vertices in the graph $G$ and are in fact "minimal" such matrices. 

We first identify the \emph{shape} of a ribbon that basically identifies the structure of a ribbon up to renaming.
\newcommand{\shape}{\mathsf{shape}}
\begin{definition}[Shape of a Ribbon]
For an $(I,J)$-ribbon $\cR$, consider the graph $U$ on the vertex set $[|\cV(\cR)|]$ whose edges are
\[
  E(U) = \{(i,j) \, : \, \text{ there is an edge in $\cR$ from the $i$-th to the $j$-th least element of $\cV(\cR)$}  \}\mper
\]
(Here we are considering $\cV(\cR)$ to have the usual ordering inherited from $[n]$.)
Also, let $U$ have two distinguished subsets of vertices $A$ and $B$, where $A = \{ i \, : \, \text{the $i$-th element of $\cV(\cR)$ is in $I$} \}$, and similarly for $B$ and $J$.
We call $U$ the \emph{shape} of $\cR$ and write $\shape(\cR) = U$.
\end{definition}

We record some observations on shapes of ribbons.
\begin{itemize}
  \item If $\cR$ is a ribbon (not an improper ribbon), its shape satisfies the assumptions of Lemma~\ref{lem:single-topology-matrix} (namely, that every vertex outside $A \cup B$ has degree at least $1$).
  \item If, for example, $\cR$ is an $(I,J)$ ribbon where $I \cap J = \{1\}$ (which must be the least element in both $I$ and $J$), then $(I',J')$-ribbon $\cR'$ only has the same shape as $\cR$ if $|I' \cap J'| = 1$ and contains only the least element in $I$ and $J$.
  More broadly, specifying the shape of a ribbon in particular specifies the pattern of intersection of its endpoints.
  \item A matrix $M \in \R^{{n \choose \leq d} \times {n \choose \leq d}}$ whose entries are given by $M(I,J) = \sum_{\cR \text{ an $(I,J)$-ribbon with shape $U$}} \chi_{\cR}$ satisfies the assumptions of Lemma~\ref{lem:single-topology-matrix}.
  In the following sections, our main strategy will be to decompose the matrices $\cQ_i$ into matrices of this form.
\end{itemize}

We are now ready to define graphical matrices.
\begin{definition}[Graphical Matrices] \label{def:graphical-matrix}
Let $U$ be a graph on the vertex set $[t]$ with two distinguished sets of vertices $A,B \subseteq [t]$. Let $\cal{T}(U)$ be the collection of all $I,J$ ribbons with shape $U$. The graphical matrix $M \in \R^{\nchoose{|A|} \times \nchoose{|B|}}$ of shape $U$ is defined by $$M(I,J) = \sum_{ \cR: \cR \text{ is an $(I,J)$-ribbon and } \shape(\cR) = U} \chi_{\cR}.$$
\end{definition}
\begin{example} \label{ex:adjacency}
When $U$ is a graph on $2$ vertices with distinguished sets $\{1\}$ and $\{2\}$ of size $1$ each and a single edge connecting vertex $1$ and $2$, the graphical matrix of shape $U$ is just the standard $\on$-adjacency matrix of the graph $G$. 
\end{example}
The following lemma will be our main tool.
It is in essence due to Medarametla and Potechin \cite{MP16} and special cases of the bound have been proven and used in \cite{HKPRS16, HKP15, DM15}.
We give a proof in the appendix for completeness.

\begin{lemma}\torestate{
  \label{lem:single-topology-matrix}
  Let $U$ be a graph on $t \leq O(\log n)$ vertices, with two distinguished subsets of vertices $A$ and $B$, and suppose:
  \begin{itemize}
    \item $U$ admits $p$ vertex-disjoint paths from $A \setminus B$ to $B \setminus A$.
    \item $|A \cap B| = r$.
    \item Every vertex outside $A \cup B$ has degree at least $1$.
  \end{itemize}
Let $M = M(G)$ be the graphical matrix with shape $U$. Then, whp, 
 $\|M\| \leq n^{\tfrac{t - p - r}{2}} \cdot 2^{O(t)} \cdot (\log n)^{O(t -r + p)}$.
  }
\end{lemma}

\begin{remark}
Lemma \ref{lem:single-topology-matrix} can be seen as a generalization of the standard upper bound on the spectral norm of the adjacency matrix. Example \ref{ex:adjacency} shows how adjacency matrix is a graphical matrix with a shape $U$ on $2$ vertices with a single edge connecting them, thus, $t = 2$ and $r = 1$. Lemma \ref{lem:single-topology-matrix} thus shows an upper bound of $\sqrt{n} \poly \log{(n)}$ on the spectral norm of the adjacency matrix which is tight up to a $\poly \log{(n)}$ factor.
\end{remark}

\subsection{PSDness for $\cQ_0$---Proof of Lemma~\ref{lem:Q0psd}}
In this section we prove Lemma~\ref{lem:Q0psd}, which we restate here.
\restatelemma{lem:Q0psd}
\begin{proof}[Proof of Lemma~\ref{lem:Q0psd}]
  To begin, we split $\cQ_0$ into its diagonal $\cQ_0^{\text{diag}}$ and its off-diagonal $\cQ_0^{\text{off-diag}}$ parts.
  \[
    \cQ_0^{\text{diag}}(S_\ell, S_r) = \begin{cases} \cQ_0(S_\ell, S_r) \text{ if $S_\ell = S_r$} \\ 0 \text{ otherwise.}\end{cases} \qquad
     \cQ_0^{\text{off-diag}}(S_\ell, S_r) = \begin{cases} \cQ_0(S_\ell, S_r) \text{ if $S_\ell \neq S_r$} \\ 0 \text{ otherwise.}\end{cases}
  \]
  Then $\cQ_0 = \cQ_0^{\text{diag}} + \cQ_0^{\text{off-diag}}$.
  Expanding $\cQ_0^{\text{diag}}$,
  \[
    \cQ_0^{\text{diag}}(S,S) = 2^{{|S| \choose 2}} \cdot \1_{S\text{ is a clique}} \cdot \left (1 + \sum_{\substack{\cR \text{ nonempty, having \ref{itm:middle}}\\\text{and no edges inside $S$}\\ |S| < |\cR| \leq \tau}} \Paren{\frac \omega n}^{|\cV(\cR)| - |S|} \cdot \chi_{\cR} \right ) = 2^{{|S| \choose 2}} \cdot \1_{S\text{ is a clique}} \cdot (1 \pm n^{-\Omega(\epsilon)})
  \]
  for all $S \in \nchoose{d}$ with high probability by a similar argument as in Lemma~\ref{lem:normalization} and a union bound.

  Next, we bound $\|\cQ_0^{\text{off-diag}}\|$ be decomposing it according to ribbon shape.
  Fix $s, t \leq \tau$.
  Let $U_1^{(s,t)},\ldots,U_q^{(s,t)}$ be all the graphs on vertex set $[t]$ with two distinguished sets of vertices $A,B$, both of size $s$, with $|A \cap B| \leq s-1$, and where there are $s - |A \cap B|$ vertex-disjoint paths from $A \setminus B$ to $B \setminus A$.
  Let $M_i^{(s,t)}$ be given by
  \[
    M_i^{(s,t)}(S_\ell,S_r) = \sum_{\cR \text{ an $(S_\ell,S_r)$-ribbon with shape $U_i^{(s,t)}$}} \chi_{\cR}\mper
  \]
  Then
  \[
    \cQ_0^{\text{off-diag}} = \sum_{\substack{s \leq d \\ t \leq \tau \\ i \leq q}} \Paren{\frac \omega n}^{t - s} \cdot M_i^{(s,t)}\mper
  \]

  We can apply Lemma~\ref{lem:single-topology-matrix} to conclude that with probability at least $1 - O(n^{-100 \log n})$,
  \[
    \Norm{ \Paren{\frac \omega n}^{t - s} \cdot M_i^{(s,t)}} \leq \Paren{\frac \omega n}^{t - s} \cdot n^{\frac{t - s}{2}} \cdot 2^{O(t)} \cdot (\log n)^{O(t - |A \cap B| + |A \setminus B|)} \leq n^{-\epsilon(t - s)} \cdot 2^{O(t)} \cdot (\log n)^{O(t - s)}\mcom
  \]
  where to conclude the bound on the exponent in $(\log n)^{O(t - |A \cap B| + |A \setminus B|)}$ we have used that $t \geq 2s - |A \cap B|$.

  Notice that for fixed $s$ and $t$, there are at most $2^{{t \choose 2} + O(t)}$ unique shapes $U_1^{(s,t)},\ldots,U_q^{(s,t)}$.
  Thus, a union bound followed by the triangle inequality, we obtain that for fixed $s$ and $t$, with probability at least $1 - O(n^{-99 \log n})$,
  \[
    \Norm{\Paren{\frac \omega n}^{t - s} \sum_{i \leq q} M_i^{(s,t)}} \leq 2^{{t \choose 2} + O(t)} \cdot n^{-\epsilon(t - s)} \cdot 2^{O(t)} \cdot (\log n)^{O(t - s)}\mper
  \]
  Under our assumptions on the parameters $d, \tau,$ and $\epsilon$, this is at most $2^{{s \choose 2}}/ (100 \tau)$.
  Summing over all $t \leq \tau$, for a fixed $s$ we have
  \[
    \Norm{\Paren{\frac \omega n}^{t - s} \sum_{\substack{t \leq \tau \\ i \leq q}} M_i^{(s,t)}} \leq \frac{2^{{s \choose 2}}}{100}\mper
  \]
  Notice that the above matrix is exactly the block of $\cQ_0^{\text{off-diag}}$ corresponding to subsets of size $s$.
  Together with our bound on $\cQ_0^{\text{diag}}$, this proves the lemma.
\end{proof}

\subsection{Norm Bounds for $\cQ_i$---Proof of Lemma~\ref{lem:Qismall}}
In this section we prove Lemma~\ref{lem:Qismall}, restated here.
\restatelemma{lem:Qismall}

We will need to bound the coefficients $c_i(\cR_m')$ used to define the matrices $\cQ_i$ which we set up in Section~\ref{sec:factorization}.
\begin{lemma}\label{lem:ci-bound}
  Let $c_1,\ldots,c_{2d}$ be the coefficient functions defined in Section~\ref{sec:factorization}.
  For all improper $(S_\ell,S_r)$-ribbons $\cR_m$ admitting exactly $p$ vertex-disjoint paths from $S_\ell$ to $S_r$, and all $i \leq 2d$,
  writing $s = \frac{|S_\ell| + |S_r|}{2}$,
  \[
    c_i(\cR_m) \leq \Paren{\frac \omega n}^{s} \cdot n^{\frac { p - |\cZ(\cR_m)| -i/2}{2} + \epsilon s}\mper 
  \]
 recalling that $\omega = n^{1/2 - \epsilon}$.
 Furthermore, if $\cR_m$ and $\cR_m'$ have the same shape, then $c_i(\cR_m) = c_i(\cR_m')$.
\end{lemma}

With this lemma in hand we can prove Lemma~\ref{lem:Qismall}.

\begin{proof}[Proof of Lemma~\ref{lem:Qismall}]
  Fix some $0 < i \leq 2d$.
  We will use Lemma~\ref{lem:single-topology-matrix}, which requires that we first decompose each $\cQ_i$ into simpler matrices.
  First of all, for a proper ribbon $\cR_m$, let
  \[
    \tilde c_i(\cR_m) = \sum_{\cR_m' \text{ an improper ribbon whose largest proper subribbon is $\cR_m$}} \Paren{\frac \omega n}^{|\cZ(\cR_m')|} \cdot c_i(\cR'_m)\mper
  \]
  
Note that we include $\cR_m$ itself in this sum as a proper ribbon is also an improper ribbon. 
\begin{claim}
$\tilde c_i(\cR_m) \leq 2(\omega/n)^{s} \cdot n^{\frac{p - i/2}{2}+\epsilon{s}}$, where $p$ is the number of vertex-disjoint paths from $S_\ell$ to $S_r$ in $\cR_m$.
\end{claim}
\begin{proof}
Consider all of the improper ribbons $\cR'_{m}$ with $k$ isolated vertices whose largest proper subribbon is $\cR_{m}$. For each such ribbon $\cR'_{m}$, by Lemma~\ref{lem:ci-bound}, $(\omega/n)^{k} c_i(\cR'_m) \leq \Paren{\frac \omega n}^{k+s} \cdot n^{\frac { p - k - i/2}{2} + \epsilon s}$. There are at most $n^k$ such improper ribbons. Adding all of their contributions together gives at most
\[
\Paren{\frac \omega {\sqrt{n}}}^{k}\Paren{\frac \omega n}^{s} \cdot n^{\frac { p - i/2}{2} + \epsilon s} < 2^{-k}(\omega/n)^{s} \cdot n^{\frac{p - i/2}{2} + \epsilon{s}}
\]
Summing this up over all $k \geq 0$ gives the result.
\end{proof}

  Now fix $s_\ell, s_r \leq d$ and $t \leq \tau$ and let $U_1^{(s_\ell,s_r,t)},\ldots,U_q^{(s_\ell,s_r,t)}$ be all graphs on the vertex set $[t]$ with two distinguished subsets of vertices: $A$ of size $s_\ell$ and $B$ of size $s_r$.
  Let
  \begin{align*}
    M_j^{(s_\ell,s_r,t)}(S_\ell,S_r) & = \sum_{\cR \text{ is an $(S_\ell,S_r)$-ribbon with shape $U_j^{(s_\ell, s_r, t)}$}} \tilde c_i(\cR) \cdot \Paren{\frac \omega n}^{t - s} \cdot \chi_{\cR}\\
    & = \tilde c_i(U_j^{(s_\ell,s_r,t)})\sum_{\cR \text{ is an $(S_\ell,S_r)$-ribbon with shape $U_j^{(s_\ell, s_r, t)}$}} \Paren{\frac \omega n}^{t - s} \cdot \chi_{\cR}\mcom
  \end{align*}
  where $s = \frac{s_{\ell}+s_r}{2}$ and we have used the fact that $\tilde c_i(\cR)$ depends only on the shape of $\cR$.

  Let $r = |A \cap B|$ where $A,B$ are the distinguished sets of vertices for $U_j^{(s_\ell,s_r,t)}$, and let $\tilde p$ be the number of vertex-disjoint paths from $A \setminus B$ to $B \setminus A$,
  so that $p = r + \tilde p$.
  We can apply Lemma~\ref{lem:single-topology-matrix} and our bound on $\tilde c_i$ to get that with probability $1 - O(n^{-100 \log n})$,
  \begin{align*}
    \Norm{M_j^{(s_\ell,s_r,t)}} & \leq \Paren{\frac \omega n}^{t - s} \cdot n^{\frac{\tilde p + r - i/2}{2} + \epsilon{s}} \cdot n^{\frac{t - \tilde p - r}{2}} \cdot 2^{O(t)} \cdot (\log n)^{O(t - r + \tilde p)}\\
    & = n^{-\epsilon(t - s) - i/4} \cdot 2^{O(t)} \cdot (\log n)^{O(t - r + \tilde p)}\\
    & = n^{-\epsilon(t - s) - i/4} \cdot 2^{O(t)} \cdot (\log n)^{O(t - s)}\mcom
  \end{align*}
  where in the last step we have used that $t \geq 2s - r$ and $\tilde p \leq s - r$.

  By inspection,
  \[
    \cQ_i = \sum_{\substack{s_\ell, s_r \leq d \\ t \leq \tau \\ j \leq q}} M_j^{(s_\ell,s_r,t)}\mper
  \]
  For a fixed $t$ there are at most $2^{{t \choose 2} + O(t)}$ choices for $U$, so $q \leq 2^{{t \choose 2} + O(t)}$.
  Now we fix $s_\ell,s_r$ and sum over $t$ to obtain the block of $\cQ_i$ corresponding to size-$s_\ell$ and size-$s_r$ subsets.
  By triangle inequality and a union bound, with probability at least $1 - O(n^{-97 \log n})$,
  \[
    \Norm{\sum_{\substack{t \leq \tau \\ j \leq q}} M_j^{(s_\ell,s_r,t)}} \leq 2^{{t \choose 2} + O(t)} \cdot n^{-\epsilon(t - s) - i/4} \cdot 2^{O(t)} \cdot (\log n)^{O(t - s)}\mper
  \]
  From our assumptions on $d,\tau,$ and $\epsilon$, this is at most $2^{{s_\ell \choose 2}/2 + {s_r \choose 2}/2}/100 d^3$.

  As usual, let $\Pi$ be the projector to $\Span \{ e_C \, : \, C \in \cC_{\leq d} \}$.
  Note that $\Pi \cQ_i = \cQ_i \Pi = \cQ_i$, since $\cQ_i(I,J) = 0$ whenever $I$ or $J$ is not a clique.
  So, to show that $D/8d \succeq \cQ_i \succeq -D/8d$, it is sufficient to show that for all vectors $v$ with $v = \Pi v$ it happens that $|v^\top{\cQ_i}v| \leq v^{T}(D/8d)v$. To see this, let $v_k$ be the part of $v$ indexed by cliques of size exactly $k$.
  Now,
\begin{align*}
|v^\top{\cQ_i}v| &\leq \sum_{k_1 = 0}^{d}{\sum_{k_2 = 0}^{d}
{\Norm{v_{k_1}}\Norm{\sum_{\substack{t \leq \tau \\ j \leq q}} M_j^{(k_1,k_2,t)}}\Norm{v_{k_2}}}} \\
&\leq \sum_{k_1 = 0}^{d}{\sum_{k_2 = 0}^{d}{\frac{1}{100d^3}\left(2^{\binom{k_1}{2}/2 + \binom{k_2}{2}/2 }\Norm{v_{k_1}}\Norm{v_{k_2}}\right)}} \\
&\leq \sum_{k_1 = 0}^{d}{\sum_{k_2 = 0}^{d}{\frac{1}{200d^3}\left(2^{\binom{k_1}{2}}\Norm{v_{k_1}}^2+2^{\binom{k_2}{2}}{\Norm{v_{k_2}}^2}\right)}} \\
&\leq \sum_{k = 0}^{d}{\frac{2^{\binom{k}{2}}}{100d^2}\Norm{v_{k}}^2} \leq v^{\top}(D/8d)v
\end{align*}
\end{proof}

\subsubsection{Coefficient Decay in the Factorization: Proof of Lemma~\ref{lem:ci-bound}}
We turn to the proof of Lemma~\ref{lem:ci-bound}, for which we want the following characterization of the effect of the separating factorization on the underlying graph of a ribbon.

We require the following combinatorial quantities:

\begin{center}
  \fbox{\begin{minipage}{\textwidth}
    \paragraph{Definitions for Lemma~\ref{lem:graph-tradeoff}}
      \begin{enumerate}
        \item $I,J,S_\ell, S_r \subseteq [n]$ of size at most $d$.
        \item Ribbons $\cR_\ell, \cR_m, \cR_r$ satisfying \ref{itm:left},\hyperref[itm:middle-sep]{3*},\ref{itm:right} but not \ref{itm:disjoint} for $S_\ell, S_r, I,J \subseteq [n]$. (Remember that $\cR_m$ may be improper.)
        \item Ribbons $\cR_\ell', \cR_m', \cR_r'$ which are the separating factorization of $\cR_\ell, \cR_m, \cR_r$, with separators $S_\ell', S_r'$.
        (Remember that $\cR_m'$ may be improper.)
        \item $p$, the number of vertex-disjoint paths from $S_\ell$ to $S_r$ in $\cR_m$.
        \item $p'$, the number of vertex-disjoint paths from $S_\ell'$ to $S_r'$ in $\cR_m'$.
        \item $r = (|\cV(\cR_\ell)| + |\cV(\cR_m)| + |\cV(\cR_\ell)| - |S_\ell| - |S_r|) - (|\cV(\cR_\ell')| + |\cV(\cR_m')| + |\cV(\cR_\ell')| - |S_\ell'| - |S_r'|)$, the number of intersections among $R_\ell, R_m, R_r$.
        \item $\fD = \cZ(\cR_m') \setminus \cZ(\cR_m)$, the newly degree-$0$ (we write \emph{isolated}) vertices in $\cR_m'$.
        \item $\fU \subseteq \cV(\cR_\ell) \cup \cV(\cR_m) \cup \cV(\cR_r)$, the set of vertices appearing in more than one of $\cV(\cR_\ell), \cV(\cR_m)$, and $\cV(\cR_r)$. Note that $\fU \subseteq \cV(\cR'_m)$.
      \end{enumerate}
  \end{minipage}}
\end{center}

\begin{lemma}\label{lem:graph-tradeoff}
  \[
    \underbrace{|S'_\ell| + |S'_r| - (|S_\ell| +|S_r|)}_{\text{increase in separator size}} + \underbrace{p - p'}_{\text{lost paths between separators}} + \underbrace{|\fD|}_{\text{new isolated vertices}} \leq \underbrace{r}_{\text{number of intersections}}\mper
  \]
\end{lemma}
The following series of claims will help us in the proof of Lemma~\ref{lem:graph-tradeoff}
  \begin{claim}\label{clm:graph-1}
    $I \cap \cV(\cR_m') \subseteq S_\ell'$ and $J \cap \cV(\cR_m') \subseteq S_r'$.
  \end{claim}
  \begin{proof}[Proof of claim]
    If $u \in I \cap \cV(\cR_m')$ then since $I \subseteq \cV(\cR_\ell')$, we have $u \in \cV(\cR_\ell') \cap \cV(\cR_m') = S_\ell'$, and similarly for the second part.
  \end{proof}
  Next we have a simple analysis of which vertices may possibly be newly isolated.
  \begin{claim}\label{clm:graph-2}
    $\fD \subseteq \fU$.
  \end{claim}
  \begin{proof}[Proof of claim]
    Let $u \in \fD$.
    If $u \in S_\ell$ or $u \in S_r$ we are done.
    Otherwise, if $u \in I$ or $u \in J$, then $u$ appeared in more than one of $\cV(\cR_\ell), \cV(\cR_m),\cV(\cR_r)$ by the definition of the canonical factorization.

    If neither of these cases hold, then $u$ was incident to an edge in at least one of $\cR_\ell, \cR_m, \cR_r$.
    Since that edge does not exist in $\cR_m'$, it must have appeared at least twice among the edge sets of $\cR_\ell, \cR_m, \cR_r$, and therefore $u$ appeared at least twice among the vertex sets, thus proving the claim.
  \end{proof}
  Next we show that some vertices in $\fU$ cannot become isolated.
  \begin{claim}\label{clm:graph-3}
    By Menger's theorem, there are $|S_\ell'|$ vertex-disjoint paths from $\fU \cap \cV(\cR_\ell)$ to $I$ in $\cR_\ell$.
    Let $u_\ell^{(1)},\ldots,u_\ell^{(|S_\ell'|)}$ be distinct vertices so that $u^{(i)}$ is the last vertex in $\fU$ along the $i$-th vertex disjoint path.
    Let $u_r^{(1)},\ldots,u_r^{(|S_r'|)}$ be similarly defined.
    None of the vertices $u$ may be in $\fD$.
  \end{claim}
  \begin{proof}[Proof of claim]
    Fix one of these vertices $u$, and consider its neighbor $v$ one step farther along the path to $I$ (or $J$).
    By definition, the vertex $v$ does not appear in more than one of $\cV(\cR_\ell),\cV(\cR_m),\cV(\cR_r)$.
    If $v \in \cR_m'$, then the edge $(u,v)$ must be in $\cR_m'$, and so $u$ is not isolated in $\cR_m'$.
    If $v \notin \cR_m'$, then $u$ must be in $S_\ell' \cup S_r'$, in which case by definition $u \notin \fD$.
  \end{proof}
  We set up sets $q$ of vertices to divide up the intersecting vertices among $\cR_\ell, \cR_m,\cR_r$ according to which ribbons witness the intersection.
  \begin{claim}\label{clm:graph-4}
  Let
  \begin{align*}
    q_{\ell,m,r} & \defeq (\cV(\cR_r) \cap \cV(\cR_m) \cap \cV(\cR_\ell)) \setminus (S_\ell \cup S_r)\\
    q_{\ell,r} & \defeq (\cV(\cR_\ell) \cap \cV(\cR_r)) \setminus \cV(\cR_m)\\
    q_{\ell,m} & \defeq (\cV(\cR_\ell) \cap \cV(\cR_m)) \setminus (S_\ell \cup \cV(\cR_r))\\
    q_{r,m} & \defeq (\cV(\cR_r) \cap \cV(\cR_m)) \setminus (S_r \cup \cV(\cR_\ell))\mper
  \end{align*}
   The sets $q$ are pairwise disjoint, and
   \[
     r = 2|q_{\ell,m,r}| + |q_{\ell,r}| + |q_{\ell,m}| + |q_{r,m}| + |S_\ell \cap (\cV(\cR_r) \setminus S_r)| + |S_r \cap (\cV(\cR_\ell) \setminus S_\ell)|\mper
   \]
   Also, $\fU = q_{\ell,m,r} \cup q_{\ell,r} \cup q_{\ell,m} \cup q_{r,m} \cup S_\ell \cup S_r$.
   \end{claim}
   \begin{proof}
     By inspection.
   \end{proof}

We are prepared to prove Lemma~\ref{lem:graph-tradeoff}.
\begin{proof}[Proof of Lemma~\ref{lem:graph-tradeoff}]
  We start by bounding the number of vertices in $\fU \setminus \fD$.
  By Claim~\ref{clm:graph-3}, there are at least $|\{u_\ell^{(1)},\ldots,u_\ell^{(|S_\ell'|)},u_r^{(1)},\ldots,u_r^{|S_r'|} \}|$ such vertices.

  Let $a$ be the number of pairs $i,j$ so that $u_\ell^{(i)} = u_r^{(j)}$.
  Then there are vertex-disjoint paths $w_1,\ldots,w_a$ from $S_\ell'$ to $S_r'$.
  The path $w$ corresponding to $u_\ell^{(i)} = u_r^{(j)}$ is given by following $u_\ell^{(i)}$'s path from $I$ to $\cU$, ending at $u_\ell^{(i)}$, then following $u_r^{(j)}$'s path from $\fU$ to $J$.
  This gives a path from $I$ to $J$, which must have a subpath from $S_\ell'$ to $S_r'$.

  Now consider the $p$ vertex-disjoint paths from $S_\ell$ to $S_r$ in $\cR_m$.
  We claim that
  \begin{align}
    p - |S_\ell \cap S_r| &\leq |q_{\ell,m,r}| + |S_{\ell} \cap \cV(\cR_r)  \setminus S_r| + |S_r \cap \cV(\cR_{\ell}) \setminus S_{\ell}| \nonumber \\
    & + |\fU \setminus (\{u_\ell^{(1)},\ldots,u_\ell^{(|S_\ell'|)},u_r^{(1)},\ldots,u_r^{|S_r'|} \} \cup \fD)| + (p' - a)\label{eq:graph-1}
  \end{align}
  In words, every nontrivial path from $S_\ell$ to $S_r$ contributes to at least one of:
  \begin{itemize}
    \item $|q_{\ell,m,r}|$, the number of $3$-way intersections,
    \item intersections between $S_{\ell}$ and $\cV(\cR_r)$ (but not $S_r$), intersections between $\cV(\cR_{\ell})$ and $S_r$ (but not $S_{\ell}$),
    \item vertices in $\fU$ which are guaranteed not to become isolated (and which we have not yet accounted for), or
    \item vertex-disjoint paths from $S_\ell'$ to $S_r'$ (which we have not yet accounted for).
  \end{itemize}

    Fix one such path.
  If it intersects $q_{\ell,m,r}$, $S_l \cap \cV(\cR_r)$, or $S_r \cap \cV(\cR_l)$ we are done, so suppose otherwise.
  If it is contained entirely in $q_{\ell,m} \cup q_{r,m} \cup (S_\ell \setminus \cV(\cR_r)) \cup (S_r \setminus \cV(\cR_l))$, then there is some edge along the path connecting a vertex in $\cV(\cR_\ell) \cap \cV(\cR_m) \setminus \cV(\cR_r)$ with one in $\cV(\cR_r) \cap \cV(\cR_m) \setminus \cV(\cR_{\ell})$.
  That edge can occur nowhere else among $\cR_\ell, \cR_m, \cR_r$, and so the incident vertices must not be in $\fD$.
  At the same time, if there is any vertex along the path which is outside $\fU$, then the nearest vertices along the path to either side which do lie in $\fU$ also must be outside $\fD$.

  In either case, there are two vertices along the path in $\fU \setminus \fD$.
  If either of these is not among the $u$ vertices, we are done.
  If both are, then by definition of the $u$ vertices this creates a path from $I$ to $J$, and so from $S_\ell'$ to $S_r'$.
  Furthermore, this path must be vertex disjoint from the paths $w_1,\ldots,w_a$ previously constructed, since the $u$ vertices involved in those paths were $\cV(\cR_\ell) \cap \cV(\cR_r)$.
  This proves \eqref{eq:graph-1}.

  It's time to put things together.
  By Claim~\ref{clm:graph-2}, we can bound $|\fD|$ by
  \[
    |\fD| \leq |\fU| - |\fU \setminus \fD|.
  \]
  We have $|\fU \setminus \fD| \geq |S_\ell'| + |S_r'| - a + |\fU \setminus (\{u_\ell^{(1)},\ldots,u_\ell^{(|S_\ell'|)},u_r^{(1)},\ldots,u_r^{|S_r'|} \} \cup \fD)|$, and $|\fU| = |q_{\ell,m,r}| + |q_{\ell,r}| + |q_{\ell,m}| + |q_{r,m}| + |S_\ell \cup S_r|$.
  This gives us
  \[
    |\fD| \leq |q_{\ell,m,r}| + |q_{\ell,r}| + |q_{\ell,m}| + |q_{r,m}| + |S_\ell \cup S_r| - |S_\ell'| - |S_r'| + a - |\fU \setminus (\{u_\ell^{(1)},\ldots,u_\ell^{(|S_\ell'|)},u_r^{(1)},\ldots,u_r^{|S_r'|} \} \cup \fD)|\mper
  \]
  Adding \eqref{eq:graph-1} to both sides and rearranging, we get
  \[
    p - p' + |\fD| \leq 2|q_{\ell,m,r}| + |S_{\ell} \cap (\cV(\cR_r) \setminus S_r)| + |S_{r} \cap (\cV(\cR_{\ell})  \setminus S_{\ell})| + |q_{\ell,r}| + |q_{\ell,m}| + |q_{r,m}| + |S_\ell \cup S_r| - |S_\ell'| - |S_r'| + |S_\ell \cap S_r|\mcom
\]
and substituting $r = 2|q_{\ell,m,r}| + |S_{\ell} \cap (\cV(\cR_r)  \setminus S_r)| + |S_{r} \cap (\cV(\cR_{\ell})  \setminus S_{\ell})| + |q_{\ell,r}| + |q_{\ell,m}| + |q_{r,m}|$ gives
\[
  p - p' + |\fD| \leq r + |S_\ell \cup S_r| - |S_\ell'| - |S_r'| + |S_\ell \cap S_r| \mper
\]
Notice that $|S_\ell \cup S_r| + |S_\ell \cap S_r| = |S_\ell| + |S_r|$, so we can rearrange to obtain the lemma.
\end{proof}

Now we can prove Lemma~\ref{lem:ci-bound}.
\begin{proof}[Proof of Lemma~\ref{lem:ci-bound}]
  First of all, we note that $c_i(\cR_m)$ depends only on the shape of $\cR_m$ by symmetry of our construction.
  We turn to the quantitative bound.

  The proof is by induction.
  The coefficients $c_0(\cR_m)$ are nonzero only for ribbons $\cR_m$ which have $\cZ(\cR_m) = \emptyset$ and admitting $|S_\ell| = |S_r| = p$ paths from $S_\ell$ to $S_r$.
  Thus in the case that $i = 0$, the statement reduces to $c_0(\cR_m) \leq 1$, which is true by definition.

  Suppose the lemma holds for $c_i$, and consider $c_{i+1}$.
  By definition, for an (improper) $S_\ell',S_r'$-ribbon $\cR_m'$ and ribbons $\cR_\ell', \cR_r'$ satisfying \ref{itm:left} and \ref{itm:right},
  \begin{align}
    c_{i+1}(\cR_m') = 
    \sum_{\substack{\cR_\ell, \cR_m, \cR_r \text{ satisfying}\\\text{\ref{itm:left},\hyperref[itm:middle]{3*},\ref{itm:right} and not \ref{itm:disjoint} for some } S_\ell, S_r\\
    r \text{ intersections outside } S_\ell, S_r\\
    \\ \text{separating factorization } \cR_\ell', \cR_m', \cR_r', S_\ell', S_r'}}
    c_i(\cR_m) \Paren{\frac \omega n}^{r}\mper \label{eq:5-1}
  \end{align}
  We introduce the shorthand $s' = \frac{|S_\ell'| + |S_r'|}{2}$.
  Consider first a particular term in the sum, $c_i(\cR_m) (\omega/n)^r$, where $\cR_m$ is an improper $S_\ell, S_r$ ribbon, and let $|\fD| = |\cZ(\cR_m') \setminus \cZ(\cR_m)|$.
  By induction and Lemma~\ref{lem:graph-tradeoff},
  \begin{align*}
  \Paren{\frac \omega n}^r \cdot c_i(\cR_m)
  & \leq \Paren{\frac \omega n}^r \cdot \Paren{\frac \omega n}^{s} \cdot n^{\frac { p - |\cZ(\cR_m)| - i/2}{2} + \epsilon s} \quad \text{ by induction}\\
  & = \Paren{\frac \omega n}^{s'} \cdot \Paren{\frac \omega n}^{r - s' + s} \cdot \cdot n^{\frac { p - |\cZ(\cR_m)| - i/2}{2} + \epsilon s}\\
  & = \Paren{\frac \omega n}^{s'} \cdot n^{-\epsilon(r - s' + s)} \cdot n^{-\tfrac 1 2 (r - s' + s)}\cdot n^{\frac { p - |\cZ(\cR_m)| - i/2}{2} + \epsilon  s} \quad \text{ using $\omega = n^{1/2 - \epsilon}$}\\
  & \leq \Paren{\frac \omega n}^{s'} \cdot n^{-\epsilon(r - s' + s)} \cdot n^{-\tfrac 1 2 ( s' - s + p - p' + |\fD|)}\cdot n^{\frac { p - |\cZ(\cR_m)| - i/2}{2} + \epsilon s} \quad \text{ by Lemma~\ref{lem:graph-tradeoff}}\\
  & = \Paren{\frac \omega n}^{s'} \cdot n^{-\epsilon(r - s' + s)} \cdot n^{\frac{p' - |\cZ(\cR_m')| - i/2 -s' +s}{2} + \epsilon s} \quad \text{ canceling terms, using $|\cZ(\cR_m')| = |\fD| + |\cZ(\cR_m)|$}\\
  & = n^{-\epsilon r} \cdot \Paren{\frac \omega n}^{s'} \cdot n^{\frac{p' - |\cZ(\cR_m')| - i/2 - (s' - s)}{2} + \epsilon s'}\\
  & \leq n^{-\epsilon r} \cdot \Paren{\frac \omega n}^{s'} \cdot n^{\frac{p' - |\cZ(\cR_m')| - (i+1)/2}{2} + \epsilon s'} \quad \text{ using $s' -s \geq 1/2$, by Lemma~\ref{lem:sep-size-incr}}\\
  \end{align*}
  
  Next we assess how many nonzero terms are in the sum \eqref{eq:5-1} for a fixed $r$ and a fixed $\cR'_m$. For each vertex of $\cR'_m$, there are 7 possibilities for which ribbon(s) it came from in $\{\cR_{\ell},\cR_m,\cR_r\}$ so there are at most $7^{\tau}$ choices overall (recall that $\cR'_m$ has at most $\tau$ vertices for the terms we are looking at). Once we have chosen which ribbon(s) each vertex of $\cR'_m$ came from, everything is fixed except for possible edges of $\cR'_m$ which appear at least twice in $\cR_{\ell}$, $\cR_m$, and $\cR_{r}$. There are two possibilities for each possible edge of $\cR'_m$ which appears twice in $\cR_{\ell}$, $\cR_{m}$, and $\cR_{r}$ and four possibilities for each possible edge of $\cR'_m$ which appers three times in $\cR_{\ell}$, $\cR_{m}$, and $\cR_{r}$. However, note that any such edge must be between an intersected vertex and either another intersected vertex or a vertex in $S_{\ell} \cup S_{r}$. Thus, there are at most $r\tau$ possible edges of $\cR'_m$ which appear at least twice in $\cR_{\ell}$, $\cR_{m}$, and $\cR_{r}$ and the total number of possibilities for these edges is at most $4^{r\tau}$. 
  
  All together there are at most $2^{O(r\tau)}$ nonzero terms for fixed $r$. This means that the total contribution from such terms is at most
  \[
  2^{O(r\tau)} \cdot n^{-\epsilon r} \cdot \Paren{\frac \omega n}^{s'} \cdot n^{\frac{p' - |\cZ(\cR_m')| - (i+1)/2}{2} + \epsilon s'}
  \]
  As long as $\tau \leq (\epsilon / C) \log n$ for some universal constant $C$, we have $2^{O(r\tau)}\cdot  n^{-\epsilon r} \ll 1/\tau$ for all $r \geq 1$.
  All in all, we obtain
  \[
  c_{i+1}(\cR_m') \leq \Paren{\frac \omega n}^{s'} \cdot n^{\frac{p' - |\cZ(\cR_m')| - (i+1)/2}{2} + \epsilon s'} 
  \]
  which completes the induction.
\end{proof}

\subsection{$\cL \cL^\top$ is Well-Conditioned---Proof of Lemma~\ref{lem:Lbound}}
In this section we prove Lemma~\ref{lem:Lbound}, restated here.

\restatelemma{lem:Lbound}
\begin{proof}[Proof of Lemma~\ref{lem:Lbound}]
  We recall the definition of $\cL$.
  \[
    \cL(I,S) = \Paren{\frac \omega n}^{-\frac{|S|}{2}} \sum_{\substack{\cR \text{ having \ref{itm:left}} \\ |\V(\cR_\ell)| \leq \tau}} \Paren{\frac \omega n}^{|\V(\cR_\ell)|} \chi_{\cR_\ell}\mper
  \]
  Consider a diagonal entry $\cL(S,S)$.
  Since every ribbon $\cR$ appearing in its expansion must have \ref{itm:left}, in particular it has no edges inside $S$.
  Thus, by the same argument as in Lemma~\ref{lem:normalization}, with probability at least $1 - O(n^{-10 \log n})$,
  \[
    \cL(S,S) = \Paren{\frac \omega n}^{\frac{|S|}{2}} (1 \pm n^{-\Omega(\epsilon)})\mper
  \]

  Let $\cL^{\text{off-diag}}$ be given by
  \[
    \cL^{\text{off-diag}}(I,S) = \begin{cases} \cL(I,S) \text{ if $I \neq S$}\\ 0 \text{ otherwise} \end{cases}.
  \]
  We will consider the block of $\cL^{\text{off-diag}}$ with rows indexed by sets of size $s_\ell$ and columns indexed by sets of size $s_r$ for some $s_\ell, s_r \leq d$.
  For a fixed $t \leq \tau$, let $U_1^{(s_\ell,s_r,t)}, \ldots, U_q^{(s_\ell, s_r,t)}$ be all the graphs on vertex set $[t]$ with distinguished subsets of vertices $A,B$ of size $s_\ell, s_r$ respectively, and where
  \begin{itemize}
    \item $A \neq B$,
    \item there are no edges inside $B$,
    \item every vertex in $U$ outside $A \cup B$ is reachable from $A$ without passing through $B$, and 
    \item $B$ is the unique minimum-size vertex separator in $U$ separating $A$ from $B$.
  \end{itemize}
  Then let $M_i^{(s_\ell,s_r,t)}$ be given by
  \[
    M_i^{(s_\ell,s_r,t)}(I,S) = \Paren{\frac \omega n}^{t - \frac{s_r}{2}} \cdot \sum_{\cR \text{ an $(I,S)$-ribbon with shape $U_i^{(s_\ell,s_r,t)}$}} \chi_{\cR}\mper
  \]
  By assumption on $U_i^{(s_\ell,s_r,t)}$, there are $s_r$ vertex-disjoint paths from $A$ to $B$.
  Let $r = |A \cap B|$.
  By Lemma~\ref{lem:single-topology-matrix}, with probability at least $1 - O(n^{-100 \log n})$,
  \begin{align*}
    \Norm{M_i^{(s_\ell,s_r,t)}} & \leq \Paren{\frac \omega n}^{\frac{s_r}{2}} \cdot \Paren{\frac \omega n}^{t - s_r}  \cdot n^{\frac{t - s_r}{2}} \cdot 2^{O(t)} \cdot (\log n)^{O(t - r + (s_r - r))}\\
    & = \Paren{\frac \omega n}^{\frac{s_r}{2}} \cdot n^{-\epsilon( t- s_r)} \cdot 2^{O(t)} \cdot (\log n)^{O(t - s_r)}\mcom
  \end{align*}
  where in the last step we have used that $t \geq s_\ell + s_r - r$ and $s_r \leq s_\ell$, which holds by the vertex-separator requirement on $B$.
  There are at most $2^{{t \choose 2} - {s_r \choose 2} + O(t)}$ choices for $U_i^{(s_\ell,s_r,t)}$ when $s_\ell,s_r,t$ are fixed, by the requirement that $U$ have no edges inside $B$.
  Summing over all $q$ for a fixed $t$, we get by triangle inequality
  \[
    \Norm{\sum_{i \leq q} M_i^{(s_\ell,s_r,t)}} \leq \Paren{\frac \omega n}^{\frac{s_r}{2}} \cdot 2^{{t \choose 2} - {s_r \choose 2} + O(t)} \cdot n^{-\epsilon(t - s_r)} \cdot (\log n)^{O(t - s_r)}
  \]
  with probability $1 - O(n^{-99 \log n})$.
  By our assumptions on $d, \tau,$ and $\epsilon$, this is at most $(\omega/n)^{s_r/2} \cdot 1/d^4$.

  The following standard manipulations now prove the lemma.
  Let $D' \in \R^{\nchoose{\leq d}}$ be the diagonal matrix with $D'(S,S) = (\omega/n)^{|S|/2}$ if $S$ is a clique in $G$ and $0$ otherwise.
  Then we can decompose $\cL = D + E + \cL^{\text{off-diag}}$, where $E$ is a diagonal matrix with $|E(S,S)| \leq n^{-\Omega(\epsilon)} \cdot (\omega/n)^{|S|/2}$.
  Then we have
  \begin{align*}
    \Pi \cL \Pi \cL^\top \Pi & = D^2 \\
                             & + \Pi(D \Pi \cL^{\text{off-diag}} + D \Pi E + E \Pi D + E \Pi \cL^{\text{off-diag}} + \cL^{\text{off-diag}} \Pi D + \cL^{\text{off-diag}} \Pi E \\
                             & + E \Pi E + \cL^{\text{off-diag}} \Pi \cL^{\text{off-diag}}) \Pi
  \end{align*}
  Each of the above matrices aside from $D^2$ is a $d \times d$ block matrix, where the $(s_\ell, s_r)$ block is $\nchoose{s_\ell} \times \nchoose{s_r}$ dimensional and has norm at most $(\omega/n)^{(s_\ell + s_r)/2} \cdot d^{-4}$.
  By the same argument as in the proof of Lemma~\ref{lem:Qismall}, using Cauchy-Schwarz to combine the $d^2$ blocks, we obtain the lemma.
\end{proof}

\subsection{High-Degree Matrices Have Small Norms}
\label{sec:xi}
In this section we prove Lemma~\ref{lem:xi-bound}, restated here:
\restatelemma{lem:xi-bound}

We recall the definition of $\xi_i$.
For a coefficient function on ribbons $c_{i-1}(\cR_m)$, we have a matrix $\cE$ given by
\begin{align*}
  & \cE(I,J) = \nonumber \\
  & \sum_{\substack{S_\ell, S_r \subseteq [n] \\ |S_\ell|, |S_r| \leq d}} \Paren{\frac \omega n }^{-\frac{|S_\ell| + |S_r|}{2}}\sum_{\substack{\cR_\ell, \cR_m, \cR_r \text{ satisfying}\\\text{\ref{itm:left},\hyperref[itm:middle]{3*},\ref{itm:right} and not \ref{itm:disjoint}}  \\
  |\V(\cR_\ell)|, |\V(\cR_m)|, |\V(\cR_r)| \leq \tau \\ \text{separating factorization} \\ \cR_\ell', \cR_m', \cR_r', S_\ell', S_r'}} c_{i-1}(\cR_m) \Paren{\frac \omega n}^{|\V(\cR_\ell)| + |\V(\cR_r)| + |\V(\cR_m)|-\frac{|S_\ell| + |S_r|}{2}} \cdot \chi_{\cR_\ell'} \cdot \chi_{\cR_m'} \cdot \chi_{\cR_r'}\mcom
  \end{align*}
  and another one, $\cE'$, given by
\begin{align*}
  & \cE'(I,J) = \\
  & \sum_{\substack{S_\ell, S_r \subseteq [n] \\ |S_\ell|, |S_r| \leq d}} \Paren{\frac \omega n }^{-\frac{|S_\ell| + |S_r|}{2}}\sum_{\substack{\cR_\ell, \cR_m, \cR_r \text{ satisfying}\\\text{\ref{itm:left},\hyperref[itm:middle]{3*},\ref{itm:right} and not \ref{itm:disjoint}}  \\
  \text{separating factorization} \\ \cR_\ell', \cR_m', \cR_r', S_\ell', S_r' \\
|\V(\cR_\ell')|, |\V(\cR_m')|, |\V(\cR_r')| \leq \tau }}  c_{i-1}(\cR_m) \Paren{\frac \omega n}^{|\V(\cR_\ell)| + |\V(\cR_r)| + |\V(\cR_m)|-\frac{|S_\ell| + |S_r|}{2}} \cdot \chi_{\cR_\ell'} \cdot \chi_{\cR_m'} \cdot \chi_{\cR_r'}\mper
  \end{align*}
Then the matrix $\xi_i$ is given by $\cE - \cE'$.

We will actually prove a bound on the Frobenious norm of each matrix $\xi_i$.
The following will allow us to control the magnitude of the entries.
It follows immediately from our concentration bound Lemma~\ref{lem:sum-parities}, which is proved via the moment method.
(Under the slightly stronger assumption $\tau \ll \epsilon \log n / \log \log n$, it would also follow from standard hypercontractivity.)

\begin{lemma}\label{lem:high-deg-scalar}
  Suppose $c_{T}$ are a collection of coefficients, one for each $T \subseteq \nchoose{2}$, and there is a constant $C$ such that
  \begin{enumerate}
  \item If $|T| > C \tau$ then $c_T = 0$.
  \item Otherwise, $|c_T| \leq (\omega/n)^{|T|/C - Cd}$.
  \end{enumerate}
  Then with probability at least $1 - O(n^{-100 \log n})$ it occurs that $\Abs{\sum_{T \subseteq \nchoose{2}} c_T \cdot \chi_T} \leq n^{-20d}$.
\end{lemma}

We will also need several facts about the coefficients of ribbons in the expansion of each matrix $\xi_i$.

\begin{lemma}\label{lem:xi-coeff}
  Every triple $\cR_\ell, \cR_m, \cR_r$ appearing with nonzero coefficient in $\xi_c$ satisfies $|\cV(\cR_\ell)| + |\cV(\cR_m)| + |\cV(\cR_r)| = \Theta(\tau)$.
\end{lemma}
\begin{proof}
  To appear with nonzero coefficient, the triple $\cR_\ell, \cR_m, \cR_r$ with separating factorization $\cR_\ell', \cR_m', \cR_r'$ must either have
  \[
    |\cV(\cR_\ell)|, |\cV(\cR_m)|, |\cV(\cR_r)| \leq \tau \quad \text{but} \quad |\cV(\cR_\ell')| > \tau \text{ or } |\cV(\cR_m')| > \tau \text{ or } |\cV(\cR_\ell')| > \tau\mcom
  \]
  or
  \[
    |\cV(\cR_\ell')|, |\cV(\cR_m')|, |\cV(\cR_r')| \leq \tau \quad \text{but} \quad |\cV(\cR_\ell)| > \tau \text{ or } |\cV(\cR_m)| > \tau \text{ or } |\cV(\cR_\ell)| > \tau\mper
  \]
  In the first case, we must have one of $|\cV(\cR_\ell)| \geq \tau/3$ or $|\cV(\cR_m)| \geq \tau /3$ or $|\cV(\cR_r)| \geq \tau/3$.
  In the second, we must have $|\cV(\cR_\ell)|, |\cV(\cR_m)|, \cV(\cR_r)| \leq 3\tau$.
\end{proof}

We are prepared to prove Lemma~\ref{lem:xi-bound}.
\begin{proof}[Proof of Lemma~\ref{lem:xi-bound}]
  We will apply Lemma~\ref{lem:high-deg-scalar} to $\xi_i(I,J)$ for each $i \leq 2d$ and $I,J \subseteq [n]$ with $|I|,|J| \leq d$.
  So consider the Fourier expansion of $\xi_i(I,J)$, given by
  \[
    \xi_i(I,J) = \sum_{T \subseteq \nchoose{2}} c_T \cdot \chi_T\mper
  \]
  From Lemma~\ref{lem:xi-coeff}, we obtain that if $|T| > C \tau$ then $c_T = 0$, for some absolute constant $C$.
  For smaller $T$ we need a bound on the magnitude $|c_T|$.
  The coefficient $c_T$ is bounded by
   \begin{align}
   |c_T| \leq
   \sum_{\substack{S_\ell, S_r \subseteq [n] \\ |S_\ell|, |S_r| \leq d}} \Paren{\frac \omega n }^{-\frac{|S_\ell| + |S_r|}{2}}\sum_{\substack{\cR_\ell, \cR_m, \cR_r \\ \text{nonzero in $\xi_i(I,J)$ as in \ref{lem:xi-coeff}}\\ \chi_{\cR_\ell} \cdot \chi_{\cR_m} \cdot \chi_{\cR_r} = \chi_T }}
  c_{i-1}(\cR_m) \Paren{\frac \omega n}^{|\V(\cR_\ell)| + |\V(\cR_r)| + |\V(\cR_m)|-\frac{|S_\ell| + |S_r|}{2}}\label{eq:xi-1}
  \end{align}
  By Lemma~\ref{lem:ci-bound}, we have $c_{i-1}(\cR_m) \leq n^{d} \leq (\omega/n)^{-2d}$.
  At the same time, there are at most $2^{O(\tau^2)}$ nonzero terms in the sum \eqref{eq:xi-1}.
  Thus by Lemma~\ref{lem:xi-coeff} and our assumptions on $d,\tau,$ and $\epsilon$, the coefficient $c_T$ is at most $(\omega/n)^{\tau/C - Cd}$ for some absolute constant $C$.

  Applying Lemma~\ref{lem:high-deg-scalar}, we obtain $|\xi_i(I,J)| \leq n^{-20d}$ with probability $1 - O(n^{-100 \log n})$.
  Taking a union bound over all $n^{2d} \leq n^{2\log n}$ entries of $\xi_i$, and over all $i \leq 2d$, we obtain that $\|\xi_0 - \ldots + \xi_{2d}\| \leq \|\xi_0 - \ldots + \xi_{2d}\|_F \leq n^{-16d}$ with probability $1 - O(n^{-96 \log n})$.
\end{proof}

\section*{Acknowledgements}

We thank Raghu Meka, Ryan O'Donnell, Prasad Raghavendra, Tselil Schramm, David Steurer, and Avi Wigderson for many useful discussions related to this paper. 

\addreferencesection 
\bibliographystyle{amsalpha}
\bibliography{mr,dblp,scholar}

\appendix


\section{Omitted Proofs}
\subsection{Calibration of $\pE$}
\label{sec:calibration}
In this subsection we prove Lemma~\ref{lem:pE-fools-simple-tests}, restated here.
\restatelemma{lem:pE-fools-simple-tests}
\begin{proof}
  The proof is straightforward by expanding the coefficients $f$ in the Fourier basis.
  For $S \subseteq [n]$, let $c_S : G \mapsto \R$ be maps so that $f_G(x) = \sum_{S \subseteq [n]} c_S \cdot x_S$.
  \begin{align*}
    \E_{G \sim \G(n,\frac{1}{2})} [\pE[f_G(x)]] & =  \E_{G \sim \G(n,\frac{1}{2})} \Brac{ \pE\Brac{ \sum_{S \subseteq [n]} c_S \cdot x_S}}\\
    & = \sum_{S \subseteq [n]}  \E_{G \sim \G(n,\frac{1}{2})}  \Brac{c_S \pE [x_S]}\\
    & = \sum_{S \subseteq [n]} \E_{G \sim \G(n,\frac{1}{2})}\Brac{\sum_{T,T' \subseteq \nchoose{2}} \widehat{c_S}(T) \widehat{\pE[x_S]}(T') \cdot \chi_T \chi_{T'}}\\
    & = \sum_{S \subseteq [n]} \sum_{T} \widehat{c_S}(T) \E_{(H,x) \sim G(n,1/2,\omega)} \Brac{\chi_T(H) \cdot x_S}\\
    & = \E_{(H,x) \sim G(n,1/2,\omega)} \Brac{ \sum_{S \subseteq [n]} \sum_T \widehat{c_S}(T) \chi_T(H) \prod_{i \in S}x_i}\\
    & = \E_{(H,x) \sim G(n,1/2,\omega)} \Brac{ \sum_{S \subseteq [n]} c_S \prod_{i \in S} x_i }\\
    & = \E_{(H,x) \sim G(n,1/2,\omega)}[f_H(x)]\mper\qedhere
  \end{align*}
\end{proof}

\subsection{Concentration Bounds for Linear Constraints}
\label{sec:normalization}
In this section we prove Lemma~\ref{lem:normalization}.
We will use the following elementary concentration bound repeatedly.
(It is the scalar version of the matrix concentration bound Lemma~\ref{lem:single-topology-matrix}; we state and prove a scalar version here because it is a good warmup for Lemma~\ref{lem:single-topology-matrix}.)

\begin{lemma}
\label{lem:sum-parities}
  Let $\cT$ be a family of subsets of $\nchoose{2}$ so that for every $T,T' \in \cT$ there exists $\sigma : [n] \rightarrow [n]$ a permutation of vertices so that $\sigma(T) = T'$.
  Let $t$ be the number of vertices incident to edges in any $T \in \cT$.
  For every $s \geq 0$ and every even $\ell$,
  \[
  \Pr_{G \sim \G(n,1/2)} \left \{ \left | \sum_{T \in \cT} \chi_T(G) \right | \leq s \right \} \geq 1 - \frac{n^{t\ell/2} \cdot (t\ell)^{t\ell}}{s^\ell}\mper
  \]
\end{lemma}
\begin{proof}
  Let $\ell \in \N$ be a parameter to be chosen later.
  We will estimate $\E_{G \sim \G(n,1/2)} [( \sum_{T \in \cT} \chi_T)^\ell]$.
  \begin{align*}
    \E_{G \sim \G(n,1/2)} \Brac{\Paren{ \sum_{T \in \cT} \chi_T}^\ell} & = \sum_{T_1,\ldots,T_{\ell} \in \cT} \E_{G \sim \G(n,1/2)} \prod_{j \leq \ell} \chi_{T_j}\\
    & = | \{ (T_1,\ldots,T_\ell) \, : \, \E \prod_{j \leq \ell} \chi_{T_j} = 1 \} |\mper
  \end{align*}
  In order to have $\E \prod_{j \leq \ell} \chi_{T_j} = 1$, every edge in the multiset $\bigcup_{j \leq \ell} T_j$ must appear at least twice, so every vertex in the multiset $\bigcup_{j \leq \ell} \V(T_j)$ also appears at least twice.
  Thus, this multiset contains at most $t\ell/2$ distinct vertices.
  Since each $T_j \in \cT$, each is uniquely determined by an ordered tuple of $t$ elements of $[n]$.
  Thus, there are at most $n^{t\ell/2} \cdot (t\ell)^{t\ell}$ distinct choices for $(T_1,\ldots,T_\ell)$, so
  \begin{align*}
    \E_{G \sim \G(n,1/2)} \Brac{\Paren{ \sum_{T \in \cT} \chi_T}^\ell}  \leq n^{t\ell/2} \cdot (t\ell)^{t\ell}.
  \end{align*}
  For even $\ell$, by Markov's inequality,
  \begin{align*}
    \Pr \left \{ \left | \sum_{T \in \cT} \chi_T \right | > s \right \} & = \Pr \left \{ \left | \sum_{T \in \cT} \chi_T) \right |^\ell > s^\ell \right \}\\
    & \leq \frac{n^{t\ell/2} \cdot (t\ell)^{t\ell}}{s^\ell}\mper \qedhere
    \end{align*}
\end{proof}

\restatelemma{lem:normalization}
\begin{proof}
  We will prove the statement regarding $\pE[1]$; the bound for $\pE[\sum_{i \in [n]} x_i]$ is almost identical.

  Recall the Fourier expansion
  \[
    \pE[1] - 1 = \sum_{\substack{T \subseteq \nchoose{2}\\ 2 \leq |\V(T)| \leq \tau}} \Paren{\frac \omega n}^{|\V(T)|} \cdot \chi_T\mper
  \]
  Considering each $T \subseteq \nchoose{2}$ as a graph, we partition $\{ T \subseteq \nchoose{2} \, : \, |\V(T)| = t\}$ into $p_t$ families $\{ \cT_i^t \}_{i = 1}^p$ by placing $T$ and $T'$ in the same family iff there exists a permutation $\sigma : [n] \rightarrow [n]$ of vertices so that $\sigma(T) = T'$.
  Thus,
  \[
    \pE[1] - 1 = \sum_{t = 2}^\tau \Paren{\frac \omega n}^{t} \sum_{i = 1}^{p_t} \sum_{T \in \cT_i^t} \chi_T \leq \sum_{t = 2}^\tau \Paren{\frac \omega n}^{t} \sum_{i = 1}^{p_t} \left | \sum_{T \in \cT_i^t} \chi_T \right |\mper
  \]
  By Lemma~\ref{lem:sum-parities} (taking $\ell = (\log n)^2$), and since $t \leq \tau \leq \log n$, each $\cT_i^t$ satisfies
  \[
    \Pr \left \{ \left | \sum_{T \in \cT_i^t} \chi_T \right | < O(n^{t/2} \cdot (\log n)^{3t} )\right \} \geq 1 - (\tau \cdot 2^{t^2} \cdot n^{\log n})^{-1}\mper
  \]
  By a union bound over all $p_t \leq 2^{t^2}$ families $\cT_i^t$, we get that with high probability,
  \[
    |\pE[1] - 1| \leq \tau \cdot \max_{t \leq \tau} \Paren{2^{t^2} \cdot \Paren{\frac \omega {\sqrt n}}^{t} }\mper
  \]
  For $\tau \leq (\epsilon /2) \log n$ and $\omega = n^{1/2 - \epsilon}$, this is at most $n^{-\Omega(\epsilon)}$.
\end{proof}.

\subsection{Combinatorial Proofs about Ribbons}
\label{sec:ribbons-proofs}
In this section we prove Lemma~\ref{lem:left-right-sep}, restated here:
\restatelemma{lem:left-right-sep}

We start by defining a natural partial order on the set of vertex separators in a ribbon $\cR$.
\begin{definition}
We write $Q_1 \leq Q_2$ for two vertex separators $Q_1$ and $Q_2$ of an $(I,J)$-ribbon $\cR$ if $Q_1$ separates $I$ and $Q_2$.
\end{definition}
Next, we check that the definition above indeed is a partial order.
\begin{lemma}
For any set of minimum vertex separators $Q_1, Q_2, Q_3$ an $(I,J)$-ribbon, we have:
\begin{enumerate}
\item $Q_1 \leq Q_1$.
\item If $Q_1 \leq Q_2$ and $Q_2 \leq Q_3$, then, $Q_1 \leq Q_3$. 
\item If $Q_1 \leq Q_2$ and $Q_2 \leq Q_1$, then, $Q_1 = Q_2$. 
\end{enumerate}
\end{lemma}
\begin{proof}
The first statement is immediate from the definition. For the second, consider a path $P$ from $I$ to $Q_3$ in $\cR$. Since $Q_2 \leq Q_3$, $P$ passes through a vertex in $Q_2$. Thus, $P$ contains a subpath that connects $I$ and $Q_2$. But since $Q_1 \leq Q_2$, this subpath must pass through $Q_1$. Thus, any such $P$ must pass through $Q_1$ and thus, $Q_1 \leq Q_3$.

Finally, for the third statement, let $k = |Q_1| = |Q_2|$. Then, using Menger's theorem (Fact \ref{fact:Mengers}, there is a set of $k$ vertex disjoint paths $P_1, P_2, \ldots, P_k$ between $I$ and $J$. By virtue of $Q_1, Q_2$ being \emph{minimum} vertex separators of $\cR$, $Q_1$ and $Q_2$ must intersect each $P_i$ in exactly one vertex. It is then immediate that the only way $Q_1 \leq Q_2$ and $Q_2 \leq Q_1$ if every $P_i$ intersects $Q_1, Q_2$ in the same vertex.
\end{proof}

Now we can prove Lemma~\ref{lem:left-right-sep}.

\begin{proof}[Proof of Lemma~\ref{lem:left-right-sep}]
It is enough to show that for any two minimum separators $Q_1, Q_2$ of size $k$ in $R$, there are separators $Q_L, Q_R$ such that $Q_L \leq Q_1 \leq Q_R$ and $Q_L \leq Q_2 \leq Q_R$. We now construct $Q_L$ and $Q_R$ as required.

Let $U = Q_1 \cap Q_2$ and $V = Q_1 \Delta Q_2$. Let $W_L \subseteq V$ be the set of vertices $w$ such that there is a path from $I$ to $w$ that doesn't pass through $Q_1 \cup Q_2$. Similarly, let $W_R \subseteq V$ be the set of vertices such that there is a path from $w$ to some vertex in $J$ that doesn't pass through any vertex in $Q_1 \cup Q_2$. Then we first observe:

\begin{claim}
$W_L \cap W_R = \emptyset.$
\end{claim}
\begin{proof}[Proof of Claim]
Assume otherwise and let $w \in W_L \cap W_R$. Then there is a path between $I$ and $J$ that doesn't go through any vertex in at least one of $Q_1$ or $Q_2$ contradicting that both are in fact vertex separators.
\end{proof}

Next, we have:
\begin{claim}
Let $Q_L = U \cup W_L$ and $Q_R = U \cup W_R$. Then $Q_L, Q_R$ are both vertex separators in $R$.
\end{claim}
\begin{proof}[Proof of Claim]
We only give the argument for $Q_L$, the other case is similar. Assume there is a path $P$ from $I$ to $J$ that does not pass through $Q_L$. $P$ must intersect $Q_1 \cup Q_2$. Then there is a vertex $v \in Q_1 \cup Q_2$ such that there is a path  $I$ to $v$ which intersects no other vertices in $Q_1 \cup Q_2$. This implies that either $v \in U$ or $v \in W_L$. But by our construction of $W_L$ this is a contradiction.
\end{proof}

Finally, we note that both $Q_L, Q_R$ must in fact be \emph{minimum} vertex separators. 
\begin{claim}
$|Q_L| = |Q_R| = |Q_1| = |Q_2| = k$
\end{claim}
\begin{proof}[Proof of Claim]
Let $|Q_1| = |Q_2| = k$. Then $2k = |Q_1| + |Q_2| = 2|U| + |V| \geq 2|U| + |W_L| + |W_R| = |U \cup W_L| + |U \cup W_R| = |Q_L| + |Q_R|$. Since $Q_L$ and $Q_R$ are vertex separators, $|Q_L|, |Q_R| \geq k$. Thus, $|Q_L| = |Q_R| = k$.
\end{proof}

Finally, we have the ordering requirement on $Q_L$ and $Q_R$.
\begin{claim}
$Q_L \leq Q_1$ and $Q_2 \leq Q_R$.
\end{claim}
\begin{proof}[Proof of Claim]
Let $P$ be a path from $I$ to $Q_1$, let $v$ be the first vertex on this path which is in $Q_1 \cup Q_2$. Then, $v \in U$ or $v \in W_L$. Thus, $Q_L \leq Q_1$. The other case is similar.\qedhere
\end{proof}
This concludes the proof of the lemma.
\end{proof}


\section{Spectral Norms}
The results in this section are in essence due to Medarametla and Potechin \cite{MP16}.
For completeness, we state and prove them here in the language and notation of the current paper, with minor modifications as needed.

\restatelemma{lem:single-topology-matrix}

\begin{proof}[Proof of Lemma~\ref{lem:single-topology-matrix}]
  We proceed by the trace power method, with a dependence-breaking step beforehand.

  \paragraph{Breaking Dependence} Let $q_1,\ldots,q_p$ be vertex-disjoint paths from $A \setminus B$ to $B \setminus A$ in $U$.
  Without loss of generality we can take each to intersect $A \setminus B$ and $B \setminus A$ only at its endpoints.
  We will partition the space of labelings $\sigma$ into disjoint sets $S_1,\ldots,S_m$.
  For each $S_k$ there will be a partition $V_1^k,V_2^k$ of $[n]$ so that $\sigma(\bigcup_{j \leq p} q_j) \subseteq V_1^k$ and $\sigma(U \setminus (\bigcup_{j \leq p} q_j)) \subseteq V_2^k$ for every $\sigma \in S_k$.
  Let $(V_1^1,V_2^1),\ldots,(V_1^m,V_2^m)$ be a sequence of independent uniformly random partitions of $[n]$.
  Call a labeling $\sigma$ \emph{good} at $k$ if the preceeding conditions apply to $\sigma$ for the partition $V_1^k,V_2^k$ and not for any $V_1^{k'}, V_2^{k'}$ for some $k' < k$.
  Let $S_k = \{ \sigma \, : \, \sigma \mbox{ is good at $k$} \}$.

  \begin{claim}
    There is $m = O(2^t \cdot t \cdot \log n)$ so that $\bigcup_{k = 1}^m S_k$ contains every labeling $\sigma : U \rightarrow G$.
  \end{claim}
  \begin{proof}
    For a fixed $\sigma$,
    \[
      \Pr \{\sigma \mbox{ not good for some $k \leq m$} \} \leq (1 - 2^{-t})^m
    \]
    since every vertex $u \in U$ is in $V_i$ with probability $1/2$.
    If $m \geq 10 t 2^t \log n$, then by a union bound over all $\sigma : U \rightarrow G$ (of which there are at most $n^t$), we get $\Pr \{ \mbox{ all $\sigma$ good for some $k \leq m$ } \} > 0$.
  \end{proof}

  Henceforth, let $S_1,\ldots,S_m$ be the partition guaranteed by the preceeding claim.
  For $k \leq m$, let $M_k(I,J) = \sum_{\sigma \in S_k \, : \, \sigma(A) = I, \sigma(B) = J} \val(\sigma)$.
  Then $M = \sum_{k = 1}^m M_k$.

  \paragraph{Moment Calculation} Let $\ell = \ell(n)$ be a parameter to be chosen later.
  By the triangle inequality, $\|M\| \leq \sum_{k =1}^m \|M_k\|$.
  Fix $k$.
  We expand $\E_G \Tr (M_k^\top M_k)^\ell$ as
  \[
    \E \Tr (M_k^\top M_k)^\ell = \E \sum_{\substack{\sigma_1,\ldots,\sigma_{2\ell} \in S_k \\ \sigma_{2i}(A) = \sigma_{2i-1}(A) \\ \sigma_{2i}(B) = \sigma_{2i+1}(B)}} \prod_{j=1}^{2\ell} \val(\sigma_j)\mper
  \]
  (Here arithmetic with indices $i$ is modulo $2\ell$, so for example we take $2i + 1 = 1$.)
  For any $\sigma$,
  \[
    \val(\sigma) = \prod_{(i,j) \in U} \g_{\sigma(i),\sigma(j)}\mper
  \]
  Notice that for all $\sigma_1,\ldots,\sigma_{2\ell}$, the expectation $\E \prod_{j=1}^{2\ell} \val(\sigma_j)$ is either $0$ or $1$.
  We will bound the number of $\sigma_1,\ldots,\sigma_{2\ell}$ for which $\E \prod_{j = 1}^{2\ell} \val(\sigma_j) = 1$ by bounding the number of distinct labels such a family of labelings may assign to vertices in $U$.

  Fix $\sigma_1,\ldots,\sigma_{2\ell} \in S_k$.
  Consider the family $q_1,\ldots,q_p$ of vertex-disjoint paths.
  Every edge in every $q_j$ receives one pair of labels from each $\sigma_i$.
  Consider these labels arranged on $2\ell$ adjoined copies of each $q_j$, one for each $\sigma$ (giving $p$ paths with $2\ell \sum_{j \leq p} |q_j|$ edges in total, where $|q_j|$ is the number of edges in $q_j$).
  Every pair of labels $\{\sigma_i(v), \sigma_i(w) \}$ appearing on an edge $(v,w)$ in this graph must also appear on some distinct edge $(v',w')$ in order to have $\E \prod_{i =1}^{2\ell} \val(\sigma_i) =1$; otherwise the disjointness of $V_1^k, V_2^k$ would be violated.
  Merging edges which received the same pair of labels, we arrive at a graph with at most $p$ connected components and at most $\ell \sum_{j \leq p} |q_j|$ edges, and so at most $\ell \sum_{j \leq p} |q_j| + p$ vertices.
  Thus, the vertices in $q_1,\ldots,q_p$ together receive at most $\ell \sum_{j \leq p} |q_j| + p$ distinct labels among all $\sigma_1,\ldots,\sigma_{2\ell}$.

  Next we account for labels of $v \notin (\bigcup_{j \leq p} q_j \cup A \cup B)$.
  If $\E_G \prod_{i = 1}^{2\ell} \val(\sigma_i) = 1$ then the $2\ell$-size multiset $\{ \sigma_i(v) \}_{i \leq 2\ell}$ of labels for such $v$ contains at most $\ell$ distinct labels, since by assumption $v$ has degree at least $1$ in $U$.

  Next we account for labels of vertices in $A \setminus (B \cup \bigcup_{j \leq p} q_j)$ and $B \setminus (A \cup \bigcup_{j \leq p} q_j)$.
  Every such vertex receives a label from every $\sigma_i$, but $\sigma_{2i}$ and $\sigma_{2i -1}$ must agree on $A$-labels and $\sigma_{2i}$ and $\sigma_{2i + 1}$ must agree on $B$-labels.
  So in total there are at most $\ell (|A| + |B| - 2p - 2r)$ distinct labels for such vertices.

  This means that among the labels $\sigma_i(j)$ for all $j \notin A \cap B$, there are at most
  \[
    \underbrace{\ell \sum_{j \leq p} |q_j| + p}_{\mbox{labels from paths}} + \underbrace{\ell (|A| + |B| - 2p - 2r)}_{\mbox{additional vertices in } A \cup B \setminus (A \cap B)} + \underbrace{\ell(t - (|A| + |B| -r) - (\sum_j |q_j| -p))}_{\mbox{vertices in } U \setminus (\bigcup_j q_j \cup A \cup B)} = \ell(t - p -r) + p
  \]
  unique labels.

  Finally, consider the labels of the $r$ vertices $j_1,\ldots, j_r$ in $A \cap B$.
  The first labelling $\sigma_1$ assigns these vertices some $\sigma_1(j_1),\ldots,\sigma_1(j_r)$ labels in $G$.
  Since $\sigma_2$ agrees with $\sigma_1$ on $A$-vertices, we must have $\sigma_2(j_1) = \sigma_1(j_1),\ldots,\sigma_1(j_r) = \sigma_2(j_r)$.
  Since $\sigma_3$ agrees with $\sigma_2$ on $B$-vertices, we must have $\sigma_3(j_1) = \sigma_2(j_1),\ldots,\sigma_3(j_r) = \sigma_2(j_r)$, and so on.
  So there are at most $r$ unique labels for such vertices.

  Now we can assess how many choices there are for $\sigma_1,\ldots,\sigma_{2\ell} \in S_k$ so that $\E \prod_{i \leq 2\ell} \val(\sigma_i) = 1$.
  To choose such a collection $\sigma_1,\ldots,\sigma_{2\ell}$, we proceed in stages.
  \begin{enumerate}[\textbf{{Stage}} 1.]
    \item Choose the labels $\sigma_i(j_1),\ldots,\sigma_i(j_r)$ of all the vertices in $A \cap B$.
      Here there are at most $n^r$ options.
    \item For each pair $(i,j)$, where $j \notin A \cap B$, choose whether $\sigma_i(j)$ it will be the first appearance of the index $\sigma_i(j) \in [n]$ or if there is some $i' < i$ and $j'$ so that $\sigma_{i'}(j') = \sigma_i(j)$.
      Here there are $2^{2\ell t}$ options.
    \item Choose the labels $\sigma_i(j) \in [n]$ for all $j \notin A \cap B$ and pairs $(i,j)$ which in Stage 2 we chose to be the first appearance of a label.
      If there are $x$ such vertices, there are at most $n^x$ options.
    \item Choose the labels $\sigma_i(j) \in [n]$ for all the pairs $(i,j)$, with $j \notin A \cap B$, which in Stage 2 we chose not to be the first apperance of a label.
      Here there are at most $x^{2\ell t -2\ell r - x}$ options.
  \end{enumerate}

  All together, there are at most $n^r \cdot 2^{2\ell t} \cdot n^x \cdot x^{2\ell (t-r) - x} \leq n^r \cdot 2^{2\ell t} \cdot n^x \cdot (2\ell{t})^{2\ell (t-r) - x}$ choices for a given $x$.
  Since $4lt \ll n$, summing up over all $x \leq \ell(t - p - r) + p$ the total number of choices is at most $2n^r \cdot 2^{2\ell t} \cdot n^{\ell(t - p - r) + p} \cdot (2 \ell t)^{\ell(t -r + p) - p}$.
  Putting it together,
  \[
    \E \Tr(M_k^\top M_k)^\ell \leq 2n^r \cdot n^{\ell(t - p - r) + p} \cdot (2 \ell t)^{\ell(t -r + p) - p}\mper
  \]


  Now using Markov's inequality and standard manipulations, for any $s$,
  \begin{align*}
    \Pr\{\|M_k\| \geq s\} & = \Pr\{ \|M_k^\top M_k\|^\ell \geq s^{2\ell} \} \\
    & \leq \frac{\E \|(M_k^\top M_k)^{\ell}\|}{s^{2\ell}} \quad \text{by Markov's}\\
    & \leq \frac{\E \Tr (M_k^\top M_k)^{\ell}}{s^{2\ell}} \quad \text{since $\|(M_k^\top M_k)^\ell\| \leq \Tr (M_k^\top M_k)^\ell$}\\
    & \leq \frac{2n^r \cdot 2^{2\ell t} \cdot n^{\ell(t - p - r) + p} \cdot (2 \ell t)^{\ell(t - r + p) - p}}{s^{2\ell}}
  \end{align*}
  Taking $\ell = (\log n)^3$ and using $p \leq t \leq O(\log n)$, there is $s = 2^t \cdot n^{(t - p - r)/2} (\log n)^{O(t -r + p)}$ so that $\Pr\{\|M_k\| \geq s\} \leq n^{-100\log n} m^{-1}$.
  By a union bound, $\Pr \{\|M_k\| \leq s \mbox{ for all $k$} \} \geq 1 - n^{-100 \log n}$, so $\|M\| \leq sm$ with probability $1 - n^{-100 \log n}$.
  Since $m \leq 2^{O(t)} \cdot \log(n)^{O(1)}$, this completes the proof.
\end{proof}

\end{document}